\documentclass{lmcs}
\pdfoutput=1

\usepackage{xparse}

\usepackage{amsmath}
\usepackage{mathtools}
\usepackage[hypertexnames=false]{hyperref}
\usepackage{todonotes}
\NewDocumentCommand{\ntodo}{O{}mm}{\todo[#1]{#2: #3}}
\NewDocumentCommand{\lionel} {O{}m}{\ntodo[color=violet!50!white,#1]{Li}{#2}}
\NewDocumentCommand{\lorenzo}{O{}m}{\ntodo[color=green,#1]{Lo}{#2}}
\NewDocumentCommand{\olivier}{O{}m}{\ntodo[color=cyan,#1]{O}{#2}}
\NewDocumentCommand{\remi}   {O{}m}{\ntodo[color=red,#1]{R}{#2}}
\usepackage{enumitem} 
\usepackage{verbatim}

\usepackage[utf8]{inputenc}
\usepackage[T1]{fontenc}
\usepackage{amsfonts}
\usepackage{amssymb}
\usepackage{cmll}
\usepackage{stmaryrd}
\usepackage{alphabeta}
\usepackage{cancel}

\usepackage{amsthm}
\usepackage{thmtools}
\usepackage[capitalise]{cleveref} 

\crefname{thm}{Theorem}{Theorems}
\crefname{thmC}{Theorem}{Theorems}
\crefname{lem}{Lemma}{Lemmas}
\crefname{lemC}{Lemma}{Lemmas}
\crefname{cor}{Corollary}{Corollaries}
\crefname{corC}{Corollary}{Corollaries}
\crefname{prop}{Proposition}{Propositions}
\crefname{propC}{Proposition}{Propositions}
\crefname{defi}{Definition}{Definitions}
\crefname{fact}{Fact}{Facts}
\crefname{exa}{Example}{Examples}
\crefname{rem}{Remark}{Remarks}
\crefname{figure}{Figure}{Figures}

\usepackage{ebproof}
\ebproofset{
  right label template={$\scriptstyle(\inserttext)$},
  center=false,
}
\newcommand{\subproof}[2]{\hypo{#1}\infer[no rule]1{#2}}

\usepackage{adjustbox}
\adjustboxset{max width = \linewidth, center = \linewidth}

\newcommand{\eqdef}{\coloneqq}

\newcommand{\definitive}[1]{\textnormal{\textbf{#1}}}
\newcommand{\suchthat}{\;|\;}
\newcommand*{\ie}{\emph{i.e.}}
\newcommand*{\eg}{\emph{e.g.}}
\newcommand*{\resp}{resp.}
\newcommand*{\aka}{\emph{a.k.a.}}

\newcommand*{\wolog}{\emph{w.l.o.g.}}

\NewDocumentCommand{\pars}{ O{} m }{#1(#2#1)}
\NewDocumentCommand{\vpars}{ m m m }
	{\IfBooleanTF{#1}{#2(#3#2)}{#3}}

\NewDocumentCommand{\braces}{ O{} m }{#1\{#2#1\}}

\NewDocumentCommand{\brackets}{ O{} m }{#1[#2#1]}

\NewDocumentCommand{\angles}{ O{} m }{#1\langle#2#1\rangle}

\NewDocumentCommand{\bars}{ O{} m }{#1|#2#1|}

\NewDocumentCommand{\setOf}{O{}m}{\braces[#1]{#2}}
\NewDocumentCommand{\card}{sO{}m}{\#\vpars{#1}{#2}{#3}}

\newcommand{\Mf}{\mathfrak{M}_f} 
\NewDocumentCommand{\MfOf}{O{}m}{\Mf\pars[#1]{#2}}

\NewDocumentCommand{\length}{O{}m}{\bars[#1]{#2}}

\NewDocumentCommand{\tupleOf}{O{}m}{\pars[#1]{#2}}

\NewDocumentCommand{\msetOf}{O{}m}{\brackets[#1]{#2}}

\newcommand*{\logsys}[1]{\textnormal{\textsf{#1}}}
\newcommand*{\hyp}{\textnormal{\textit{hyp}}}
\newcommand*{\mix}{\textnormal{\textit{mix}}}
\newcommand*{\MLL}{\ensuremath{\logsys{MLL}}}
\newcommand*{\MLLhmix}{\ensuremath{\MLL^{0,2}_{\hyp}}}

\newcommand*{\MALLhmix}{\ensuremath{\logsys{MALL}^{0,2}_{\hyp}}}

\newcommand*{\tensor}{\otimes}
\newcommand*{\cut}{\textnormal{\textit{cut}}}
\newcommand*{\ax}{\textnormal{\textit{ax}}}
\newcommand*{\orth}{^\perp}
\newcommand*{\pw}{\parr\backslash\with} 

\newcommand*{\someproof}{\pi}
\newcommand*{\ixretore}{\text{ix-R\'etor\'e}}
\newcommand*{\mixretore}{\text{m}\ixretore}
\newcommand*{\Mixretore}{\text{M}\ixretore}

\newcommand*{\somepn}{\theta} 
\newcommand*{\G}{\mathcal{G}} 
\newcommand*{\somelinking}{\lambda} 
\newcommand*{\otherlinking}{\mu}
\newcommand*{\otherotherlinking}{\nu}
\newcommand*{\somesetlinking}{\Lambda} 
\newcommand*{\somelink}{a} 

\newcommand*{\someres}{R} 
\newcommand*{\otherres}{S}
\newcommand*{\otherotherres}{T}

\newlist{enumeratecref}{enumerate}{1}
\crefname{enumeratecrefi}{}{}

\newcommand*{\vertexset}{\mathcal{V}}

\newcommand*{\edgeset}{\mathcal{E}}
\newcommand*{\source}{\textnormal{\textsf{s}}}
\newcommand*{\target}{\textnormal{\textsf{t}}}
\newcommand*{\incidence}{\psi}
\newcommand*{\somegraph}{G}
\newcommand*{\colorgraph}{H}
\newcommand*{\someindepset}{S}
\newcommand*{\somevertex}{v}
\newcommand*{\othervertex}{u}
\newcommand*{\otherothervertex}{x}
\newcommand*{\otherotherothervertex}{y}
\newcommand*{\otherotherotherothervertex}{z}
\newcommand*{\somewith}{w}

\newcommand*{\someleaf}{l}
\newcommand*{\otherleaf}{r}

\newcommand*{\someedge}{e}
\newcommand*{\otheredge}{f}
\newcommand*{\otherotheredge}{g}
\newcommand*{\otherotherotheredge}{h}
\newcommand*{\somesetedge}{P}
\newcommand*{\someline}{e}
\newcommand*{\otherline}{f}
\newcommand*{\somepath}{p}
\newcommand*{\otherpath}{q}
\newcommand*{\otherotherpath}{\rho}
\newcommand*{\otherotherotherpath}{\chi}

\newcommand*{\somecycle}{\omega}
\newcommand*{\othercycle}{\sigma}
\newcommand*{\somesetcycle}{\Omega}

\newcommand*{\someclosedyeo}{\somecycle'}
\newcommand*{\otherclosed}{d}

\newcommand*{\somepier}{\kappa}
\newcommand*{\somesubgraph}{S}
\newcommand*{\othersubgraph}{R}
\newcommand*{\giso}{\simeq}
\newcommand*{\isbridge}[2]{b^{#1}_{#2}}

\newcommand*{\concatpath}[2]{\ensuremath{#1\cdot #2}}
\newcommand*{\reversedpath}[1]{\ensuremath{\overline{#1}}}
\newcommand*{\subpath}[3]{\ensuremath{{#1}_{(#2, #3)}}}
\newcommand*{\mincycles}[1]{\mathcal{M}_{#1}}

\newcommand*{\necubfcycles}{\mathfrak{O}}

\NewDocumentCommand{\orderparr}{o}{\IfNoValueTF{#1}{\prec}{\overset{#1}{\prec}}}
\NewDocumentCommand{\ssbfpath}{o}{\IfNoValueTF{#1}{\curvearrowright}{\overset{#1}{\curvearrowright}}}
\NewDocumentCommand{\orderseq}{o}{\IfNoValueTF{#1}{\prec}{\overset{#1}{\prec}}}

\newcommand{\sobfnbpathsymbol}{\Rsh}
\NewDocumentCommand{\sobfnbpath}{o}{\IfNoValueTF{#1}{\sobfnbpathsymbol}{\overset{#1}{\sobfnbpathsymbol}}}
\NewDocumentCommand{\orderyeo}{o}{\IfNoValueTF{#1}{\lhd}{\overset{#1}{\lhd}}}
\newcommand*{\king}[1]{\textit{k}(#1)}

\newcommand*{\colors}{\textnormal{\textsf{C}}}
\newcommand*{\coloring}{\textnormal{\textsf{c}}}

\newcommand*{\somecolor}{\alpha}
\newcommand*{\othercolor}{\beta}
\newcommand*{\otherothercolor}{\tau}
\newcommand*{\otherotherothercolor}{\epsilon}
\newcommand*{\colenc}[1]{\overline{#1}}

\newcommand*{\someps}{\ensuremath{\rho}}
\newcommand*{\deseq}{\mathcal{D}}
\newcommand*{\cdeg}{d}

\usetikzlibrary{fit, patterns, shapes.geometric}
\NewDocumentEnvironment{genscope}{O{black}O{solid,->}O{0.8pt}}{
\begin{scope}[
every node/.style={circle,minimum size=8mm,inner sep=0,draw=#1,line width=#3},
every edge/.style={draw=#1,#2,line width=#3},
every path/.style={thick,draw=#1,#2,-,line width=#3}]
\color{#1}
}{
\end{scope}
}
\newenvironment{myscopehigh}[1]{
\begin{genscope}[#1][dashed,-][2.5pt]
}{
\end{genscope}
}
\tikzset{
    arete nommee/.style={draw=none,inner sep=1pt,outer sep=2pt,rectangle,minimum size=0pt,auto},
    chemin nomme/.style={arete nommee,sloped},
    subps/.style={draw,red,dashed,shape=ellipse,minimum height=1cm},
    colored ps/.style={
        every node/.style = { circle, thick, draw, inner sep=1pt, minimum size=6mm },
        every edge/.style = { thick, draw, - },
    }
}

\newcommand*{\alternating}{$\sobfnbpathsymbol$}
\newcommand*{\edgeofsetcycles}{\mathfrak{\someedge}}
\newcommand*{\sourceofsetcycles}{\mathfrak{\someleaf}}
\newcommand*{\targetofsetcycles}{\mathfrak{\somewith}}

\makeatletter
\newcommand{\setword}[2]{\phantomsection#1\def\@currentlabel{\unexpanded{#1}}\label{#2}}
\makeatother

\keywords{Linear Logic, Proof Net, Sequentialization, Graph Theory, Yeo’s Theorem}

\begin{document}

\title[Yeo's Theorem for Locally Colored Graphs: the Path to Sequentialization]{Yeo's Theorem for Locally Colored Graphs:\texorpdfstring{\\}{ }the Path to Sequentialization in Linear Logic}

\titlecomment{{\lsuper*}
The present paper is a revised and extended version of~\cite{DiGuardiaLaurentdeFalcoVauxAuclair25}, expanding the result to the additive proof nets and with a different organization.}

\author[R.~{Di~Guardia}]{R\'emi {Di~Guardia}\lmcsorcid{0009-0004-8632-108X}}[a]
\author[O.~Laurent]{Olivier Laurent\lmcsorcid{0009-0007-1306-8994}}[b]
\author[L.~{Tortora~de~Falco}]{Lorenzo {Tortora de Falco}\lmcsorcid{0000-0002-3987-1095}}[c]
\author[L.~{Vaux~Auclair}]{Lionel {Vaux Auclair}\lmcsorcid{0000-0001-9466-418X}}[d]

\address{Université Paris Cité, Inria, CNRS, IRIF, F-75013, Paris, France}
\email{remi.di.guardia@ens-lyon.org}

\address{CNRS, ENS de Lyon, Université Claude Bernard Lyon~1, LIP, UMR~5668, Lyon, France}
\email{olivier.laurent@ens-lyon.fr}

\address{Università Roma Tre, Dipartimento di Matematica e Fisica, Rome, Italy \& GNSAGA, Istituto Nazionale di Alta Matematica, Rome, Italy}
\email{tortora@uniroma3.it}

\address{Aix Marseille Univ, CNRS, I2M, Marseille, France}
\email{lionel.vaux@univ-amu.fr}

\begin{abstract}
\noindent
We revisit sequentialization proofs associated with the Danos-Regnier correctness criterion in the theory of proof nets of linear logic.
Our approach relies on a generalization of Yeo's theorem for graphs, based on colorings of half-edges.
This happens to be the appropriate level of abstraction to extract sequentiality information from a proof net
without modifying its graph structure.
We thus obtain different ways of recovering a sequent calculus derivation from a proof net inductively,
by relying on a \emph{splitting} vertex,
which we can impose to be a $\parr$-vertex, or a terminal vertex, or a non-axiom vertex, etc.,
in a modular way.
This approach applies in presence of the $\mix$-rules as well as for proof nets of unit-free multiplicative-additive linear logic
(through an appropriate further generalization of Yeo's theorem).

The proof of our Yeo-style theorem relies on a key lemma that we call \emph{cusp minimization}.
Given a coloring of half-edges, a cusp in a path is a vertex whose adjacent half-edges in the path have the same color.
And, given a cycle with at least one cusp and subject to suitable hypotheses,
cusp minimization constructs a cycle with strictly less cusps.
In the absence of cusp-free cycles, cusp minimization is then enough to ensure the existence of a splitting vertex, \ie\ a vertex that is a cusp of any cycle it belongs to.
Our theorem subsumes several graph-theoretical results, including some known to be equivalent to Yeo's theorem.
The novelty is that they can be derived in a straightforward way, just by defining a dedicated coloring,
again without any modification of the underlying graph structure (vertices and edges)
-- similar results from the literature required more involved encodings.
\end{abstract}

\maketitle

\section{Introduction}

Proof nets are a major contribution from linear logic~\cite{ll}.
Contrary to the usual representation of proofs as derivation trees in sequent calculus, proof nets represent proofs as general graphs respecting some \emph{correctness criterion}~\cite{structmult},
which imposes the absence of a particular kind of cycle.
Proof nets identify the derivations of sequent calculus up to rule permutations~\cite{mallpncom} and,
as a consequence of this canonicity, results like cut elimination become
easier to prove in this formalism.
A key theorem in this approach is the fact that each proof net is indeed the
graph representation of a derivation in the sequent calculus:
the process of recovering such a derivation tree is called \emph{sequentialization}.
Many proofs of this result can be found in the literature~\cite[etc.]{ll,structmult,quantif2,pn,curienludics2}, but proving sequentialization is still considered as not easy.

Not only many proofs but more generally many equivalent correctness criteria have been introduced in the last 40 years,
based on the existence or absence of particular paths in an associated graph
(long trips, switching cycles, alternating-elementary-cycles)~\cite{ll,structmult,rbpn},
on the success of a rewriting procedure (contractibility, parsing)~\cite{phddanos,pnin,cclin,NauroisMogbil11},
on homological~\cite{homolpn} or topological~\cite{Mellies2004} properties, etc.
They all describe the same set of valid graphs (those which are the image of a sequent calculus derivation) but through very different statements of properties characterizing the appropriate structure.
The diversity of these approaches reflects both the central nature of the concept of proof net in linear logic,
and the variety of motivations in the design of correctness criteria:
some ensure tight complexity bounds (especially those based on contractibility),
some weave connexions with other fields (\eg, topology or graph theory),
some are more naturally generalized to other logical systems,
etc.

On the other hand, when it comes to the study of the theory of proof nets
(confluence, normalization, reduction strategies, etc.) most of those
approaches are hardly usable in practice.
This gives the Danos-Regnier criterion~\cite{structmult} a special status:
the absence of switching cycles is of direct use for proving results about proof nets.
For instance, it forbids the occurrence of axiom-cut cycles along cut elimination~\cite{synsem};
it ensures the confluence of reduction in multiplicative-exponential linear logic~\cite{snll};
it provides the existence of so-called closed cuts~\cite{Laurent20},
which play a crucial rôle in geometry of interaction~\cite{goi1};
it allows for the definition of a parallel procedure of cut elimination
for multiplicative \cite{DBLP:journals/lmcs/ChouquetA21}
or even multiplicative-exponential linear logic~\cite{DBLP:conf/lpar/GuerrieriMFA24};
etc.
This means in particular that, based on this criterion, it becomes possible to develop the theory of proof nets without referring to the sequent calculus anymore.
For this reason, we are interested in a better understanding of this precise criterion and its links with the sequential structure of tree derivations, \emph{via} sequentialization.
Following previous lines of work on relating graph theory and proof net theory~\cite{rbpn,pnehrhard,pngraph}, we looked for a direct link between graph properties and the sequential structure of proof nets: \emph{splitting vertices}.
Indeed, the key step for extracting a sequent calculus derivation from a proof net is to inductively decompose it into sub-graphs themselves satisfying the correctness condition.

In graph theory, it is common to have several (equivalent) characterizations for a same class of graphs, and an inductive characterization may allow for simpler proofs -- see \eg, cographs~\cite{CORNEIL1981163}, $k$-trees~\cite{Beineke_Pippert_1971} or graphs with a unique perfect matching~\cite{Kotzig1959}.
Such an inductive characterization may be deduced from the existence of a vertex or of an edge separating the graph in a ``nice'' manner (\eg\ a bridge~\cite{bangdigraphs}).
Five theorems yielding such a vertex or edge have been shown equivalent by Szeider~\cite{Szeider04c}, meaning they can be deduced from each other using an encoding of the graph under consideration.
Among those five are Yeo's theorem on colored graphs~\cite{yeotheorem},
Kotzig's theorem on unique perfect matchings~\cite{Kotzig1959},
but also Shoesmith and Smiley's theorem on turning vertices~\cite{thseqsemicycle}
-- interestingly the approach of the latter bears striking resemblance with our own work,
that we discuss more in detail in the paper (see \cref{sec:graph_comparison:variants}).

On the proof net side,
Rétoré remarked that perfect matchings provide an alternative presentation of proof nets~\cite{rbpn}:
in this context, he recovered sequentialization proofs based on different notions of splitting vertex, in the spirit of Kotzig's theorem~\cite{Kotzig1959}.
Remarkably, Nguy\~{\^e}n later established that Kotzig's theorem is in fact equivalent to the sequentialization theorem of unit-free multiplicative proof nets with $\mix$~\cite{pngraph}, again through graph encodings.

In the present paper, we focus on Yeo's theorem~\cite{yeotheorem} instead, which is about \emph{edge-colored} undirected graphs.
Our goal is to obtain the existence of splitting vertices in proof nets by a direct application of a Yeo-style statement to an edge-coloring of the proof net (with no modification of the graph structure at all, \ie\ keeping the same vertices and edges).
In an edge-colored graph,
a cycle is \emph{alternating} when all its consecutive edges have different colors.
Yeo's theorem states that an edge-colored graph $\somegraph$ with no alternating cycle has a \emph{splitting} vertex $\somevertex$,
\ie\ such that no connected component of $\somegraph - \somevertex$ (the removal of $\somevertex$) is joined to $\somevertex$ with edges of more than one color -- see \cref{fig:ex_Yeo} (a splitting vertex is also called in the literature a cut-color vertex, or a cut vertex separating colors).
This decomposition can be carried on, so as to give an inductive representation of graphs with no alternating cycle.
This important structural result on edge-colored graphs
has been used extensively in the literature (see \eg\ the book~\cite{bangdigraphs} or papers such as~\cite{pathsandtrails,Fujita2018}).

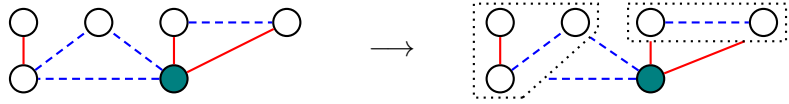
\begin{figure}
\centering
\begin{tikzpicture}[baseline=(base)]
	\coordinate (base) at (0,.3);
\begin{scope}[every node/.style={circle,thick,draw}]
	\node[fill=teal] (0) at (0,0) {};
	\node (1) at (-2,.75) {};
	\node (2) at (0,.75) {};
	\node (3) at (1.5,.75) {};
	\node (4) at (-2,0) {};
	\node (5) at (-1,.75) {};
\end{scope}
\begin{genscope}[red][-]
	\path (0) edge (2);
	\path (0) edge (3);
	\path (1) edge (4);
\end{genscope}
\begin{genscope}[blue][-,densely dashed]
	\path (0) edge (4);
	\path (0) edge (5);
	\path (2) edge (3);
	\path (4) edge (5);
\end{genscope}
\end{tikzpicture}
\qquad$\longrightarrow$\qquad
\begin{tikzpicture}[baseline=(base)]
	\coordinate (base) at (0,.3);
\begin{scope}[every node/.style={circle,thick,draw}]
	\node[fill=teal] (0) at (0,0) {};
	\node (1) at (-2,.75) {};
	\node (2) at (0,.75) {};
	\node (3) at (1.5,.75) {};
	\node (4) at (-2,0) {};
	\node (5) at (-1,.75) {};
	\coordinate (02) at (0,.5);
	\coordinate (03) at (1.25,.5);
	\coordinate (04) at (-1.4,0);
	\coordinate (05) at (-.8,.55);
\end{scope}
\begin{genscope}[red][-]
	\path (0) edge (02);
	\path (0) edge (03);
	\path (1) edge (4);
\end{genscope}
\begin{genscope}[blue][-,densely dashed]
	\path (0) edge (04);
	\path (0) edge (05);
	\path (2) edge (3);
	\path (4) edge (5);
\end{genscope}
\begin{genscope}[black][-,dotted,thin]
	\draw (-1.7,-.25) -- (-2.35,-.25) -- (-2.35,1) -- (-.7,1) -- (-.7,.65) -- cycle;
	\draw (-.3,1) -- (1.8,1) -- (1.8,.5) -- (-.3,.5) -- cycle;
\end{genscope}
\end{tikzpicture}
\caption{Example of Yeo's theorem with a \textcolor{teal}{filled} splitting vertex and dotted connected components}
\label{fig:ex_Yeo}
\end{figure}

To allow for a direct application to proof nets, we generalize Yeo's theorem in two directions.
First, we consider a more general notion of edge-coloring, that we call \emph{local coloring}:
it associates a color with each endpoint of each edge
(this is equivalent to coloring half-edges, but we avoid to introduce half-edges formally, just to stick to more basic graph-theoretic notions).
Second, we introduce a parameter (a set of \emph{vertex-color pairs}, \ie\ a set of vertices labeled with colors) which gives us finer control over the obtained splitting vertex.

Our proof of this new result is elementary and based on a key lemma we call \emph{cusp minimization}, as well as on the definition of an ordering on vertex-color pairs induced by local coloring.
Formally, a cusp in a path of a locally colored graph is a pair of two successive edges,
such that the color associated with the middle vertex is the same for both edges.
The ordering on vertex-color pairs is induced by particular cusp-free paths.
Moreover, given a cycle \(\somecycle\) containing a cusp,
and a non-cusp vertex \(\somevertex\) of \(\somecycle\),
satisfying some additional technical conditions,
our cusp minimization result (\cref{lem:bjumping_Yeo_base}) yields either a cusp-free cycle,
or another cycle with strictly less cusps than \(\somecycle\), but also having \(\somevertex\) as a non-cusp vertex.
In a locally colored graph without cusp-free cycle,
our generalization of Yeo's theorem then follows easily
by considering a maximal vertex-color pair among those in the parameter.

Cusp minimization also provides a proof of the original version of Yeo's theorem,
as simple as known short proofs from the literature~\cite{revisityeo,shortkotzig}.
While the generalization to local colorings gives a statement that we prove equivalent to Yeo's theorem, it seems difficult to reduce the parametrized version to the non-parametrized one. 
We moreover show how the local and parametrized generalization of Yeo's theorem allows to deduce each of the statements considered in~\cite{Szeider04c} (as well as~\cite[Theorem~2]{hcycles}),
simply by choosing appropriate colorings,
without modifying the sets of vertices and edges of the graph under consideration.
Cusp minimization could thus be easily transposed as a proof of any of these results,
just by modifying the definition of a \emph{cusp} -- where is our sole use of the coloring.

Back to linear logic and the theory of proof nets, it is possible to derive
the existence of a splitting vertex (in the sense of sequentialization) from the generalization of Yeo's theorem,
and we are even able to modularly focus on a particular kind of splitting vertex:
an arbitrary splitting vertex,
a splitting multiplicative vertex (\(\parr\) or \(\tensor\)),
a splitting $\parr$ (\aka\ section~\cite{phddanos}),
a terminal splitting multiplicative vertex, etc.
From any of these choices, a sequentialization procedure is easy to deduce.
Notably, this proof of the sequentialization theorem applies directly in the presence of the $\mix$ rules, and the $\mix$-free case can be easily deduced.
This means our variation of Yeo's theorem is a well-stated formulation for the five equivalent theorems from~\cite{Szeider04c}, as well as for~\cite[Theorem~2]{hcycles} and for the sequentialization theorem.

Our new graph-theoretical result can even be further generalized to accommodate the presence of some alternating cycles, with a more technical proof that still reposes on the cusp minimization lemma.
While the conditions on allowed cycles are not standard from a graph-theoretical point of view, this generalization gives a splitting vertex for the unit-free multiplicative-additive proof nets from Hughes and van Glabbeek~\cite{mallpnlong}, and thus a proof of sequentialization in this much more involved context.
Our approach is robust enough to also enable sequentialization through terminal vertices, as opposed to what is done in~\cite{mallpnlong}.
Remarkably, the connection between graph theory and proof net theory, previously restricted to multiplicative proof nets, is thus extended to the additive connectives.

Putting everything together, we get a direct and simple proof of sequentialization for the Danos-Regnier criterion,
assuming no prerequisite in graph theory.
The path to sequentialization in linear logic that we propose starts from cusp minimization then goes to the generalization of Yeo's theorem and concludes with the extraction of an inductive decomposition of proof nets.

\paragraph*{Outline.}
This paper is organized into three main parts.
\begin{enumerate}
\item
First, a purely graph-theoretical part about our generalization of Yeo's theorem.
We start by recalling usual notions -- graphs, paths, etc. -- and with our definition of local coloring (\cref{sec:graph_cusp}).
Then, we state and prove our generalization of Yeo's theorem, \cref{th:ParamLocalYeo}, through the cusp minimization lemma (\cref{sec:genyeo}).
We moreover show that the parameter-free version of our Yeo-style result is equivalent to the original one,
and we expose how to also derive from it the four other equivalent theorems from~\cite{Szeider04c},
as well as a generalization of Yeo's theorem to $\colorgraph$-colored graphs~\cite[Theorem~2]{hcycles} (\cref{sec:graph_comparison})
-- each one is obtained in a straightforward way, just by defining an appropriate coloring.
\item
Next comes a part about logic, with a definition of unit-free multiplicative linear logic with the $\mix$ rules and the associated notion of proof net (\cref{sec:mllpn}).
We then give various proofs of the sequentialization theorem for these proof nets, leveraging our generalization of Yeo's theorem (\cref{sec:seq}).
In the passing, we consider more closely the notion of connectedness in proof nets: first, we introduce the \emph{almost connected} proof nets -- these are the proof nets characterized by the existence of a switching path between the premises of each $\parr$-vertex -- and show that this condition is equivalent to being a disjoint union of connected proof nets -- equivalently, a proof net is almost connected if and only if it is the desequentialization of a proof whose $\mix$ rules are all at the root of the tree --; second, we compare our ordering with the kingdom ordering~\cite{kingemp}, the standard ordering in the literature of proof nets without the $mix$ rules (\cref{sec:more_connectedness}).
\item
Afterwards, we extend our technique in the presence of the additive connectives, mimicking the previous two parts.
We have, again, a purely graph-theoretical part with a further generalization of Yeo's theorem in the presence of some alternating cycles (\cref{sec:yeo_mall}).
It is followed by another logical part, composed first of a definition of unit-free multiplicative-additive linear logic with the $\mix$ rules and an associated notion of proof net (\cref{section:additives}), and then of various proofs of the sequentialization theorem for these proof nets, that are obtained by leveraging our last generalization of Yeo's theorem (\cref{sec:seqmall}).
\end{enumerate}


\section{Graphs and Cusps}
\label{sec:graph_cusp}

\subsection{Partial Undirected Graphs and Paths}
\label{sec:graph_def}

As we take interest in proof nets and Yeo's theorem in this paper, we study undirected paths in finite undirected partial multigraphs.
We recall here quickly some basic notions from graph theory, for more details we refer the reader to~\cite{bondymurty2}.

A (finite undirected multi) \definitive{partial graph} (without loop) is a triple $(\vertexset, \edgeset, \incidence)$ where $\vertexset$ (\definitive{vertices}) and $\edgeset$ (\definitive{edges}) are disjoint finite sets and $\incidence$ (\definitive{the incidence function}) associates to each edge a set of at most two vertices.
An edge $\someedge$ is \definitive{total} when $\incidence(\someedge)$ is of cardinal two, and a \definitive{total graph} (or simply a \definitive{graph}) is one whose edges are total.
Many notions lift immediately from total graphs to partial graphs, \eg\ isomorphisms that we denote by $\giso$.
An edge $\someedge$ is \definitive{incident} to a vertex $\somevertex$ if $\somevertex \in \incidence(\someedge)$, in which case $\somevertex$ is an \definitive{endpoint} of $\someedge$.

A \definitive{path} $\somepath$ is a non-empty finite alternating sequence of vertices and edges of the shape $(\somevertex_0, \someedge_1, \somevertex_1, \someedge_2, \somevertex_2, \dots, \someedge_n, \somevertex_n)$ such that for all $i \in \{1, \dots, n\}$, the endpoints of $\someedge_i$ are exactly $\somevertex_{i-1}$ and $\somevertex_i$ (which are distinct).
A path always has at least one vertex, but it can have no edge and be reduced to a single vertex $(\somevertex_0)$, in which case it is called an \definitive{empty} path.
With the notation above, $\somevertex_0$ is the \definitive{source} of $\somepath$, $\somevertex_n$ is its \definitive{target} and both make the \definitive{endpoints} of $\somepath$.
By the \definitive{vertices of} $\somepath$, as well as the \definitive{edges of} $\somepath$, we mean those it contains.
Since a given vertex may occur more than once in a path, we may have to talk about \definitive{occurrences of vertices in a path} to distinguish these equal values.
We use the following notations:
\begin{itemize}
\item
the \definitive{concatenation} of two paths $\somepath_1 = (\somevertex_0, \someedge_1, \dots, \someedge_k, \somevertex_k)$ and $\somepath_2 = (\somevertex_k, \someedge_{k+1}, \dots, \someedge_n, \somevertex_n)$ is the path $\concatpath{\somepath_1}{\somepath_2} = (\somevertex_0, \someedge_1, \dots, \someedge_k, \somevertex_k, \someedge_{k+1}, \dots, \someedge_n, \somevertex_n)$;
\item
the \definitive{reverse} of a path $\somepath = (\somevertex_0, \someedge_1, \somevertex_1, \dots, \someedge_k, \somevertex_k)$ is the path $\reversedpath{\somepath} = (\somevertex_k, \someedge_k, \somevertex_{k-1}, \dots, \someedge_1, \somevertex_0)$;
\item
if $\somevertex$ and $\othervertex$ are two (occurrences of) vertices of a path $\somepath$, with $\somevertex$ occurring before $\othervertex$, $\subpath{\somepath}{\somevertex}{\othervertex}$ is the unique \definitive{sub-path} (\ie\ sub-sequence that is a path) of $\somepath$ with source $\somevertex$ and target $\othervertex$.
\end{itemize}

A path is \definitive{simple} if its edges are pairwise distinct and its vertices are pairwise distinct except possibly its endpoints which may be equal.
A path is \definitive{closed} if it has equal endpoints, otherwise it is \definitive{open}.
A \definitive{cycle} is a non-empty simple closed path.
A graph with no cycle is called \definitive{acyclic}.

%
\begin{lem}[Concatenation of Simple Paths]\label{lem:concatsimplepaths}
If $\somepath_1$ and $\somepath_2$ are two simple open paths and their unique common vertices are the target of $\somepath_1$ and the source of $\somepath_2$, and possibly the target of $\somepath_2$ and the source of $\somepath_1$, and if the last edge of $\somepath_1$ is different from the first edge of $\somepath_2$, then $\concatpath{\somepath_1}{\somepath_2}$ is simple.
\end{lem}

\begin{lem}[Concatenation of Disjoint Simple Paths]\label{lem:concatsimplepathsprefix}
If $\somepath_1$ and $\somepath_2$ are two simple open or empty paths such that the target of $\somepath_1$ is the source of $\somepath_2$ and this is their unique common vertex, then $\concatpath{\somepath_1}{\somepath_2}$ is simple and open or empty.
\end{lem}


Given a partial graph $\somegraph = (\vertexset, \edgeset, \incidence)$, a \definitive{sub-graph} of $\somegraph$ is a partial graph $\somegraph' = (\vertexset', \edgeset', \incidence')$ such that $\vertexset' \subseteq \vertexset$, $\edgeset' \subseteq \edgeset$ and $\incidence'$ is the restriction of $\incidence$ to $\edgeset'$ in its domain and sets of $\vertexset'$ in its codomain.
In other words, for $\someedge\in\edgeset'$ we have $\incidence'(\someedge)=\incidence(\someedge)\cap\vertexset'$.
Remark that a sub-graph $\somegraph'$ is uniquely defined by the data of $\vertexset'$ and $\edgeset'$.
For $\somegraph_{1}$ and $\somegraph_{2}$ sub-graphs of the partial graph $\somegraph$, the partial sub-graph $\somegraph_{1}\cup\somegraph_{2}$ (\resp\ $\somegraph_{1}\cap\somegraph_{2}$) has for vertices and edges the union (\resp\ intersection) of those of $\somegraph_{1}$ and $\somegraph_{2}$.
One can define a partial order $\subseteq$ on sub-graphs as follows: given $\somegraph_{1} = (\vertexset_1, \edgeset_1, \incidence_1)$ and $\somegraph_{2} = (\vertexset_2, \edgeset_2, \incidence_2)$ sub-graphs of a same partial graph, $\somegraph_{1} \subseteq \somegraph_{2}$ if and only if $\vertexset_1 \subseteq \vertexset_2$ and $\edgeset_1 \subseteq \edgeset_2$ (these last two being inclusions of sets).
In other words, $\somegraph_{1} \subseteq \somegraph_{2}$ if $\somegraph_1$ is a sub-graph of $\somegraph_2$.

Connectedness is not immediate to define in partial graphs because paths go from vertices to vertices.
Two vertices $\somevertex$ and $\othervertex$ are \definitive{connected} when there exists a path with endpoints $\somevertex$ and $\othervertex$.
Two edges are \definitive{connected} if they are equal (this is necessary for edges with no endpoint) or if they are incident to two connected vertices.
An edge $\someedge$ and a vertex $\somevertex$ are \definitive{connected} if $\someedge$ is incident to a vertex connected to $\somevertex$.
A partial graph $\somegraph = (\vertexset, \edgeset, \incidence)$
is \definitive{connected} when it is non-empty,
and for any pair \((x,y)\in(\vertexset\cup\edgeset)^2\), \(x\) and \(y\) are connected.
A \definitive{connected component} is a maximal connected sub-graph.

\subsection{Local Coloring and Cusps}\label{sec:localcolor}

Let $\somegraph$ be a partial graph.
A \definitive{local coloring} of $\somegraph$ is given
by a finite set \(\colors\) (the set of \definitive{colors})
and a function $\coloring$ mapping each pair of an edge $\someedge$ and one of its endpoints $\somevertex$,
to a color $\coloring(\someedge, \somevertex)$.
The intuition is that given an edge $\someedge$ and one of its endpoints $\somevertex$, $\coloring(\someedge, \somevertex)$ is the color of $\someedge$ according to $\somevertex$.
A local coloring can also be seen as a coloring of half-edges, \ie\ $\coloring(\someedge, \somevertex)$ is the color of the half of $\someedge$ near $\somevertex$.
When drawing a partial graph, we therefore represent $\coloring(\someedge, \somevertex)$ by coloring the part of $\someedge$ touching $\somevertex$, with colors also given by the shape of the edges (\textcolor{red}{solid}, \textcolor{blue}{dashed}, etc.).
We recover the standard notion of \definitive{edge-coloring}, which maps edges to colors, when for every edge $\someedge$, $\coloring(\someedge, \_)$ has the same value for all endpoints of $\someedge$.
An example of locally colored partial graph is given on \cref{fig:ex_loc_col},
where
$\coloring(\someedge, \somevertex) = \text{\textcolor{red}{solid}}$,
$\coloring(\someedge, \othervertex) = \text{\textcolor{red}{solid}}$,
$\coloring(\otheredge, \othervertex) = \text{\textcolor{red}{solid}}$,
$\coloring(\otheredge, \otherothervertex) = \text{\textcolor{blue}{dashed}}$,
$\coloring(\otherotheredge, \somevertex) = \text{\textcolor{blue}{dashed}}$,
$\coloring(\otherotheredge, \otherothervertex) = \text{\textcolor{red}{solid}}$ and
$\coloring(\otherotherotheredge, \somevertex) = \text{\textcolor{violet}{dotted}}$.
We generally keep the set \(\colors\) implicit and only define the coloring function \(\coloring\),
but the reader should be aware that this function needs not be surjective.

A \definitive{cusp} at $\somevertex$ of color $\somecolor$ is a triple $(\someedge, \somevertex, \otheredge)$ where $\someedge$ and $\otheredge$ are \emph{distinct} edges such that $\somevertex$ is an endpoint of both of these edges
and $\coloring(\someedge, \somevertex) = \somecolor = \coloring(\otheredge, \somevertex)$
(as a consequence $(\otheredge, \somevertex, \someedge)$ is also a cusp).
In this case, $\somevertex$ is called the \definitive{vertex of the cusp}, $\somecolor$ the \definitive{color of the cusp} and $(\somevertex, \somecolor)$ is called a \definitive{cusp-point}.
The locally colored partial graph in \cref{fig:ex_loc_col} has two cusps,
$(\someedge, \othervertex, \otheredge)$ and $(\otheredge, \othervertex, \someedge)$, both of vertex $\othervertex$ and color \textcolor{red}{solid},
so that $(\othervertex, \text{\textcolor{red}{solid}})$ is the only cusp-point of this partial graph.

More generally,
we will consider \definitive{vertex-color pairs} which are arbitrary pairs made of a vertex and a color.
A cusp-point is a particular instance of a vertex-color pair.
Note that, in general, having fixed a locally colored partial graph, 
we might well consider vertex-color pairs \((\somevertex,\somecolor)\in\vertexset\times\colors\)
that are not realized in the graph
-- \ie, such that there is no edge \(\someedge\) adjacent to \(\somevertex\)
and such that \(\somecolor=\coloring(\someedge,\somevertex)\).

A \definitive{cusp of a path} $\somepath$ is either a cusp made by a sub-sequence $(\someedge, \somevertex, \otheredge)$ of this path -- named an \definitive{internal cusp} of $\somepath$ -- or, in case $\somepath$ is closed, a cusp $(\someedge_n, \somevertex_0, \someedge_1)$ made by its last edge $\someedge_n$, its source (and target) $\somevertex_0$ and its first edge $\someedge_1$.
Remark that the reverse of a path contains the same number of cusps as this path.
A \definitive{cusp-free path}, also called an \definitive{alternating path}, is one without cusp.
Given a \emph{non-empty} path $\somepath$, whose source is $\somevertex_0$ and first edge is $\someedge_1$, its \definitive{starting color} is $\coloring(\someedge_1, \somevertex_0)$.
Similarly, if its target is $\somevertex_n$ and its last edge is $\someedge_n$, then the \definitive{ending color} of $\somepath$ is $\coloring(\someedge_n, \somevertex_n)$.
Remark the starting (\resp\ ending) color of $\reversedpath{\somepath}$ is the ending (\resp\ starting) color of $\somepath$.
For instance, in the partial graph depicted on \cref{fig:ex_loc_col} the path $(\somevertex, \someedge, \othervertex, \otheredge, \otherothervertex)$ has one cusp at $u$ of color \textcolor{red}{solid}, its starting color is \textcolor{red}{solid} and its ending color is \textcolor{blue}{dashed}.

\begin{fact}
\label{lem:reversing_trick}
Let $\somecycle$ be a cycle with no cusp at its source, and $\somecolor$ a color.
Then $\somecolor$ is not the starting color of $\somecycle$ or $\somecolor$ is not the starting color of $\reversedpath{\somecycle}$.
\end{fact}

We call \definitive{splitting} a vertex $\somevertex$ such that any cycle containing it has a cusp at $\somevertex$.
We will show in \cref{sec:stdyeo} that this fits the notion at play in the conclusion of Yeo's theorem~\cite{yeotheorem}.

\begin{rem}
\label{rem:local_coloring_on_vertex}
We invented this ``local coloring'', which is not standard in the literature, and the name ``cusp''.
When used only through the notions of cusps and splitting vertices, that a same color is used on different vertices has no impact.
Hence, we could use different sets of colors depending on each vertex, or not use more colors than the maximal degree of the graph.
Equivalently, a local coloring is an equivalence relation on the edges incident to $\somevertex$, for each vertex $\somevertex$.
We keep the idea of local coloring as it is a direct generalization of edge-coloring.
\end{rem}

\begin{figure}
\begin{minipage}[b]{.40\textwidth}
\centering
\begin{tikzpicture}
\begin{genscope}
	\node (u) at (0,2) {$\othervertex$};
	\node (v) at (-1,0) {$\somevertex$};
	\node (w) at (1,0) {$\otherothervertex$};
	\coordinate (uv) at (-.5,1);
	\coordinate (uw) at (.5,1);
	\coordinate (vw) at (0,0);
	\coordinate (vt) at (-1.8,0.8);
\end{genscope}
	\node at (-.7,1.2) {$\someedge$};
	\node at (.7,1.2) {$\otheredge$};
	\node at (0,-.2) {$\otherotheredge$};
	\node at (-1.4,0.75) {$\otherotherotheredge$};
\begin{genscope}[red][-]
	\path (u) edge (uv);
	\path (u) edge (uw);
	\path (v) edge (uv);
	\path (w) edge (vw);
\end{genscope}
\begin{genscope}[blue][-, dashed]
	\path (v) edge (vw);
	\path (w) edge (uw);
\end{genscope}
\begin{genscope}[violet][-,dotted]
	\path (v) edge (vt);
\end{genscope}
\end{tikzpicture}
\caption{Example of locally colored partial graph}
\label{fig:ex_loc_col}
\end{minipage}\hfill
\begin{minipage}[b]{.60\textwidth-1em}
\centering
\begin{tikzpicture}
\begin{genscope}
	\node (v) at (0,0) {$\somevertex$};
	\node (pier) at (3,-1) {$\othervertex$};
	\node (other) at (6,-1) {$\otherothervertex$};
\end{genscope}
\begin{myscopehigh}{cyan}
 	\path[out=-90,in=-90] (pier.-50) edge (other.-120);
	\path[out=180,in=0] (other.200) edge node[chemin nomme,inner sep=3pt,below]{$\otherclosed$} (pier.-20);
\end{myscopehigh}
\begin{genscope}
	\path[out=-90,in=180] (v) edge[-] (pier);
 	\path[out=0,in=180] (pier) edge[-] (other);
 	\path[out=45,in=45] (other)  edge[-] node[chemin nomme,midway,below]{$\somecycle$} (v);
 	\path[out=-90,in=-90] (pier) edge[-,red] node[chemin nomme,fill=white]{$\otherpath$} (other);
\end{genscope}
\begin{myscopehigh}{blue}
	\path[out=-90,in=180] (v.250) edge[looseness=1.05] node[chemin nomme,inner sep=3pt,below]{$\someclosedyeo$} (pier.200);
	\path[out=-90,in=-90] (pier.250) edge[looseness=1.1] (other.290);
 	\path[out=45,in=45] (other.25) edge[looseness=1.07] (v.65);
\end{myscopehigh}
\end{tikzpicture}
\caption{Illustration of \cref{lem:bjumping_Yeo_base}}
\label{fig:bjumping_Yeo_base}
\end{minipage}
\end{figure}


\subsection{Cusp Minimization}

The key ingredient for proving our Yeo-style theorem is showing that for any pair $(\somevertex, \somecolor)$ maximal for the strict partial order $\orderyeo$ (\cref{def:orderyeo}), $\somevertex$ is splitting.
It is a consequence of the following:

\begin{lem}[Cusp Minimization]\label{lem:bjumping_Yeo_base}
Fix a partial graph $\somegraph$ with a local coloring.
Assume $\somecycle$ is a cycle starting from a vertex $\somevertex$, with no cusp at $\somevertex$ but containing a cusp of vertex $\othervertex$ and color $\somecolor$.
Suppose there exists a simple open cusp-free path $\otherpath$ starting from $\othervertex$ \emph{not} with color $\somecolor$, and ending on a vertex $\otherothervertex$ of $\somecycle$ (and with no other vertex in common with $\somecycle$ than $\othervertex$ and $\otherothervertex$).
Then either there exists a cusp-free cycle having $\otherpath$ as a sub-path
or there exists a cycle with source $\somevertex$, with no cusp at $\somevertex$ and with strictly less cusps than $\somecycle$.
\end{lem}

\begin{proof}
Use \cref{fig:bjumping_Yeo_base} as a reference for notations.
We use the notation $\somevertex_1$ for the occurrence of $\somevertex$ at the source of $\somecycle$, and $\somevertex_2$ for its occurrence at the target of $\somecycle$.
Call $\othercolor$ the ending color of $\otherpath$.

By symmetry (considering the reverse of $\somecycle$ if necessary), we can assume that $\otherothervertex$ is in $\subpath{\somecycle}{\othervertex}{\somevertex_2}$ and if $\otherothervertex=\somevertex_2$ then $\othercolor$ is not the starting color of $\somecycle$.
Indeed, if $\otherothervertex\notin\subpath{\somecycle}{\othervertex}{\somevertex_2}$, we reverse $\somecycle$.
Otherwise and if $\otherothervertex=\somevertex_2$, we apply \cref{lem:reversing_trick} to $\somecycle$ and $\othercolor$ to get that $\somecycle$ or $\reversedpath{\somecycle}$ respects our assumption.

Consider the cycles $\someclosedyeo \eqdef \concatpath{\subpath{\somecycle}{\somevertex_1}{\othervertex}}{\concatpath{\otherpath}{\subpath{\somecycle}{\otherothervertex}{\somevertex_2}}}$ and $\otherclosed \eqdef \concatpath{\otherpath}{\subpath{\reversedpath{\somecycle}}{\otherothervertex}{\othervertex}}$ (see \cref{fig:bjumping_Yeo_base}).
Both of these paths are indeed non-empty (because $\otherpath$ is non-empty) and simple (using \cref{lem:concatsimplepaths}).

Let us count the number of cusps in $\somecycle$, $\someclosedyeo$ and $\otherclosed$.
Recall that $\othervertex$ is a cusp of $\somecycle$ of color $\somecolor$, $\otherpath$ is cusp-free and its starting color is not $\somecolor$, and that $\someclosedyeo$ has no cusp at $\somevertex$ (by our symmetry argument above).
Thus, there are $n_1+1+n_2+\isbridge{\somecycle}{\otherothervertex}+n_3$ cusps in $\somecycle$, $n_1+\isbridge{\someclosedyeo}{\otherothervertex}+n_3$ cusps in $\someclosedyeo$, and $\isbridge{\otherclosed}{\otherothervertex}+n_2$ cusps in $\otherclosed$, where:
\begin{itemize}
\item $n_1$ (\resp\ $n_2$, $n_3$) is the number of cusps of $\subpath{\somecycle}{\somevertex_1}{\othervertex}$ (\resp\ $\subpath{\somecycle}{\othervertex}{\otherothervertex}$, $\subpath{\somecycle}{\otherothervertex}{\somevertex_2}$);
\item $\isbridge{\somecycle}{\otherothervertex}$ (\resp\ $\isbridge{\someclosedyeo}{\otherothervertex}$, $\isbridge{\otherclosed}{\otherothervertex}$) is $1$ if $\somecycle$ (\resp\ $\someclosedyeo$, $\otherclosed$) has a cusp at $\otherothervertex$ and $0$ otherwise.
\end{itemize}
If $\someclosedyeo$ has strictly less cusps than $\somecycle$ we are done, otherwise $\isbridge{\someclosedyeo}{\otherothervertex}\geq 1+n_2+\isbridge{\somecycle}{\otherothervertex}$.
Hence, $n_2=0$, $\isbridge{\somecycle}{\otherothervertex}=0$ and $\isbridge{\someclosedyeo}{\otherothervertex}=1$.
But the last two imply $\isbridge{\otherclosed}{\otherothervertex} = 0$, so that $\otherclosed$ is a cusp-free cycle containing $\otherpath$ as a sub-path.
\end{proof}

\begin{rem}
The proof of \cref{lem:bjumping_Yeo_base} is obviously constructive: changing if necessary the orientation of $\somecycle$, the cycles we are looking for are $\concatpath{\otherpath}{\subpath{\reversedpath{\somecycle}}{\otherothervertex}{\othervertex}}$ and $\concatpath{\subpath{\somecycle}{\somevertex_1}{\othervertex}}{\concatpath{\otherpath}{\subpath{\somecycle}{\otherothervertex}{\somevertex_2}}}$.
\end{rem}

For a vertex $\somevertex$, we denote by $\mincycles{\somevertex}$ the set of cycles:
\begin{itemize}
\item with source (and target) $\somevertex$;
\item whose last and first edges do not make a cusp at $\somevertex$;
\item with a minimal number of cusps among all cycles respecting the previous two conditions.
\end{itemize}
Observe that the set $\mincycles{\somevertex}$ is empty if and only if $\somevertex$ is splitting.

\begin{cor}[Cusp Cycling]\label{lem:bjumping_Yeo}
Fix a partial graph $\somegraph$ with a local coloring,
$\somevertex$ a vertex of $\somegraph$, and $\somecycle\in\mincycles{\somevertex}$.
Assume $\somecycle$ contains a cusp $(\otheredge_1,\somepier,\otheredge_2)$ of color $\somecolor$,
and that there is a simple open cusp-free path $\otherpath$ starting from $\somepier$
\emph{not} with color $\somecolor$, and
ending on a vertex $\othervertex$ of $\somecycle$.
Then there exists a cusp-free cycle containing $\somepier$ in $\somegraph$. 
\end{cor}

\begin{proof}
By taking a prefix of $\otherpath$ if necessary, we can assume that $\otherpath$ does not share any vertex with $\somecycle$ other than its (distinct) endpoints $\somepier$ and $\othervertex$.
We apply \cref{lem:bjumping_Yeo_base}.
Since $\somecycle\in\mincycles{\somevertex}$, we cannot find a cycle with source $\somevertex$, no cusp at $\somevertex$ and strictly less cusps than $\somecycle$.
We thus have a cusp-free cycle containing $\somepier$.
\end{proof}

The above corollary is sufficient for our proof of Yeo's theorem (in \cref{sec:genyeo}).
The next result will be useful only for a further generalization of Yeo's theorem, that allows some cusp-free cycles (in \cref{sec:yeo_mall}).

\begin{figure}
\centering
\begin{tikzpicture}
\begin{genscope}
	\node (v) at (0,0) {$\somevertex$};
	\node (x) at (2.5,-1) {$\otherothervertex$};
	\node (y) at (5.5,-1) {$\otherotherothervertex$};
	\node[magenta] (pier0) at (4,-3) {$\somepier_0$};
	\node[magenta] (pier) at (4.5,-4.5) {$\somepier$};
	\node (other) at (8,-1) {$\othervertex$};
	\node (te) at (4.5,-6) {$\someleaf$};
	\coordinate (etb) at (4.65,-6.75);
	\coordinate (pr') at (6.1,-5.75);
	\coordinate (pr) at (7.25,-3.75);
	\coordinate (rho1up) at (5,-1.5);
	\coordinate (rho2up) at (4.5,-3.75);

	\path[out=-90,in=180] (v) edge[-] (x);
	\path (x) edge[-] (y);
	\path[out=0,in=180] (y) edge[-] (other);
	\path[out=45,in=45] (other)
          edge[-]
          node[chemin nomme,midway,below]{$\somecycle$}
          (v);
	\path[out=-90,in=180] (x)
          edge[-,magenta]
          node[chemin nomme]{$\otherotherpath_0$}
          (pier0);
	\path[out=-125,in=180] (pier0) edge[-,orange] (rho2up);
	\path[out=0,in=180] (rho2up) edge[-,orange] (rho1up);
	\path[out=0,in=45] (rho1up)
          edge[-,orange]
          node[chemin nomme,below]{$\otherotherpath_1$}
          (pier);
	\path[out=-90,in=0] (y)
          edge[-,magenta]
          node[chemin nomme,above]{$\otherotherotherpath_0$}
          (pier0);
	\path[out=-90,in=135] (pier0)
          edge[-,magenta]
          node[chemin nomme,below]{$\otherotherotherpath_1$}
          (pier);
	\path[out=-90,in=90] (pier)
          edge[-,olive]
          node[arete nommee,left,inner sep=5pt]{$\someedge$}
          (te);
	\path[out=-90,in=-90] (te)
          edge[-,red]
          node[chemin nomme,below,inner sep=7pt]{$\somepath$}
          (other);
\end{genscope}
\begin{myscopehigh}{blue}
	\path[out=-90,in=180] (v.250) edge node[chemin nomme,inner sep=3pt,below]{$\somecycle'$} (x.200);
	\path[out=-90,in=180] (x.-110) edge (pier0.200);
	\path[out=-90,in=0] (y.-70) edge (pier0.-20);
	\path[out=180,in=0] (other.200) edge (y.-20);
	\path[out=45,in=45] (other.25) edge[looseness=1.05] (v.65);
\end{myscopehigh}
\begin{myscopehigh}{cyan}
	\path[out=-90,in=135] (pier0.-70) edge (pier.115);
	\path[out=-90,in=90] (pier.-70) edge (te.70);
	\path[out=-90,in=110] (te.-70) edge (etb);
	\path[out=-70,in=-130] (etb) edge (pr');
	\path[out=50,in=-120] (pr') edge (pr);
	\path[out=70,in=-90] (pr) edge node[chemin nomme,above]{$\otherpath$} (other.250);
\end{myscopehigh}
\end{tikzpicture}
\caption{Illustration of \cref{lem:bjumping_gen}}
\label{fig:bjumping_gen}
\end{figure}

\begin{cor}[Cusp Minimization 2]\label{lem:bjumping_gen}
Fix a partial graph $\somegraph$ with a local coloring.
Assume $\somecycle$ is a cycle starting from a vertex $\somevertex$, with no cusp at $\somevertex$, and $\otherothervertex$ and $\otherotherothervertex$ are two vertices of $\somecycle$, both different from $\somevertex$, with $\otherothervertex$ occurring before (or equal to) $\otherotherothervertex$, and with at least one cusp between $\otherothervertex$ and $\otherotherothervertex$ (possibly at $\otherothervertex$ or $\otherotherothervertex$).
Suppose $\somepier$ is a vertex such that we have a path $\otherotherpath$ from $\otherothervertex$ to $\somepier$, a path $\otherotherotherpath$ from $\otherotherothervertex$ to $\somepier$ and an edge $\someedge$ with endpoints $\somepier$ and $\someleaf$ such that:
\begin{itemize}
\item $\concatpath{\otherotherpath}{(\somepier,\someedge,\someleaf)}$ is a simple cusp-free path whose starting color is not the ending color of $\subpath{\somecycle}{\somevertex}{\otherothervertex}$;
\item $\concatpath{\otherotherotherpath}{(\somepier,\someedge,\someleaf)}$ is a simple cusp-free path whose starting color is not the starting color of $\subpath{\somecycle}{\otherotherothervertex}{\somevertex}$;
\item the only vertices of $\otherotherpath$ or $\otherotherotherpath$ which may belong to $\concatpath{\subpath{\somecycle}{\otherotherothervertex}{\somevertex}}{\subpath{\somecycle}{\somevertex}{\otherothervertex}}$ are $\otherothervertex$ and $\otherotherothervertex$ (which might be equal to $\somepier$ when one of these paths is empty).
\end{itemize}
If there is a simple open cusp-free path $\concatpath{(\somepier,\someedge,\someleaf)}{\somepath}$ such that:
\begin{itemize}
\item its target $\othervertex$ belongs to $\concatpath{\subpath{\somecycle}{\otherotherothervertex}{\somevertex}}{\subpath{\somecycle}{\somevertex}{\otherothervertex}}$
\item both $\otherotherpath$ and $\otherotherotherpath$ have no vertex in common with $\somepath$
\end{itemize}
then either there exists a cusp-free cycle containing $\someedge$
or there exists a cycle with source $\somevertex$, with no cusp at $\somevertex$ and with strictly less cusps than $\somecycle$.
\end{cor}

\begin{proof}
Use \cref{fig:bjumping_gen} as a reference for notations.
By taking a prefix of $\somepath$ if necessary, we can assume that $\somepath$ does not share any vertex with $\concatpath{\subpath{\somecycle}{\otherotherothervertex}{\somevertex}}{\subpath{\somecycle}{\somevertex}{\otherothervertex}}$ other than its target $\othervertex$ (in case $\someleaf$ is in $\concatpath{\subpath{\somecycle}{\otherotherothervertex}{\somevertex}}{\subpath{\somecycle}{\somevertex}{\otherothervertex}}$, $\somepath$ is the empty path).
Consider $\somepier_0$ the first vertex of $\otherotherpath$ which belongs to $\otherotherotherpath$.
We have $\otherotherpath=\concatpath{\otherotherpath_0}{\otherotherpath_1}$
and $\otherotherotherpath=\concatpath{\otherotherotherpath_0}{\otherotherotherpath_1}$ with $\somepier_0$ as target of $\otherotherpath_0$ and $\otherotherotherpath_0$.
We consider the closed path $\somecycle' \eqdef \concatpath{\subpath{\somecycle}{\somevertex}{\otherothervertex}}{\concatpath{\concatpath{\otherotherpath_0}{\reversedpath{\otherotherotherpath_0}}}{\subpath{\somecycle}{\otherotherothervertex}{\somevertex}}}$ which starts with $\somevertex$ and has no cusp at $\somevertex$.
The closed path $\somecycle'$ is simple since $\concatpath{(\concatpath{\subpath{\somecycle}{\otherotherothervertex}{\somevertex}}{\subpath{\somecycle}{\somevertex}{\otherothervertex}})}{(\concatpath{\otherotherpath_0}{\reversedpath{\otherotherotherpath_0}})}$ is, by \cref{lem:concatsimplepaths,lem:concatsimplepathsprefix}, thus it is a cycle.
If $\somecycle'$ has no cusp at $\somepier_0$ then $\somecycle'$ has strictly less cusps than $\somecycle$ (which contains a cusp between $\otherothervertex$ and $\otherotherothervertex$).
Otherwise $\somecycle'$ contains a cusp at $\somepier_0$ (and at most as many cusps as $\somecycle$).
No vertex of $\otherotherotherpath_1$, except its source, belongs to $\somecycle'$, nor does any vertex of $\somepath$, except its target.
We apply \cref{lem:bjumping_Yeo_base} to $\somecycle'$ and $\otherpath \eqdef \concatpath{\otherotherotherpath_1}{\concatpath{(\somepier,\someedge,\someleaf)}{\somepath}}$.
We have already seen that $\somecycle'$ satisfies the hypotheses of \cref{lem:bjumping_Yeo_base}.
Now concerning $\concatpath{\otherotherotherpath_1}{\concatpath{(\somepier,\someedge,\someleaf)}{\somepath}}$:
it is a simple path by \cref{lem:concatsimplepathsprefix} since $\concatpath{(\somepier,\someedge,\someleaf)}{\somepath}$ is simple and $\otherotherotherpath$ has no vertex in common with $\somepath$, and it is cusp-free since $\concatpath{(\somepier,\someedge,\someleaf)}{\somepath}$ and $\concatpath{\otherotherotherpath}{(\somepier,\someedge,\someleaf)}$ are cusp-free.
Finally, since $\somecycle'$ has a cusp at $\somepier_0$, $\concatpath{\otherotherotherpath}{(\somepier,\someedge,\someleaf)}$ is cusp-free and the starting color of $\concatpath{\otherotherotherpath}{(\somepier,\someedge,\someleaf)}$ is not the starting color of $\subpath{\somecycle}{\otherotherothervertex}{\somevertex}$, the starting color of $\concatpath{\otherotherotherpath_1}{\concatpath{(\somepier,\someedge,\someleaf)}{\somepath}}$ cannot be the color of the cusp at $\somepier_0$ in $\somecycle'$.
\end{proof}

\section{A Generalization of Yeo's Theorem}\label{sec:genyeo}

\subsection{Parametrized Local Yeo}\label{sec:paramyeo}

We prove a version of Yeo's theorem~\cite{yeotheorem} (see \cref{th:yeo} for the original statement by Yeo) for locally colored partial graphs, which is moreover parametrized by the choice of a set of vertex-color pairs (subject to a technical condition): \cref{th:ParamLocalYeo}.
This result allows us to find splitting vertices in a locally colored graph with no cusp-free cycle.
We first fix a partial graph $\somegraph$ with a local coloring $\coloring$.

The main idea is to follow a path that is an evidence of progression, \ie\ a strict partial order: a vertex is smaller than another when there is a(n appropriate) path from the first one to the second, and we will prove that a maximal vertex is splitting.
As the hypothesis of the theorem is about cusp-free cycles, it makes sense to consider cusp-free paths in this ordering $\sobfnbpath$.
However, two issues prevent $\sobfnbpath$ from being an order.
First, the concatenation of two cusp-free paths may not be cusp-free.
To have $\sobfnbpath$ transitive, we impose a condition on the starting and ending colors of the cusp-free path -- which is why we consider vertex-color pairs and not simply vertices.
Second, there is no reason for this relation $\sobfnbpath$ of ``being linked by a cusp-free path'' to not loop.
Hence, we add a condition on the path that there is no way to go back on it, yielding from $\sobfnbpath$ a relation $\orderyeo$ which will be our strict partial order -- see \cref{fig:ex_orderyeo} for an illustration.
This entails the following:

\begin{defi}\label{def:orderyeo}
Let $\somevertex$ and $\othervertex$ be vertices, and $\somecolor$ and $\othercolor$ be colors.
\begin{itemize}
\item
We write $(\somevertex, \somecolor) \sobfnbpath[\somepath] (\othervertex, \othercolor)$ if $\somepath$ is a simple open cusp-free path from $\somevertex$ to $\othervertex$ with starting color \emph{not} $\somecolor$ and with ending color $\othercolor$.
We simply write $(\somevertex, \somecolor) \sobfnbpath (\othervertex, \othercolor)$ whenever such a path exists.
\item
We note $(\somevertex, \somecolor) \orderyeo[\somepath] (\othervertex, \othercolor)$ when $(\somevertex, \somecolor) \sobfnbpath[\somepath] (\othervertex, \othercolor)$ and for all vertex $\otherothervertex$, color $\otherothercolor$ and path $\otherpath$ such that $(\othervertex, \othercolor) \sobfnbpath[\otherpath] (\otherothervertex, \otherothercolor)$,
$\otherothervertex$ is not in $\somepath$.
We simply write $(\somevertex, \somecolor) \orderyeo (\othervertex, \othercolor)$ when there is some $\somepath$ such that $(\somevertex, \somecolor) \orderyeo[\somepath] (\othervertex, \othercolor)$.
\end{itemize}
\end{defi}

\begin{figure}
\centering
\begin{tikzpicture}
	\node at (2,-1.5) {$(\somevertex, \somecolor) \sobfnbpath[\somepath] (\othervertex, \othercolor)$ and $(\somevertex, \somecolor) \orderyeo[\somepath] (\othervertex, \othercolor)$};
\begin{genscope}
	\node (v) at (0,0) {$\somevertex$};
	\node (x) at (2,0) {$\otherothervertex$};
	\node (u) at (4,0) {$\othervertex$};
\end{genscope}
\begin{genscope}[blue][solid,-]
	\path (v) edge node[chemin nomme,pos=0.8]{$\somepath$} (x);
	\path (x) edge (u);
\end{genscope}
\begin{genscope}[red][solid,-]
	\path[out=-90,in=-90] (u) edge node[chemin nomme,midway]{$\otherpath$} (x);
\end{genscope}
\begin{genscope}[green][solid,-]
	\path (v) edge node[arete nommee]{$\somecolor$} ++(-.75,0);
\end{genscope}
\begin{genscope}[orange][solid,-]
	\path (u) edge node[arete nommee]{$\othercolor$} ++(-.75,0);
	\path (u) edge node[arete nommee]{$\othercolor$} ++(0,-.75);
\end{genscope}
\end{tikzpicture}
\quad\vrule\quad
\begin{tikzpicture}
	\node at (2,-1.5) {$(\somevertex, \somecolor) \sobfnbpath[\somepath] (\othervertex, \othercolor)$ but $(\somevertex, \somecolor) \cancel{\orderyeo[\somepath]} (\othervertex, \othercolor)$};
\begin{genscope}
	\node (v) at (0,0) {$\somevertex$};
	\node (x) at (2,0) {$\otherothervertex$};
	\node (u) at (4,0) {$\othervertex$};
\end{genscope}
\begin{genscope}[blue][solid,-]
	\path (v) edge node[chemin nomme,pos=0.8]{$\somepath$} (x);
	\path (x) edge (u);
\end{genscope}
\begin{genscope}[red][solid,-]
	\path[out=-90,in=-90] (u) edge node[chemin nomme,midway]{$\otherpath$} (x);
\end{genscope}
\begin{genscope}[green][solid,-]
	\path (v) edge node[arete nommee]{$\somecolor$} ++(-.75,0);
\end{genscope}
\begin{genscope}[orange][solid,-]
	\path (u) edge node[arete nommee]{$\othercolor$} ++(-.75,0);
\end{genscope}
\begin{genscope}[brown][solid,-]
	\path (u) edge node[arete nommee]{$\othercolor'\neq\othercolor$} ++(0,-.75);
\end{genscope}
\end{tikzpicture}
\caption{Illustration of the relation $\orderyeo$ (\cref{def:orderyeo})}
\label{fig:ex_orderyeo}
\end{figure}

\begin{rem}\label{rem:sobfnbpath_prefix}
If $(\somevertex, \somecolor) \sobfnbpath[\somepath] (\othervertex, \othercolor)$ and $\somepath'$ is a \emph{non-empty} prefix of $\somepath$, then $(\somevertex, \somecolor) \sobfnbpath[\somepath'] (\otherothervertex, \otherothercolor)$ where $\otherothervertex$ is the target of $\somepath'$ and $\otherothercolor$ its ending color.
\end{rem}

\begin{lem}\label{lem:orderyeo_sobfnbpath}
Let $\somevertex$, $\othervertex$ and $\otherothervertex$ be vertices, $\somecolor$, $\othercolor$ and $\otherothercolor$ be colors, and $\somepath$ and $\otherpath$ be paths.
If $(\somevertex, \somecolor)\orderyeo[\somepath](\othervertex, \othercolor)$
and $(\othervertex, \othercolor)\sobfnbpath[\otherpath](\otherothervertex, \otherothercolor)$
then $(\somevertex, \somecolor)\sobfnbpath[\concatpath{\somepath}{\otherpath}](\otherothervertex, \otherothercolor)$.
\end{lem}
\begin{proof}
  Assume $(\somevertex, \somecolor)\orderyeo[\somepath](\othervertex, \othercolor)\sobfnbpath[\otherpath](\otherothervertex, \otherothercolor)$,
  and consider the path $\concatpath{\somepath}{\otherpath}$,
  with source $\somevertex$ and target $\otherothervertex$.
  Its starting color is not $\somecolor$ and its ending color is $\otherothercolor$.

  If $\otherpath$ contains a vertex of $\somepath$ which is not $\othervertex$, let $\otherpath'$ be the prefix of $\otherpath$ ending on the first such occurrence $\otherothervertex'$ with ending color $\otherothercolor'$.
  We get $(\othervertex, \othercolor)\sobfnbpath[\otherpath'](\otherothervertex',\otherothercolor')$ (\cref{rem:sobfnbpath_prefix}) with $\otherothervertex'\in\somepath$, contradicting $(\somevertex, \somecolor)\orderyeo[\somepath](\othervertex, \othercolor)$.
  We can thus use \cref{lem:concatsimplepathsprefix} to deduce that $\concatpath{\somepath}{\otherpath}$ is simple and open.
  Finally $\concatpath{\somepath}{\otherpath}$ is cusp-free for it has no cusp at $\othervertex$ since the ending color of $\somepath$ is $\othercolor$ which is not the starting color of $\otherpath$.
\end{proof}

\begin{lem}\label{lem:orderyeo_is_order}
  The relation $\orderyeo$ is a strict partial order on vertex-color pairs.
\end{lem}

\begin{proof}
  The relation $\orderyeo$ is irreflexive:
  we cannot have $(\somevertex, \somecolor)\sobfnbpath(\somevertex, \othercolor)$ by definition.

  Now assume $(\somevertex, \somecolor)\orderyeo[\somepath](\othervertex, \othercolor)\orderyeo[\otherpath](\otherothervertex, \otherothercolor)$.
  We obtain $(\somevertex, \somecolor)\sobfnbpath[\concatpath{\somepath}{\otherpath}](\otherothervertex, \otherothercolor)$ by \cref{lem:orderyeo_sobfnbpath}.
  To get $(\somevertex, \somecolor)\orderyeo[\concatpath{\somepath}{\otherpath}](\otherothervertex, \otherothercolor)$,
  it remains only to show that if $(\otherothervertex, \otherothercolor)\sobfnbpath[\otherotherpath](\otherotherothervertex, \otherotherothercolor)$
  then $\otherotherothervertex$ does not occur in $\concatpath{\somepath}{\otherpath}$.
  First observe that $\otherotherothervertex$ cannot occur in $\otherpath$ as $(\othervertex, \othercolor)\orderyeo[\otherpath](\otherothervertex, \otherothercolor)$.
  And, by \cref{lem:orderyeo_sobfnbpath}, $(\othervertex, \othercolor)\sobfnbpath[\concatpath{\otherpath}{\otherotherpath}](\otherotherothervertex, \otherotherothercolor)$,
  so $(\somevertex, \somecolor)\orderyeo[\somepath](\othervertex, \othercolor)$ implies that $\otherotherothervertex$ does not occur in $\somepath$ either.
\end{proof}

\begin{prop}\label{prop:order_max_Yeo}
Let $\somevertex$ be a non-splitting vertex of a locally colored partial graph with no cusp-free cycle.
For any color $\somecolor$
there exists a cusp-point $(\othervertex, \othercolor)$ such that $(\somevertex, \somecolor)\orderyeo(\othervertex, \othercolor)$.
\end{prop}
\begin{proof}
Since $\somevertex$ is not splitting, we have $\mincycles{\somevertex}\neq\emptyset$.
Take some $\somecycle\in\mincycles{\somevertex}$, considered as starting by $\somevertex$ with starting color not $\somecolor$, thanks to \cref{lem:reversing_trick}
(this is possible for there is no cusp at $\somevertex$ in elements of $\mincycles{\somevertex}$ and $\mincycles{\somevertex}$ is closed under reversing).
For $\somecycle$ cannot be cusp-free, it contains at least one cusp: denote by $\othervertex$ the vertex of the first cusp of $\somecycle$, and by $\othercolor$ its color.
We have $(\somevertex, \somecolor)\sobfnbpath[\subpath{\somecycle}{\somevertex}{\othervertex}](\othervertex, \othercolor)$,
and conclude $(\somevertex, \somecolor)\orderyeo[\subpath{\somecycle}{\somevertex}{\othervertex}](\othervertex, \othercolor)$
by \cref{lem:bjumping_Yeo}.
\end{proof}

We now state and prove our generalization of Yeo's Theorem, simply by applying \cref{prop:order_max_Yeo}.
A set $\somesetedge$ of vertex-color pairs \definitive{dominates cusp-points} if for any cusp-point $(\somevertex, \somecolor)$, either $(\somevertex, \somecolor) \in \somesetedge$ or there is $(\othervertex, \othercolor) \in \somesetedge$ with $(\somevertex, \somecolor)\orderyeo(\othervertex, \othercolor)$.

\begin{thm}[Parametrized Local Yeo]\label{th:ParamLocalYeo}
Consider $\somegraph$ a partial graph with a local coloring
and pose $\somesetedge$ a set of vertex-color pairs which dominates cusp-points.
If $\somegraph$ has no cusp-free cycle,
the vertex of any $\orderyeo$-maximal element of $\somesetedge$ (\ie\ for $\orderyeo$ restricted to $\somesetedge$) is splitting.
\end{thm}
\begin{proof}
Let $(\somevertex, \somecolor)$ be an element of $\somesetedge$.
If $\somevertex$ is not splitting then, by \cref{prop:order_max_Yeo}, we have a cusp-point $(\othervertex, \othercolor)$ such that $(\somevertex, \somecolor)\orderyeo(\othervertex, \othercolor)$.
By hypothesis on $\somesetedge$, either $(\othervertex, \othercolor)\in\somesetedge$ or we can find $(\otherothervertex, \otherothercolor)\in\somesetedge$ with $(\othervertex, \othercolor)\orderyeo(\otherothervertex, \otherothercolor)$.
This means $(\somevertex, \somecolor)$ is not maximal for $\orderyeo$ in $\somesetedge$.
\end{proof}

Note that the converse of \cref{th:ParamLocalYeo} is false: see \cref{fig:splitting_terminal_no_max} for a graph with a splitting vertex which is part of no maximal vertex-color pair.

\subsection{Terminality}

We prove here a result that will be of use for proof nets of both multiplicative and multiplicative-additive linear logic so as to affirm a maximal vertex for the ordering $\orderyeo$ can only be terminal.
In order to not prove it twice, we state here a generalization on locally colored partial graphs.

\begin{lem}\label{lem:terminal_not_maximal_order}
Let $\somegraph$ be a partial graph with a local coloring $\coloring$.
Consider an edge $\someedge$ of endpoints $\somevertex$ and $\othervertex$.
Assume that $(\somevertex, \coloring(\someedge, \somevertex))$ is not a cusp-point (\ie\ all edges $\otheredge$ with endpoint $\somevertex$ respect that $(\otheredge, \somevertex, \someedge)$ is not a cusp).
Then, either $\somevertex$ belongs to a cusp-free cycle, or for all colors $\somecolor \neq \coloring(\someedge, \somevertex)$, $(\somevertex, \somecolor) \orderyeo[(\somevertex, \someedge, \othervertex)](\othervertex, \coloring(\someedge,\othervertex))$.
\end{lem}
\begin{proof}
One has $(\somevertex, \somecolor) \sobfnbpath[(\somevertex, \someedge, \othervertex)](\othervertex, \coloring(\someedge,\othervertex))$ since $\somecolor \neq \coloring(\someedge, \somevertex)$.
The only vertices of the path $(\somevertex, \someedge, \othervertex)$ are $\somevertex$ and $\othervertex$.
If there exists $\somepath$ and $\othercolor$ such that $(\othervertex, \coloring(\someedge,\othervertex))\sobfnbpath[\somepath](\somevertex, \othercolor)$, then $\concatpath{(\somevertex, \someedge, \othervertex)}{\somepath}$ is a cycle containing $\somevertex$ (\cref{lem:concatsimplepaths}); furthermore, it is cusp-free as $(\somevertex, \coloring(\someedge, \somevertex))$ is not a cusp-point.
Otherwise, $(\somevertex, \somecolor) \orderyeo[(\somevertex, \someedge, \othervertex)](\othervertex, \coloring(\someedge,\othervertex))$ holds.
\end{proof}

\section{Comparison of our Generalized Yeo's Theorem with the Literature}\label{sec:graph_comparison}

\subsection{Local and Global Colorings}\label{sec:stdyeo}

First, remark our parametrized version implies a simpler one, closer to Yeo's theorem.

\begin{thm}[Local Yeo]\label{th:LocalYeo}
Consider $\somegraph$ a locally colored partial graph with at least one vertex.
If $\somegraph$ has no cusp-free cycle, then there exists a splitting vertex in $\somegraph$.
\end{thm}
\begin{proof}
If there is a vertex which is not the endpoint of any edge, it is splitting.
Otherwise, the set $\somesetedge$ of all vertex-color pairs of $\somegraph$
is finite and non-empty, and thus contains a maximal element $(\somevertex, \somecolor)$ with respect to $\orderyeo$ (\cref{lem:orderyeo_is_order}).
The vertex $\somevertex$ is splitting (\cref{th:ParamLocalYeo}).
\end{proof}

As an example, the partial graph depicted on \cref{fig:ex_loc_col} has no cusp-free cycle, and $\othervertex$ is its only splitting vertex.

We now bridge the gap with the terminology from Yeo's theorem~\cite{yeotheorem} and prove it is a direct consequence of our local version.
For $\somegraph$ a partial graph and $\somevertex$ one of its vertices, the partial graph $\somegraph\setminus\somevertex$ is the sub-graph obtained by removing $\somevertex$ from the vertices of $\somegraph$ (same edges with possibly less endpoints).
This gives an alternative characterization of splitting vertices in locally colored partial graphs:
a vertex $\somevertex$ is splitting if and only if any two edges with endpoint $\somevertex$ and connected in $\somegraph\setminus\somevertex$ have the same color on $\somevertex$.

Let us move to the terminology for total graphs:
\begin{itemize}
\item
As $\somegraph\setminus\somevertex$ leads in general to a partial graph, it has to be replaced with the operation $\somegraph-\somevertex$ on total graphs, which removes not only $\somevertex$ but also all its incident edges.
Connectedness on partial graphs gives the standard notion when restricted to total graphs, and a (non-empty) total graph is connected if all its vertices are.
\item
The standard notion of coloring is an \definitive{edge-coloring}, that maps edges to colors.
An \definitive{alternating cycle} for an edge-coloring is the restriction of the same notion for a local coloring: a cycle whose consecutive edges are of different colors, including its last and first edges.
\end{itemize}

\begin{thm}[Yeo's Theorem]\label{th:yeo}
If $\somegraph$ is a non-empty edge-colored graph with no alternating cycle,
then there exists a vertex $\somevertex$ of $\somegraph$ such that no connected component of $\somegraph-\somevertex$ is joined to $\somevertex$ with edges of more than one color.
\end{thm}
\begin{proof}
Call $\coloring$ the edge-coloring of $\somegraph$, we set a local coloring $\coloring'$ by $\coloring'(\someedge,\somevertex) \eqdef \coloring(\someedge)$ for any edge $\someedge$ and any endpoint $\somevertex$ of $\someedge$.
Cycles of $\somegraph$ that are alternating (or cusp-free) with respect to $\coloring'$ are exactly those that are alternating with respect to $\coloring$.
\Cref{th:LocalYeo} yields a splitting vertex $\somevertex$ for $\coloring'$: any two edges with endpoint $\somevertex$ and connected in $\somegraph\setminus\somevertex$ have the same color given by $\coloring'(\_, \somevertex)$.
That is, no connected component of $\somegraph-\somevertex$ is joined to $\somevertex$ with edges of more than one color.
\end{proof}

While at first glance \cref{th:LocalYeo} seems more general than \cref{th:yeo}, we deduce the first from the second by a graph encoding.
Partial edges play no role in \cref{th:LocalYeo}, so we consider only total graphs.
Consider $\somegraph$ a graph with local coloring $\coloring$, we associate with it a graph $\colenc{\somegraph}$ with an edge-coloring $\colenc{\coloring}$:
\begin{itemize}
\item all vertices of $\somegraph$ are considered as vertices of $\colenc{\somegraph}$ (and some are going to be added);
\item
\begin{minipage}[t]{0.57\linewidth}
with each edge $\someedge$ of $\somegraph$ of endpoints $\somevertex$ and $\othervertex$ such that $\coloring(\someedge,\somevertex) = \coloring(\someedge,\othervertex)$, we associate one edge $\otheredge$ in $\colenc{\somegraph}$ with the same endpoints as $\someedge$ and $\colenc{\coloring}(\otheredge) = \coloring(\someedge,\somevertex) = \coloring(\someedge,\othervertex)$;
\end{minipage}
\hfill
\begin{minipage}[t]{0.42\linewidth}
\begin{tikzpicture}[baseline=(base)]
\begin{scope}[every node/.style={circle,minimum size=3.5mm,inner sep=0,draw=black,thick}]
	\node (1) at (0.5,0) {$\somevertex$};
	\coordinate (1a) at (0.5,-0.48);
	\coordinate (1b) at (0.5,-0.52);
	\node (2) at (0.5,-1) {$\othervertex$};
	\node (1') at (3.5,0) {$\somevertex$};
	\node (2') at (3.5,-1) {$\othervertex$};
\end{scope}
	\node at (0.75,-0.5) {$\someedge$};
	\node at (1.75,-0.5) {$\mapsto$};
	\node at (3.75,-0.5) {$\otheredge$};
\begin{genscope}[red]
	\path (1) edge[-] (1a);
	\path (1b) edge[-] (2);
	\path (1') edge[-] (2');
\end{genscope}
\node (base) at (-1,-0.1) {\vphantom{x}};
\end{tikzpicture}
\end{minipage}
\item
\begin{minipage}[t]{0.57\linewidth}
with each edge $\someedge$ of $\somegraph$ of endpoints $\somevertex$ and $\othervertex$ such that $\coloring(\someedge,\somevertex) \neq \coloring(\someedge,\othervertex)$ and $\someedge$ belongs to a cycle, we associate two edges $\otheredge_1$ and $\otheredge_2$ and a new vertex $\otherothervertex$, the endpoints of $\otheredge_1$ being $\somevertex$ and $\otherothervertex$, and the endpoints of $\otheredge_2$ being $\otherothervertex$ and $\othervertex$, with $\colenc{\coloring}(\otheredge_1) = \coloring(\someedge,\somevertex)$ and $\colenc{\coloring}(\otheredge_2) = \coloring(\someedge,\othervertex)$;
\end{minipage}
\hfill
\begin{minipage}[t]{0.42\linewidth}
\begin{tikzpicture}[baseline=(3)]
\begin{scope}[every node/.style={circle,minimum size=3.5mm,inner sep=0,draw=black,thick}]
	\node (3) at (0.5,-4) {$\somevertex$};
	\coordinate (4) at (0.5,-4.8);
	\node (5) at (0.5,-5.6) {$\othervertex$};
	\node (3') at (3.5,-4) {$\somevertex$};
	\node (5') at (3.5,-5.6) {$\othervertex$};
\end{scope}
\begin{scope}[every node/.style={rectangle,thick,draw,inner sep=0,minimum size=10}]
	\node (4') at (3.5,-4.8) {$\otherothervertex$};
\end{scope}
	\node at (0.75,-4.8) {$\someedge$};
	\node at (1.75,-4.8) {$\mapsto$};
	\node at (3.75,-4.4) {$\otheredge_1$};
	\node at (3.75,-5.2) {$\otheredge_2$};
\begin{genscope}[red]
	\path (3) edge[-] (4);
	\path (3') edge[-] (4');
\end{genscope}
\begin{genscope}[blue]
	\path (4) edge[-,densely dashed] (5);
	\path (4') edge[-,densely dashed] (5');
\end{genscope}
	\path[bend right=100] (3) edge (5);
	\path[bend right=100] (3') edge (5');
\node at (-1,-5) {\vphantom{x}};
\end{tikzpicture}
\end{minipage}
\item
\begin{minipage}[t]{0.57\linewidth}
with each edge $\someedge$ of $\somegraph$ of endpoints $\somevertex$ and $\othervertex$ such that $\coloring(\someedge,\somevertex) \neq \coloring(\someedge,\othervertex)$ and $\someedge$ is not in a cycle (\ie\ $\someedge$ is a bridge), we associate one edge $\otheredge$ in $\colenc{\somegraph}$ with the same endpoints as $\someedge$ and an arbitrary color $\colenc{\coloring}(\otheredge)$.
\end{minipage}
\hfill
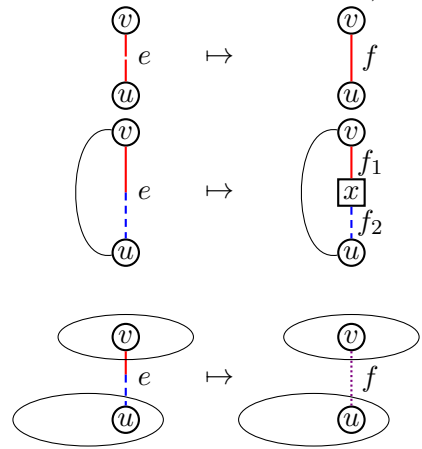
\begin{minipage}[t]{0.42\linewidth}
\begin{tikzpicture}[baseline=(6)]
\begin{scope}[every node/.style={circle,minimum size=3.5mm,inner sep=0,draw=black,thick}]
	\node (6) at (0.5,.4) {$\somevertex$};
	\coordinate (7) at (0.5,-.1);
	\node (8) at (0.5,-.7) {$\othervertex$};
	\node (6') at (3.5,.4) {$\somevertex$};
	\node (8') at (3.5,-.7) {$\othervertex$};
\end{scope}
	\node at (0.75,-.15) {$\someedge$};
	\node at (1.75,-.15) {$\mapsto$};
	\node at (3.75,-.15) {$\otheredge$};
\begin{genscope}[red][-]
	\path (6) edge (7);
\end{genscope}
\begin{genscope}[blue][-,densely dashed]
	\path (7) edge (8);
\end{genscope}
\begin{genscope}[violet][-,densely dotted]
	\path (6') edge (8');
\end{genscope}
	\draw (6) ellipse (9mm and 3mm);
	\draw (0,-.7) ellipse (10mm and 3.5mm);
	\draw (6') ellipse (9mm and 3mm);
	\draw (3,-.7) ellipse (10mm and 3.5mm);
\node at (-1,0) {\vphantom{x}};
\end{tikzpicture}
\end{minipage}
\end{itemize}
The number of vertices (\resp\ edges) of $\colenc{\somegraph}$ is then the number of vertices (\resp\ edges) of $\somegraph$ plus the number of edges $\someedge$ of $\somegraph$ contained in at least one cycle and such that $\coloring(\someedge,\_)$ has not the same value for both endpoints of $\someedge$.


\begin{figure}
\centering
\begin{tikzpicture}[baseline=(3),rotate=90]
\begin{scope}[every node/.style={circle,thick,draw,inner sep=0,minimum size=10}]
	\node[fill=teal] (1) at (2,0) {};
	\coordinate (2) at (2.25,0.5);
	\node (3) at (2.5,1) {};
	\node[fill=teal] (4) at (3.5,1) {};
	\node (5) at (2.5,-1) {};
	\coordinate (6) at (3.5,0);
	\node (7) at (3.5,-1) {};
	\coordinate (8) at (2,1);
	\node (9) at (2,2) {};
	\node (10) at (2,3) {};
	\node (11) at (3.5,2) {};
	\coordinate (12) at (2.75,3);
	\node (13) at (3.5,3) {};
	\coordinate (14) at (2,4);
	\node (15) at (2,5) {};
	\coordinate (16) at (3.5,4);
	\node[fill=teal] (17) at (3.5,5) {};
\end{scope}
\begin{genscope}[red][-]
	\path (1) edge (2);
	\path (3) edge (5);
	\path (6) edge (7);
	\path (9) edge (11);
	\path (12) edge (13);
	\path (14) edge (15);
\end{genscope}
\begin{genscope}[blue][-,densely dashed]
	\path (5) edge (7);
	\path (1) edge (8);
	\path (10) edge (14);
	\path (11) edge (13);
	\path (15) edge (17);
	\path (16) edge (17);
\end{genscope}
\begin{genscope}[olive][-,densely dash dot]
	\path (2) edge (3);
	\path (3) edge (4);
	\path (4) edge (6);
	\path (8) edge (9);
	\path (9) edge (10);
	\path (10) edge (12);
	\path (13) edge (16);
\end{genscope}
\end{tikzpicture}
\hfill$\mapsto$\hfill
\begin{tikzpicture}[baseline=(3),rotate=90]
\begin{scope}[every node/.style={circle,thick,draw,inner sep=0,minimum size=10}]
	\node[fill=teal] (1) at (2,0) {};
	\node (3) at (2.5,1) {};
	\node[fill=teal] (4) at (3.5,1) {};
	\node (5) at (2.5,-1) {};
	\node (7) at (3.5,-1) {};
	\node (9) at (2,2) {};
	\node (10) at (2,3) {};
	\node (11) at (3.5,2) {};
	\node (13) at (3.5,3) {};
	\node (15) at (2,5) {};
	\node[fill=teal] (17) at (3.5,5) {};
\end{scope}
\begin{scope}[every node/.style={rectangle,thick,draw,inner sep=0,minimum size=10}]
	\node (6) at (3.5,0) {};
	\node (12) at (2.75,3) {};
	\node (14) at (2,4) {};
	\node (16) at (3.5,4) {};
\end{scope}
\begin{genscope}[red][-]
	\path (3) edge (5);
	\path (6) edge (7);
	\path (9) edge (11);
	\path (12) edge (13);
	\path (14) edge (15);
\end{genscope}
\begin{genscope}[blue][-,densely dashed]
	\path (5) edge (7);
	\path (10) edge (14);
	\path (11) edge (13);
	\path (15) edge (17);
	\path (16) edge (17);
\end{genscope}
\begin{genscope}[olive][-,densely dash dot]
	\path (3) edge (4);
	\path (4) edge (6);
	\path (9) edge (10);
	\path (10) edge (12);
	\path (13) edge (16);
\end{genscope}
\begin{genscope}[violet][-,densely dotted]
	\path (1) edge (3);
	\path (1) edge (9);
\end{genscope}
\end{tikzpicture}
\vspace*{1em}\hrule\vspace*{1em}
\begin{tikzpicture}[baseline=(3),rotate=90]
\begin{scope}[every node/.style={circle,thick,draw,inner sep=0,minimum size=10}]
	\node[fill=teal] (1) at (2,0) {};
	\coordinate (2) at (2.25,0.5);
	\node (3) at (2.5,1) {};
	\node[fill=teal] (4) at (3.5,1) {};
	\node (5) at (2.5,-1) {};
	\coordinate (6) at (3.5,0);
	\node (7) at (3.5,-1) {};
	\coordinate (8) at (2,1);
	\node (9) at (2,2) {};
	\node[fill=teal] (10) at (2,3) {};
	\node (11) at (3.5,2) {};
	\coordinate (12) at (2.75,3);
	\node (13) at (3.5,3) {};
	\coordinate (14) at (2,4);
	\node[fill=teal] (15) at (2,5) {};
\end{scope}
\begin{genscope}[red][-]
	\path (1) edge (2);
	\path (3) edge (5);
	\path (6) edge (7);
	\path (9) edge (11);
	\path (12) edge (13);
	\path (14) edge (15);
\end{genscope}
\begin{genscope}[blue][-,densely dashed]
	\path (5) edge (7);
	\path (1) edge (8);
	\path (10) edge (14);
	\path (11) edge (13);
\end{genscope}
\begin{genscope}[olive][-,densely dash dot]
	\path (2) edge (3);
	\path (3) edge (4);
	\path (4) edge (6);
	\path (8) edge (9);
	\path (9) edge (10);
	\path (10) edge (12);
\end{genscope}
\end{tikzpicture}
\hfill$\mapsto$\hfill
\begin{tikzpicture}[baseline=(3),rotate=90]
\begin{scope}[every node/.style={circle,thick,draw,inner sep=0,minimum size=10}]
	\node[fill=teal] (1) at (2,0) {};
	\node (3) at (2.5,1) {};
	\node[fill=teal] (4) at (3.5,1) {};
	\node (5) at (2.5,-1) {};
	\node (7) at (3.5,-1) {};
	\node (9) at (2,2) {};
	\node[fill=teal] (10) at (2,3) {};
	\node (11) at (3.5,2) {};
	\node (13) at (3.5,3) {};
	\node[fill=teal] (15) at (2,5) {};
\end{scope}
\begin{scope}[every node/.style={rectangle,thick,draw,inner sep=0,minimum size=10}]
	\node (6) at (3.5,0) {};
	\node (12) at (2.75,3) {};
\end{scope}
\begin{genscope}[red][-]
	\path (3) edge (5);
	\path (6) edge (7);
	\path (9) edge (11);
	\path (12) edge (13);
\end{genscope}
\begin{genscope}[blue][-,densely dashed]
	\path (5) edge (7);
	\path (11) edge (13);
\end{genscope}
\begin{genscope}[olive][-,densely dash dot]
	\path (3) edge (4);
	\path (4) edge (6);
	\path (9) edge (10);
	\path (10) edge (12);
\end{genscope}
\begin{genscope}[violet][-,densely dotted]
	\path (1) edge (3);
	\path (1) edge (9);
	\path (10) edge (15);
\end{genscope}
\end{tikzpicture}
\caption{Two examples of encoding of local coloring as edge-coloring, where \textcolor{teal}{filled} vertices are splitting ones and square vertices represent added ones}
\label{fig:example_encoding_yeo}
\end{figure}
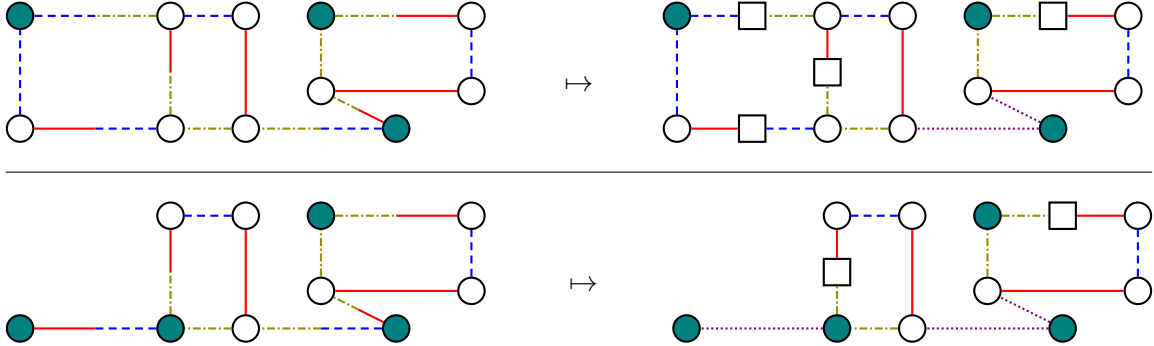

Two examples of this encoding are given on \cref{fig:example_encoding_yeo}.
The key properties of this encoding are that:
\begin{itemize}
\item
alternating/cusp-free cycles in the obtained graph $\colenc{\somegraph}$ correspond to cusp-free cycles of $\somegraph$, and in particular $\colenc{\somegraph}$ has no cusp-free cycle if and only if $\somegraph$ has none;
\item
a vertex of $\somegraph$ is splitting in $\somegraph$ if and only if the corresponding one is splitting in $\colenc{\somegraph}$;
\item
no added vertex is splitting in $\colenc{\somegraph}$.
\end{itemize}
Using these properties, one easily deduces \cref{th:LocalYeo} for a graph $\somegraph$ with a local coloring $\coloring$ from \cref{th:yeo} applied to $\colenc{\somegraph}$ and $\colenc{\coloring}$.

\begin{rem}
The encoding $\colenc{(\_)}$ is not local, meaning it is not stable by sub-graph.
This is because $\colenc{(\_)}$ considers the presence (or absence) of cycles, which is a global criterion depending on more than the neighborhoods of the endpoints of the considered edge.
For instance, consider the locally colored graphs $\somegraph$ and $\somegraph'$ respectively on the upper and lower parts of \cref{fig:example_encoding_yeo}, with their respective encodings $\colenc{\somegraph}$ and $\colenc{\somegraph'}$.
Observe that, while $\somegraph'$ is a sub-graph of $\somegraph$, $\colenc{\somegraph'}$ is not a sub-graph of $\colenc{\somegraph}$: the unique \textcolor{red}{solid}-\textcolor{blue}{dashed} edge of $\somegraph'$ is not in any cycle, while the corresponding edge in $\somegraph$ is.

A local encoding, stable by sub-graph, seems hard to come by.
In particular, an idea that cannot work is adding a same ``gadget'' graph in the middle of each edge (or of each ``bicolored'' edge) so as to duplicate each edge and to color them correspondingly -- whether this gadget is simply a single vertex or a more complex graph.
Indeed, the gadget to add must not have any cusp-free cycle so as to be able to apply \cref{th:yeo}, nor should it have any splitting vertex as one wants to find a splitting vertex in the original graph.
Such a graph cannot exist by \cref{th:yeo} itself!
\end{rem}

\subsection{Variants of Yeo's Theorem}
\label{sec:graph_comparison:variants}


It is known that Yeo's theorem is equivalent to various other graph-theoretical results (\ie\ they can be proved from one another).
In particular, Szeider~\cite{Szeider04c} exhibited four such alternative statements.
One of them is Kotzig's theorem, proved equivalent to the sequentialization of unit-free multiplicative proof nets with $\mix$~\cite{pngraph}.
We will also consider the generalization of Yeo's theorem to $\colorgraph$-coloring from~\cite{hcycles}.

In~\cite{Szeider04c} are given non-trivial encodings of graphs into graphs such that applying one theorem on an encoding allows to prove another theorem on the initial graph.
We show here that \cref{th:ParamLocalYeo} provides a natural unifying principle subsuming all these results (Theorems~\ref{th:yeo},~\ref{th:kotzig},~\ref{th:seymourandgiles},~\ref{th:grossmanandhaggkvist},~\ref{th:shoesmithandsmiley} and~\ref{th:HYeo}).
Indeed, we prove each of these results by applying \cref{th:ParamLocalYeo} to a well-chosen local coloring of the graph with no modification of its structure (vertices and edges), giving somehow ``encoding-less'' proofs.
Besides, this implies that our proof of \cref{th:ParamLocalYeo} \emph{via} cusp minimization is also a proof of each of these results, just by adapting the definition of a cusp.

A \definitive{perfect matching}, or $1$-factor, of a graph $\somegraph$ is a set of edges $F$ such that every vertex has a unique edge in $F$ incident to it.
It is well known that a perfect matching $F$ in a graph $\somegraph$ is unique if and only if $\somegraph$ contains no \definitive{$F$-alternating cycle}, which is a cycle whose edges are alternatively in and out of $F$, including the last and first ones (it is \eg\ a simple variant of~\cite[Theorem~1]{berge57} which considers $F$-alternating open paths).
A \definitive{bridge} is an edge whose removal increases the number of connected components of the graph; equivalently, it is an edge not in a cycle.
\begin{thm}[Kotzig~\cite{Kotzig1959}]\label{th:kotzig}
If a (non-empty) graph $\somegraph$ has a unique perfect matching $F$, then $\somegraph$ has a bridge which belongs to $F$.
\end{thm}
\begin{proof}
It suffices to define an edge-coloring $\coloring$ of $\somegraph$ into $\{0, 1\}$ by $\coloring(\someedge) = 1$ iff $\someedge \in F$.
Then $F$-alternating cycles are exactly cusp-free cycles, so by \cref{th:LocalYeo} (here even \cref{th:yeo} would suffice) there is a splitting vertex $\somevertex$.
The unique edge of $F$ incident to $\somevertex$ is a bridge as it is the only edge of endpoint $\somevertex$ with color $1$.
\end{proof}

\begin{thmC}[\cite{Seymour1978}]\label{th:seymourandgiles}
Consider a (non-empty) graph $\somegraph$ and a function $\phi$ from its vertices to its edges such that, for all vertex $\somevertex$, $\phi(\somevertex)$ is incident to $\somevertex$.
If $\somegraph$ has no cycle $\somecycle$ satisfying $\phi(\somevertex)\in\somecycle$ for every $\somevertex\in\somecycle$ -- such a cycle is called \definitive{$\phi$-conformal} -- then there exists a vertex $\othervertex$ such that $\phi(\othervertex)$ is a bridge.
\end{thmC}
\begin{proof}
Set a local coloring into $\{0, 1\}$ by $\coloring(\someedge,\somevertex) = 1$ iff $\someedge = \phi(\somevertex)$.
With this local coloring, $\phi$-conformal cycles of $\somegraph$ are exactly cusp-free cycles, so \cref{th:LocalYeo} gives a splitting vertex $\somevertex$: $\phi(\somevertex)$ is a bridge as it is the only edge of endpoint $\somevertex$ with color $1$.
\end{proof}

\begin{rem}
Observe that an edge-coloring $\coloring$ cannot prove \cref{th:seymourandgiles} without changing the structure of the graph: consider the graph drawn on \cref{fig:coloring_seymour}.
To have the equivalence between $\phi$-conformal cycles and cusp-free cycles, \eg\ considering the cycle $\somecycle \eqdef (\othervertex,\otherline,\otherothervertex,\someline,\somevertex,\phi(\othervertex),\othervertex)$, one would need
$\coloring(\otheredge) = \coloring(\someedge)$ (looking at $\otherothervertex$), $\coloring(\someedge) = \coloring(\phi(\othervertex))$ (looking at $\somevertex$) and $\coloring(\phi(\othervertex)) \neq \coloring(\otheredge)$ (looking at $\othervertex$), thence $\coloring(\otheredge) = \coloring(\someedge) =  \coloring(\phi(\othervertex)) \neq \coloring(\otheredge)$, absurd.
\end{rem}

\begin{figure}
\begin{minipage}[b]{.5\textwidth}
\centering
\begin{tikzpicture}
\begin{scope}[every node/.style={circle,minimum size=6mm,inner sep=0,draw=black,line width=0.8pt}, every edge/.style={draw=black,-,line width=0.8pt}]
	\node (v) at (0,0) {$\somevertex$};
	\node (u) at (-1,.8) {$\othervertex$};
	\node (x) at (1,.8) {$\otherothervertex$};
	\node (y) at (4,.8) {$\otherotherothervertex$};
	\node (w) at (3,0) {$y$};

	\path (v) edge[-] node[arete nommee]{$\phi(\othervertex)$} (u);
	\path (v) edge[-] node[arete nommee,swap]{$\someedge$} (x);
	\path (u) edge[-] node[arete nommee]{$\otheredge$} (x);
	\path (x) edge[-] node[arete nommee]{$\phi(\otherothervertex) = \phi(\otherotherothervertex)$} (y);
	\path (w) edge[-] node[arete nommee]{$\phi(\somevertex) = \phi(y)$} (v);
\end{scope}
\end{tikzpicture}
\caption{No edge-coloring for \cref{th:seymourandgiles}}
\label{fig:coloring_seymour}
\end{minipage}%
\begin{minipage}[b]{.5\textwidth}
\centering
\begin{tikzpicture}
\begin{scope}[every node/.style={circle,minimum size=6mm,inner sep=0,draw=black,line width=0.8pt}, every edge/.style={draw=black,-,line width=0.8pt}]
	\node (v) at (0,0) {$\somevertex$};
	\node (u) at (1,.8) {$\othervertex$};
	\node (x) at (2,0) {$\otherothervertex$};

	\path (v) edge[->] node[arete nommee]{$\someedge$} (u);
	\path (x) edge[<-] node[arete nommee]{$\otheredge$} (v);
	\path (u) edge[->] node[arete nommee]{$\otherotheredge$} (x);
\end{scope}
\end{tikzpicture}
\caption{No edge-coloring for \cref{th:shoesmithandsmiley}}
\label{fig:coloring_shoesmith}
\end{minipage}
\end{figure}

\begin{thmC}[\cite{Grossman1983}]\label{th:grossmanandhaggkvist}
Any (non-empty) 2-edge-colored graph has a splitting vertex or an alternating cycle.
\end{thmC}
\begin{proof}
This is just the particular case of \cref{th:yeo} restricted to two colors.
\end{proof}

The next theorem considers \emph{undirected} paths in \emph{directed} graphs.
A \definitive{directed graph} is the same as a (total) graph defined in \cref{sec:graph_def},
except that, instead of a single incidence function,
we have two functions giving the
\definitive{source} \(\source(\someedge)\) and
\definitive{target} \(\target(\someedge)\)
of each edge $\someedge$, requiring \(\source(\someedge)\not=\target(\someedge)\).
The underlying graph is obtained by defining the incidence function \(\incidence\),
setting \(\incidence(\someedge)=\{\source(\someedge),\target(\someedge)\}\)
(which is always a pair):
the sets of vertices and edges remain the same, only the notion of incidence is relaxed.
We do not consider directed paths here (we will do so, briefly, in \cref{sec:mllpn:ps}):
a path in a directed graph is just a path in the underlying graph.
The only role of directedness lies in the definition of turning vertices:
a vertex $\somevertex$ of a cycle $\somecycle$ is a \definitive{turning vertex} of $\somecycle$
if the edges incident to $\somevertex$ in $\somecycle$ are either both of source $\somevertex$
or both of target $\somevertex$.
\begin{thm}[Shoesmith and Smiley~\cite{thseqsemicycle}]\label{th:shoesmithandsmiley}
If a non-empty set $S$ of vertices of a directed graph $\somegraph$ contains a turning vertex of each cycle of $\somegraph$,
then $S$ contains a vertex which is a turning vertex of every cycle it belongs to.
\end{thm}
\begin{proof}
Forgetting the orientation of the graph, we want a local coloring whose cusp-points are exactly the turning vertices in $S$.
Such a coloring can be obtained through setting $\coloring(\someedge,\somevertex) = 0$ if $\somevertex \in S$ is the source of $\someedge$, $\coloring(\someedge,\somevertex) = 1$ if $\somevertex \in S$ is the target of $\someedge$ and $\coloring(\someedge,\somevertex) = \someedge$ otherwise.
Cycles with no turning vertex in $S$ are exactly cusp-free cycles, so \cref{th:ParamLocalYeo} with $\somesetedge \eqdef \{(\somevertex, \somecolor) \suchthat \somevertex \in S, \somecolor \in \{0,1\}\}$ yields a splitting vertex $\somevertex \in S$.
By definition of $\coloring$, $\somevertex$ is a turning vertex of every cycle it belongs to.
\end{proof}

We need the parametrized version of our result to deal in a simple way with the parameter $S$.
Here again, an edge-coloring $\coloring$ is not enough for proving \cref{th:shoesmithandsmiley} without changing the structure of the graph: look at \cref{fig:coloring_shoesmith} with all vertices in $S$.
To have the equivalence between cycles without turning vertex and cusp-free cycles, one needs $\coloring(\otheredge) = \coloring(\otherotheredge) \neq  \coloring(\someedge) = \coloring(\otheredge)$.

\begin{rem}
Shoesmith and Smiley's stated and proved \cref{th:shoesmithandsmiley} to handle a particular kind of proofs represented as graphs~\cite{multipleconclusionlogic}, sharing striking similarities with proof nets of multiplicative linear logic (notably, forbidding some classes of cycles).\footnote{
We were not aware of this work during the research leading to the present paper: it only came to our attention \emph{via} Szeider's equivalence results~\cite{Szeider04c}.
As far as we know, 48 years after the publication of~\cite{multipleconclusionlogic} and 39 years after the publication of~\cite{ll}, the first line of work has been ignored by the linear logic community: it would certainly be of interest to investigate further connexions with proof nets.
}
Moreover, \cref{th:shoesmithandsmiley} can be used directly to obtain a splitting $\parr$ in a proof net by instantiating $S$ as the set of all $\parr$-vertices.
Furthermore, Shoesmith and Smiley's proof of this theorem is quite similar to our proof by cusp minimization:
the key idea of both proofs is to look at cycles with a minimal number of cusps (or turning vertices).
Still, there are important differences:
we construct an explicit order relation on vertex-color pairs, while their proof builds an infinite path to reach a contradiction;
besides, the association of colors with vertices in our parameter makes our result more modular.
This is particularly relevant for proof nets:
\cref{th:shoesmithandsmiley} seems limited to giving a splitting $\parr$, without the
unifying character of \cref{th:ParamLocalYeo} seen in \cref{sec:paramyeo}.
\end{rem}


\Cref{th:LocalYeo} implies another generalization of Yeo's theorem to $\colorgraph$-colored graphs~\cite{hcycles}.
Given a graph $\colorgraph$, an \definitive{$\colorgraph$-coloring} of a graph $\somegraph$ is an edge-coloring of $\somegraph$ with as colors the vertices of $\colorgraph$.
An \definitive{$\colorgraph$-cycle} is a cycle in $\somegraph$ where the colors of consecutive edges (including the last and first ones) are linked by an edge in $\colorgraph$.
When $\colorgraph$ is a complete graph, we recover the standard edge-coloring and $\colorgraph$-cycles correspond to alternating cycles.
A \definitive{complete multipartite} graph $\othersubgraph$ has vertices $\someindepset_1\uplus\dotsc\uplus\someindepset_k$ (disjoint union) where each $\someindepset_i$ is an independent set of vertices (no edge in $\othersubgraph$ between vertices of $\someindepset_i$) and if $\somevertex\in\someindepset_i$ and $\othervertex\in\someindepset_j$ (with $i\neq j$) then there is exactly one edge between them in $\othersubgraph$.

\begin{defi}
Given a graph $\somegraph$ with an $\colorgraph$-coloring $\coloring$, and $\somevertex$ a vertex of $\somegraph$, \definitive{$\somegraph_\somevertex$} is the graph with vertices the edges of $\somegraph$ incident to $\somevertex$, and one edge between $\someedge$ and $\otheredge$ if and only if their colors $\coloring(\someedge)$ and $\coloring(\otheredge)$ are linked by an edge in $\colorgraph$.
\end{defi}

Note that $\somegraph_\somevertex$ only depends on the neighbourhood of $\somevertex$ (the edges incident to $\somevertex$) and on the sub-graph of $\colorgraph$ induced by the colors of these edges.

\begin{thmC}[{\cite[Theorem~2]{hcycles}}]\label{th:HYeo}
Consider $\colorgraph$ a graph and $\somegraph$ a non-empty $\colorgraph$-colored graph.
Assume $\somegraph$ has no $\colorgraph$-cycle
and that, for every vertex $\somevertex$ of $\somegraph$, $\somegraph_\somevertex$ is a complete multipartite graph.
Then there exists a vertex $\somevertex$ of $\somegraph$ such that every connected component $D$ of $\somegraph-\somevertex$ satisfies that the set of edges of $\somegraph$ between $\somevertex$ and vertices of $D$ is an independent set in $\somegraph_\somevertex$.
\end{thmC}
\begin{proof}
Define a local coloring $\coloring$ by $\coloring(\someedge,\somevertex)$ is the independent set in $\somegraph_{\somevertex}$ to which $\someedge$ belongs.
For $\somegraph$ has no $\colorgraph$-cycle, it has no cusp-free cycle for $\coloring$, and the result follows by \cref{th:LocalYeo}.
\end{proof}

As $\somegraph_\somevertex$ is a complete multipartite graph, one can consider its independent sets of vertices as corresponding to a given color, thus defining a local coloring.
Local colorings seem more natural than the complete multipartite structure of some induced sub-graphs of $\colorgraph$.
Note that Theorems~\ref{th:LocalYeo},~\ref{th:yeo} and~\ref{th:HYeo} are all equivalent.
As far as we know, this theorem was not known to be equivalent to Yeo's theorem before the conference version of the present paper.

\cref{th:HYeo} is actually a slight reformulation of~\cite[Theorem~2]{hcycles} as its authors require $\colorgraph$ to have at most one edge between two of its vertices, and $\colorgraph$ and $\somegraph$ to have no isolated vertices.
These modifications clearly have no impact on the theorem.


\section{Multiplicative Proof Nets}
\label{sec:mllpn}

\subsection{Unit-Free Multiplicative Linear Logic with Mix}

We focus on unit-free multiplicative linear logic whose formulas are given by:
\begin{equation*}
  A ::= X \mid X\orth \mid A\tensor A \mid A\parr A
\end{equation*}
The \definitive{dual} operator $(\_)\orth$ is extended to an involution on all formulas by De Morgan duality:
$(X\orth)\orth=X$, $(A\tensor B)\orth=A\orth\parr B\orth$ and $(A\parr B)\orth=A\orth\tensor B\orth$.

We will in fact consider \definitive{localized formulas},
which are obtained by labeling each syntactic construct
with a unique tag (its \definitive{location})
from some fixed infinite denumerable set.
Formally:
\[
  A ::= X_e \mid X_e\orth \mid A\tensor_e A \mid A\parr_e A
\]
where the tags \(e\) are locations, and we require that no location is used twice in a formula:
each localized formula is called an \definitive{instance} of the underlying (untagged) formula.
We then define a \definitive{sequent} as a set of localized formulas, again without repeated tag.
This provides a clear notion of \definitive{occurrence} of a formula \(A\)
(an instance of \(A\), which occurs as a sub-formula of a given (localized) formula
-- the latter possibly being an element of a sequent)
while avoiding the need for an explicit exchange rule.
Moreover it will make the correspondence with proof nets more direct:
in this we follow previous approaches, \eg,~\cite{mllpnpspace}.
Keeping in line with the more traditional presentation,
we will most often denote a sequent \(\Gamma\) as an enumeration of its elements,
moreover keeping locations implicit:
\eg, we may write \(A\orth,A\) for a set of two localized formulas
(with disjoint sets of tags)
whose underlying formulas are dual to each other;
and we may write \(\Gamma,\Delta\) for the union of two sequents
(again implicitly requiring that their sets of locations are disjoint).

We consider the deduction system \MLLhmix\ of \emph{open} derivations in multiplicative linear logic with $\mix$ rules.
This consists of the usual set of rules for classical multiplicative linear logic:
\[
    \begin{prooftree}
      \infer0[\ax]{\vdash A\orth, A}
    \end{prooftree}
    \quad
    \begin{prooftree}
      \hypo{\vdash A, \Gamma}
      \hypo{\vdash A\orth, \Delta}
      \infer2[\cut]{\vdash \Gamma, \Delta}
    \end{prooftree}
    \quad
    \begin{prooftree}
      \hypo{\vdash A, \Gamma}
      \hypo{\vdash B, \Delta}
      \infer2[\tensor]{\vdash A\tensor B, \Gamma, \Delta}
    \end{prooftree}
    \quad
    \begin{prooftree}
      \hypo{\vdash A, B, \Gamma}
      \infer1[\parr]{\vdash A\parr B, \Gamma}
    \end{prooftree}
\]
together with the two \(\mix\) rules,
and an additional rule introducing any single-formula sequent:
\[
    \begin{prooftree}
      \hypo{\vdash \Gamma}
      \hypo{\vdash \Delta}
      \infer2[\mix_2]{\vdash \Gamma, \Delta}
    \end{prooftree}
    \qquad
    \begin{prooftree}
      \infer0[\mix_0]{\vdash}
    \end{prooftree}
    \qquad
    \begin{prooftree}
      \infer0[\hyp]{\vdash A}
    \end{prooftree}\,.
\]

The rôle of the \((\hyp)\) rule is to allow for open derivations,
in which some formulas are left unproved.
If $\someproof$ is a derivation with hypotheses $\vdash A_1, \dotsc, \vdash A_n$ and conclusion $\vdash B_1,\dotsc,B_k$, we call $\someproof$ \definitive{a derivation of} $A_1,\dotsc,A_n \vdash B_1,\dotsc,B_k$.
Note that we restrict ourselves to hypotheses consisting of single formulas rather than arbitrary sequents:
this allows us to define a notion of substitution of a proof for an hypothesis whose counterpart in proof structures always preserves correctness
(see \cref{rem:ps} and \cref{lem:deseqsub} below).
If $\someproof_1$ is a derivation of $\Sigma\vdash\Gamma,A$ and $\someproof_2$ is a derivation of $A, \Theta\vdash\Delta$
(moreover assuming that there is no repeated location in \(\Sigma,\Gamma,\Theta,\Delta\)),
the \definitive{substitution} of $\someproof_1$ in $\someproof_2$ is a derivation of $\Sigma,\Theta\vdash\Gamma,\Delta$:
it is obtained from $\someproof_2$ by replacing the ($\hyp$) rule on $\vdash A$ with $\someproof_1$
(this adds $\Gamma$ to all sequents of $\someproof_2$ below $\vdash A$).

We also consider the following rewriting of derivations which we call \definitive{\mixretore\ reduction} (due to its similarity to Rétoré's reduction on the exponential connective $\wn$~\cite[page~77]{phddanos}, with contraction and weakening forming a monoid):
\[
  \begin{prooftree}[center=true]
    \hypo{\vdash \Gamma}
    \infer0[\mix_0]{\vdash}
    \infer2[\mix_2]{\vdash \Gamma}
  \end{prooftree}
  \quad\rightsquigarrow\quad
  \begin{prooftree}
    \hypo{\vdash \Gamma}
  \end{prooftree}
  \qquad\qquad
  \begin{prooftree}[center=true]
    \infer0[\mix_0]{\vdash}
    \hypo{\vdash \Gamma}
    \infer2[\mix_2]{\vdash \Gamma}
  \end{prooftree}
  \quad\rightsquigarrow\quad
  \begin{prooftree}
    \hypo{\vdash \Gamma}
  \end{prooftree}
\]
It defines a confluent and strongly normalizing rewriting system on derivations.

\begin{lem}[\Mixretore\ Normal Forms]\label{lem:mixretorenf}
  If $\someproof$ is a derivation from \MLLhmix\ in \mixretore\ normal form, either it is
  \begin{prooftree}
    \infer0[\mix_0]{\vdash}
  \end{prooftree},
or it does not contain the $(\mix_0)$ rule.
\end{lem}
\begin{proof}
Observe that $(\mix_2)$ is the only rule accepting an empty sequent as a premise.
\end{proof}

\subsection{Proof Structures}
\label{sec:mllpn:ps}

A \definitive{proof structure} is a directed acyclic partial graph,
together with particular labelings of vertices and of edges,
subject to extra conditions that we detail below.

A \definitive{directed partial graph} is the same as a partial graph as defined
\cref{sec:graph_def}, except that, instead of a single incidence function,
we have two \emph{partial} functions $\source$ and $\target$ giving the \definitive{source}
\(\source(\someedge)\) and \definitive{target} \(\target(\someedge)\) of each
edge $\someedge$.
The \definitive{underlying partial graph}
is obtained by defining the incidence function \(\incidence\) as follows:
\(\incidence(\someedge)\) is the set of (at most \(2\)) vertices 
containing \(\source(\someedge)\), if it is defined,
and \(\target(\someedge)\), if it is defined.
This definition generalizes the one we used in \cref{sec:graph_comparison:variants}
to the setting of partial graphs:
again, the sets of vertices and edges are unchanged,
only the notion of incidence is relaxed.
A \definitive{path} in a directed partial graph is a path 
in the underlying partial graph.
Such a path 
$(\somevertex_0, \someedge_1, \somevertex_1, \someedge_2, \somevertex_2, \dots, \someedge_n, \somevertex_n)$
is said to be \definitive{directed} when,
for all $i \in \{1, \dots, n\}$,
\(\source(\someedge_i)=\somevertex_{i-1}\) and \(\target(\someedge_i)=\somevertex_i\).
A \definitive{directed acyclic partial graph} is a
directed partial graph without any directed cycle.

In a proof structure, vertices are labeled with names of rules,
$\ax$, $\cut$, $\tensor$ or $\parr$,
and then named according to their label: \definitive{$\ax$-vertices},
\definitive{$\cut$-vertices},
\definitive{$\tensor$-vertices} and \definitive{$\parr$-vertices}.
A \definitive{premise} (resp.\ a \definitive{conclusion})
\definitive{of a proof structure}
is any edge without source (resp.\ without target);
and a \definitive{premise} (resp.\ a \definitive{conclusion}) 
\definitive{of a vertex} is any edge with target (resp.\ source) this vertex.
Incidences are moreover subject to constraints:
\begin{itemize}
\item each $\ax$-vertex has two conclusions (and no premise);
\item each $\cut$-vertex has two premises (and no conclusion);
\item each $\tensor$-vertex and each \(\parr\)-vertex has two premises
  and one conclusion.
\end{itemize}
As a direct consequence of the definition,
for each pair \((\someedge,\otheredge)\) of edges in a proof structure,
there is at most one directed path from \(\target(\someedge)\)
to \(\source(\otheredge)\).
An \definitive{initial edge} is a premise of the proof structure,
or a conclusion of an \(\ax\)-vertex
(\emph{i.e.}\ an edge whose source, if any, is not the target of another edge).
A \definitive{terminal vertex} is one whose conclusions are also conclusions of
the proof structure (\emph{i.e.}\ not premises of other vertices).
A $\cut$-vertex is always terminal.

Additionally, edges are labeled with localized formulas:
the label of an edge is called its \definitive{type}.
Typing is subject to the following local constraints
(depicted in \cref{fig:mllpn:ps:typing}):
\begin{itemize}
  \item the conclusions of an \(\ax\)-vertex
    (resp.\ the premises of a \(\cut\)-vertex)
    must have dual types;
  \item the conclusion of a \(\tensor\)-vertex (resp.\ \(\parr\)-vertex)
    must have type \(A\tensor B\) (resp.\ \(A\parr B\)),
    where \(A\) and \(B\) are the types of its premises;
\end{itemize}
as well as to the following global constraint on tags:
\begin{itemize}
  \item the sets of locations of initial edges are pairwise disjoint.
\end{itemize}
The \definitive{conclusion sequent} (resp.\ \definitive{premise sequent})
of a proof structure is then the set of types of its conclusions (resp.\ premises)
-- where, indeed, no location is repeated, thanks to typing constraints
and our previous observation on directed paths.
A proof structure is said to be \definitive{closed} if it has no premise.
A \definitive{connected component} of a proof structure is the proof structure
induced by a connected component of the underlying partial graph.

\begin{figure}
\[
\begin{tikzpicture}[baseline=(baseline)]
\begin{scope}[every node/.style={circle,minimum size=6mm,inner sep=0,draw=black,line width=0.8pt}, every edge/.style={draw=black,-,line width=0.8pt}]
	\node (v) at (0,0) {$\ax$};
	\coordinate (1) at (-.75,-.6) {};
	\coordinate (2) at (.75,-.6) {};
        \path (v) edge[->] node[arete nommee,above,pos=.7] {$A\orth$} (1);
	\path (v) edge[->] node[arete nommee,above,pos=.7] {$A$} (2);
\end{scope}
\node (baseline) at (0,0) {\vphantom{x}};
\end{tikzpicture}
\qquad
\begin{tikzpicture}[baseline=(baseline)]
\begin{scope}[every node/.style={circle,minimum size=6mm,inner sep=0,draw=black,line width=0.8pt}, every edge/.style={draw=black,-,line width=0.8pt}]
	\coordinate (a) at (-.75,0.5) {};
	\coordinate (b) at (.75,0.5) {};
	\node (v) at (0,0) {$\cut$};
        \path (a) edge[->] node[arete nommee,above,pos=.6] {$A$} (v);
	\path (b) edge[->] node[arete nommee,above,pos=.6] {$A\orth$} (v);
\end{scope}
\node (baseline) at (0,0.3) {\vphantom{x}};
\end{tikzpicture}
\qquad
\begin{tikzpicture}[baseline=(baseline)]
\begin{scope}[every node/.style={circle,minimum size=6mm,inner sep=0,draw=black,line width=0.8pt}, every edge/.style={draw=black,-,line width=0.8pt}]
	\coordinate (a) at (-.75,0.5) {};
	\coordinate (b) at (.75,0.5) {};
	\node (v) at (0,0) {$\tensor$};
	\coordinate (ab) at (0,-.8) {};
        \path (a) edge[->] node[arete nommee,above,pos=.6] {$A$} (v);
	\path (b) edge[->] node[arete nommee,above,pos=.6] {$B$} (v);
	\path (v) edge[->] node[arete nommee] {$A\tensor B$} (ab);
\end{scope}
\node (baseline) at (0,0.3) {\vphantom{x}};
\end{tikzpicture}
\qquad
\begin{tikzpicture}[baseline=(baseline)]
\begin{scope}[every node/.style={circle,minimum size=6mm,inner sep=0,draw=black,line width=0.8pt}, every edge/.style={draw=black,-,line width=0.8pt}]
	\coordinate (a) at (-.75,0.5) {};
	\coordinate (b) at (.75,0.5) {};
	\node (v) at (0,0) {$\parr$};
	\coordinate (ab) at (0,-.8) {};
	\path (a) edge[->] node[arete nommee,above,pos=.6] {$A$} (v);
	\path (b) edge[->] node[arete nommee,above,pos=.6] {$B$} (v);
	\path (v) edge[->] node[arete nommee] {$A\parr B$} (ab);
\end{scope}
\node (baseline) at (0,0.3) {\vphantom{x}};
\end{tikzpicture}
\]
\caption{Typing constraints in proof structures}
\label{fig:mllpn:ps:typing}
\end{figure}
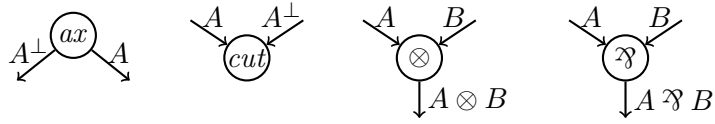

Throughout our discussion of proof structures,
directedness plays no other rôle than enabling the previous definitions:
from now on, whenever we mention a path in a proof structure,
this is to be taken in the underlying partial graph.
Moreover, when depicting proof structures,
we generally leave the orientation of edges implicit,
by drawing them from top to bottom.
An example of closed proof structure, with conclusion sequent
$(X\orth\parr X)\tensor Y, Y\orth$,
is given on \cref{fig:ps_ex}.
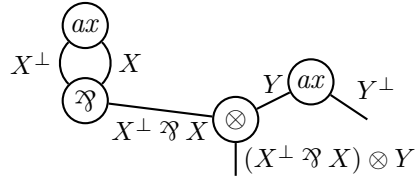
\begin{figure}
\centering
\begin{tikzpicture}
\begin{scope}[every node/.style={circle,minimum size=6mm,inner sep=0,draw=black,line width=0.8pt}, every edge/.style={draw=black,-,line width=0.8pt}]
	\node (ax1) at (-1,2) {$\ax$};
	\node (p) at (-1,1) {$\parr$};
	\node (t) at (1,0.75) {$\tensor$};
	\node (ax2) at (2,1.25) {$\ax$};
	\coordinate (c1) at (1,0);
	\coordinate (c2) at (2.75,.75);
\end{scope}
\begin{genscope}
	\path[out=-45,in=45] (ax1) edge[-] node[arete nommee,right] {\small$X$} (p);
	\path[out=-135,in=135] (ax1) edge[-] node[arete nommee,left] {\small$X\orth$} (p);
	\path (p) edge[-] node[arete nommee,below] {\small$X\orth\parr X$} (t);
	\path (ax2) edge[-] node[arete nommee,above] {\small$Y$} (t);
	\path (t) edge[-] node[arete nommee] {\small$(X\orth\parr X)\tensor Y$} (c1);
	\path (ax2) edge[-] node [arete nommee]{\small$Y\orth$} (c2);
\end{genscope}
\end{tikzpicture}
\caption{Example of proof structure}
\label{fig:ps_ex}
\end{figure}

\begin{rem}\label{rem:ps}
There are many ways to define proof structures.
In the literature, one generally considers closed proof structures only:
in the typed multiplicative case considered here, it is easy to check 
that the notion of (closed) proof structure is essentially equivalent
to others in the literature (\eg~\cite{ll}),
or deviate in ways that are not relevant for sequentialization
-- for instance, some definitions (\eg~\cite{pngraph})
impose no typing of edges \emph{a priori}.
Open proof structures are sometimes called \emph{modules}~\cite{structmult}.
Here we avoid this name because modules are generally 
intended to be assembled via some form of vertical
composition, where the conclusions of a module are identified
(one-to-one) with the premises of another:
in general this might introduce cycles, whereas we will only consider the much
simpler operation mimicking the substitution of a proof for an hypothesis,
which amounts to gluing two proof structures along a single edge
(see \cref{lem:deseqsub}).
\end{rem}

To identify proof structures corresponding to proofs, and create a distinction between $\tensor$- and $\parr$-vertices, it is usual to ask for a proof structure to respect a \emph{correctness criterion}.
As explained in the introduction, we use one due to Danos and Regnier~\cite{structmult}.
A path in a proof structure is called \definitive{switching} when it does not contain the two premises of any $\parr$-vertex.
A proof structure is \definitive{DR-correct}, and is called a \definitive{proof net}, if it has no switching cycle.

\begin{rem}\label{rem:correcgraph}
The original definition of the acyclicity condition in the Danos-Regnier correctness criterion~\cite{structmult} (extended to ($\mix_2$) in~\cite{mixpn}) is in fact slightly different.
They consider \definitive{correctness graphs}:
partial graphs obtained by disconnecting one of the two premises of each $\parr$-vertex
(changing its target to be undefined).
A proof structure is correct when all its $2^{n}$ correctness graphs -- where $n$ is the number of $\parr$-vertices of $\someps$ -- are acyclic (and connected in the original work without the $\mix$ rules).
This condition is equivalent to the fact that any cycle in the proof structure must contain the two premises of some $\parr$-vertex
(\ie\ no cycle is \emph{feasible} in the sense of~\cite{mixpn}).
This is also equivalent to the apparently weaker condition that any cycle in the proof structure must go through the two premises of some $\parr$-vertex \emph{consecutively}:
\end{rem}

\begin{lem}[Local-Global Principle]\label{lem:localglobal}
A simple path that never goes through the two premises of a $\parr$-vertex consecutively (including as last and first edges for a cycle) is a switching path.
\end{lem}
\begin{proof}
If the two premises $\someedge_1$ and $\someedge_2$ of a $\parr$-vertex $\somevertex$ occur in a simple path $\somepath$, then:
either $\somevertex$ occurs exactly once in $\somepath$, and $\someedge_1$ and $\someedge_2$ must appear one right before $\somevertex$ and the other right after $\somevertex$ in $\somepath$;
or $\somepath$ is a cycle, and $\someedge_1$ and $\someedge_2$ are its first and last edges.
\end{proof}

Given some DR-correct proof structure $\someps$, its \definitive{DR-connectedness degree} $\cdeg(\someps)$ is the number of connected components of any of its correctness graphs, as defined in \cref{rem:correcgraph}.
Note that, thanks to acyclicity, $\cdeg(\someps)$ does not depend on the choice of the correctness graph.
We say $\someps$ is \definitive{connected} if $\cdeg(\someps)=1$
(in particular it is not empty):
note that this implies that \(\someps\) consists of a single connected component,
but the converse implication does not hold (for instance,
the DR-connectedness degree of a proof net consisting of a single \(\parr\)-vertex,
together with its premises and conclusion, is 2).

\subsection{Desequentialization}\label{sec:desequentialization}

We define, by induction on a derivation $\someproof$ of $A_1,\dotsc,A_n\vdash B_1,\dotsc,B_k$, its \definitive{desequentialization} $\deseq(\someproof)$ which is a DR-correct proof structure with hypotheses labeled $A_1, \dotsc A_n$ and conclusions labeled $B_1, \dotsc, B_k$ (we also say that ``$\deseq(\someproof)$ is a DR-correct proof structure on the sequent with hypotheses $A_1,\dotsc,A_n\vdash B_1,\dotsc,B_k$'').
\begin{itemize}
\item
\begin{minipage}[t]{0.5\linewidth}
If $\someproof$ is reduced to an ($\ax$) rule with conclusion $\vdash A\orth, A$, then $\deseq(\someproof)$ is the proof structure with one $\ax$-vertex $\somevertex$ and two conclusions labeled $A\orth$ and $A$, both with source $\somevertex$.
\end{minipage}
\hfill
\begin{minipage}[t]{0.49\linewidth}
\begin{tikzpicture}[baseline=(baseline)]
	\node at (.2,0.5) {$
		\begin{prooftree}
		\infer0[\ax]{\vdash A\orth, A}
		\end{prooftree}
		$};
	\node at (1.9,0.5) {$\mapsto$};
\begin{scope}[every node/.style={circle,minimum size=6mm,inner sep=0,draw=black,line width=0.8pt}, every edge/.style={draw=black,-,line width=0.8pt}]
	\node (v) at (3.5,.6) {$\ax$};
	\coordinate (1) at (2.75,0) {};
	\coordinate (2) at (4.25,0) {};
	\path (v) edge node[arete nommee,above] {$A\orth$} (1);
	\path (v) edge node[arete nommee,above] {$A$} (2);
\end{scope}
\node (baseline) at (-1.37,1.4) {\vphantom{x}};
\end{tikzpicture}
\end{minipage}
\item
\begin{minipage}[t]{0.5\linewidth}
If the last rule of $\someproof$ is a ($\cut$) rule applied to two derivations $\someproof_1$ and $\someproof_2$ then $\deseq(\someproof)$ is obtained from the disjoint union of $\deseq(\someproof_1)$ and $\deseq(\someproof_2)$ by adding a new $\cut$-vertex $\somevertex$.
The conclusions of $\deseq(\someproof_1)$ and $\deseq(\someproof_2)$ labeled by the principal formulas $A$ and $A\orth$ of the ($\cut$) rule now have $\somevertex$ as target.
\end{minipage}
\hfill
\begin{minipage}[t]{0.49\linewidth}
\begin{tikzpicture}[baseline=(baseline)]
	\node at (.2,0) {$
    		\begin{prooftree}
      	\subproof{\pi_1}{\vdash A, \Gamma}
      	\subproof{\pi_2}{\vdash A\orth, \Delta}
      	\infer2[\cut]{\vdash \Gamma, \Delta}
    		\end{prooftree}
		$};
	\node at (1.9,0.25) {$\mapsto$};
	\node (pi1) at (2.75,.75) {\color{red}$\deseq(\someproof_1)$};
	\draw[red, dashed] (2.75,.6) ellipse (5.2mm and 8mm);
	\node (pi2) at (4.05,.75) {\color{red}$\deseq(\someproof_2)$};
	\draw[red, dashed] (4.05,.6) ellipse (5.2mm and 8mm);
\begin{scope}[every node/.style={circle,minimum size=6mm,inner sep=0,draw=black,line width=0.8pt}, every edge/.style={draw=black,-,line width=0.8pt}]
	\coordinate (a) at (2.65,0.1) {};
	\coordinate (b) at (4.15,0.1) {};
	\node (v) at (3.4,-.5) {$\cut$};
	\path (a) edge node[arete nommee,above] {$A$} (v);
	\path (b) edge node[arete nommee,above] {$A\orth$} (v);
\end{scope}
\node (baseline) at (-1.37,1.5) {\vphantom{x}};
\end{tikzpicture}
\end{minipage}
\item
\begin{minipage}[t]{0.5\linewidth}
If the last rule of $\someproof$ is a ($\tensor$) rule applied to two derivations $\someproof_1$ and $\someproof_2$ then $\deseq(\someproof)$ is obtained from the disjoint union of $\deseq(\someproof_1)$ and $\deseq(\someproof_2)$ by adding a new $\tensor$-vertex $\somevertex$.
The conclusions of $\deseq(\someproof_1)$ and $\deseq(\someproof_2)$ labeled by the principal formulas $A$ and $B$ of the ($\tensor$) rule now have $\somevertex$ as target, and we add a new conclusion edge, labeled $A\tensor B$, with source $\somevertex$.
\end{minipage}
\hfill
\begin{minipage}[t]{0.49\linewidth}
\begin{tikzpicture}[baseline=(baseline)]
	\node at (.2,0) {$
    		\begin{prooftree}
      	\subproof{\pi_1}{\vdash A, \Gamma}
      	\subproof{\pi_2}{\vdash B, \Delta}
      	\infer2[\tensor]{\vdash A\tensor B, \Gamma, \Delta}
    		\end{prooftree}
		$};
	\node at (1.9,0.25) {$\mapsto$};
	\node (pi1) at (2.75,.75) {\color{red}$\deseq(\someproof_1)$};
	\draw[red, dashed] (2.75,.6) ellipse (5.2mm and 8mm);
	\node (pi2) at (4.05,.75) {\color{red}$\deseq(\someproof_2)$};
	\draw[red, dashed] (4.05,.6) ellipse (5.2mm and 8mm);
\begin{scope}[every node/.style={circle,minimum size=6mm,inner sep=0,draw=black,line width=0.8pt}, every edge/.style={draw=black,-,line width=0.8pt}]
	\coordinate (a) at (2.65,0.1) {};
	\coordinate (b) at (4.15,0.1) {};
	\node (v) at (3.4,-.5) {$\tensor$};
	\coordinate (ab) at (3.4,-1.3) {};
	\path (a) edge node[arete nommee,above] {$A$} (v);
	\path (b) edge node[arete nommee,above] {$B$} (v);
	\path (v) edge node[arete nommee] {$A\tensor B$} (ab);
\end{scope}
\node (baseline) at (-1.37,1.5) {\vphantom{x}};
\end{tikzpicture}
\end{minipage}
\item
\begin{minipage}[t]{0.5\linewidth}
If the last rule of $\someproof$ is a ($\parr$) rule applied to a derivation $\someproof_1$ then $\deseq(\someproof)$ is obtained from $\deseq(\someproof_1)$ by adding a new $\parr$-vertex $\somevertex$.
The conclusions of $\deseq(\someproof_1)$ labeled by the principal formulas $A$ and $B$ of the ($\parr$) rule now have $\somevertex$ as an additional endpoint, and we add a new edge, labeled $A\parr B$, with source $\somevertex$.
\end{minipage}
\hfill
\begin{minipage}[t]{0.49\linewidth}
\begin{tikzpicture}[baseline=(baseline)]
	\node at (.2,0) {$
    		\begin{prooftree}
      	\subproof{\pi_1}{\vdash A, B, \Gamma}
      	\infer1[\parr]{\vdash A\parr B, \Gamma}
    		\end{prooftree}
		$};
	\node at (1.9,0) {$\mapsto$};
	\node (pi1) at (3.4,.7) {\color{red}$\deseq(\someproof_1)$};
	\draw[red, dashed] (3.4,.55) ellipse (11mm and 7mm);
\begin{scope}[every node/.style={circle,minimum size=6mm,inner sep=0,draw=black,line width=0.8pt}, every edge/.style={draw=black,-,line width=0.8pt}]
	\coordinate (a) at (2.65,0.1) {};
	\coordinate (b) at (4.15,0.1) {};
	\node (v) at (3.4,-.5) {$\parr$};
	\coordinate (ab) at (3.4,-1.3) {};
	\path (a) edge node[arete nommee,above] {$A$} (v);
	\path (b) edge node[arete nommee,above] {$B$} (v);
	\path (v) edge node[arete nommee] {$A\parr B$} (ab);
\end{scope}
\node (baseline) at (-1.37,1.2) {\vphantom{x}};
\end{tikzpicture}
\end{minipage}
\item
\begin{minipage}[t]{0.5\linewidth}
If the last rule of $\someproof$ is a ($\mix_2$) rule applied to two derivations $\someproof_1$ and $\someproof_2$ then $\deseq(\someproof)$ is the disjoint union of $\deseq(\someproof_1)$ and $\deseq(\someproof_2)$.
\end{minipage}
\hfill
\begin{minipage}[t]{0.49\linewidth}
\begin{tikzpicture}[baseline=(baseline)]
	\node at (.2,0) {$
    		\begin{prooftree}
      	\subproof{\pi_1}{\vdash \Gamma}
      	\subproof{\pi_2}{\vdash \Delta}
      	\infer2[\mix_2]{\vdash \Gamma, \Delta}
    		\end{prooftree}
		$};
	\node at (1.9,0) {$\mapsto$};
	\node (pi1) at (2.8,0) {\color{red}$\deseq(\someproof_1)$};
	\draw[red, dashed] (2.8,0) ellipse (5.5mm and 7mm);
	\node (pi2) at (4.1,0) {\color{red}$\deseq(\someproof_2)$};
	\draw[red, dashed] (4.1,0) ellipse (5.5mm and 7mm);
\node (baseline) at (-1.37,.4) {\vphantom{x}};
\end{tikzpicture}
\end{minipage}
\item
\begin{minipage}[t]{0.5\linewidth}
If $\someproof$ is reduced to a ($\mix_0$) rule, $\deseq(\someproof)$ is the empty graph (no vertex, no edge).
\end{minipage}
\hfill
\begin{minipage}[t]{0.49\linewidth}
\begin{tikzpicture}[baseline=(baseline)]
	\node at (.2,0) {$
    		\begin{prooftree}
      	\infer0[\mix_0]{\vdash}
    		\end{prooftree}
		$};
	\node at (1.9,0) {$\mapsto$};
\node (baseline) at (-1.37,.1) {\vphantom{x}};
\end{tikzpicture}
\end{minipage}
\item
\begin{minipage}[t]{0.5\linewidth}
If $\someproof$ is reduced to a ($\hyp$) rule on $\vdash A$, then $\deseq(\someproof)$ is the proof structure with no vertex and a single edge with no endpoint, labeled $A$.
\end{minipage}
\hfill
\begin{minipage}[t]{0.49\linewidth}
\begin{tikzpicture}[baseline=(baseline)]
	\node at (.2,0.5) {$
		\begin{prooftree}
		\infer0[\hyp]{\vdash A}
		\end{prooftree}
		$};
	\node at (1.9,0.5) {$\mapsto$};
\begin{scope}[every node/.style={circle,minimum size=6mm,inner sep=0,draw=black,line width=0.8pt}, every edge/.style={draw=black,-,line width=0.8pt}]
	\coordinate (1) at (3.5,.8) {};
	\coordinate (2) at (3.5,0) {};
	\path (1) edge node[arete nommee] {$A$} (2);
\end{scope}
\node (baseline) at (-1.37,.8) {\vphantom{x}};
\end{tikzpicture}
\end{minipage}
\end{itemize}
There is a bijection between the ($\ax$), ($\cut$), ($\tensor$) and ($\parr$) rules of $\someproof$ and the vertices of $\deseq(\someproof)$;
and there is a bijection between the formulas of \(\someproof\) (more precisely, the union of all sequents in \(\someproof\))
and the edges of \(\deseq(\someproof)\).

It should be clear from the definitions that the desequentialization of a proof is indeed a proof net.
The following two results also follow straightforwardly.

\noindent
\begin{minipage}[t]{0.6\linewidth}
\begin{lem}[Desequentialization of a substitution]\label{lem:deseqsub}
If $\someproof$ is the substitution of a derivation $\someproof_1$ for a hypothesis $A$ in a derivation $\someproof_2$, then $\deseq(\someproof)$ is obtained from the disjoint union of $\deseq(\someproof_1)$ and $\deseq(\someproof_2)$ by identifying the conclusion $\someedge$ of $\deseq(\someproof_1)$ labeled $A$ with the hypothesis $\someedge'$ of $\deseq(\someproof_2)$ labeled $A$.
The obtained edge has label $A$, source \(\source(\someedge)\), and target \(\target(\someedge')\).
\end{lem}
\end{minipage}
\hfill
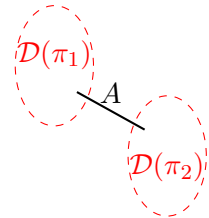
\begin{minipage}[t]{0.38\linewidth}
\begin{tikzpicture}[baseline=(baseline)]
	\node (pi1) at (0,.75) {\color{red}$\deseq(\someproof_1)$};
	\draw[red, dashed] (0,.6) ellipse (5.2mm and 8mm);
	\node (pi2) at (1.5,-.75) {\color{red}$\deseq(\someproof_2)$};
	\draw[red, dashed] (1.5,-.6) ellipse (5.2mm and 8mm);
\begin{scope}[every node/.style={circle,minimum size=6mm,inner sep=0,draw=black,line width=0.8pt}, every edge/.style={draw=black,-,line width=0.8pt}]
	\coordinate (a) at (.3,0.25) {};
	\coordinate (b) at (1.2,-0.25) {};
	\path (a) edge node[arete nommee,above] {$A$} (b);
\end{scope}
\node (baseline) at (-1,1.5) {\vphantom{x}};
\end{tikzpicture}
\end{minipage}

\begin{lem}[Desequentialization and \(\mix\)-rules]
  \label{lem:deseq:mix}
  If $\someproof_2$ is obtained from $\someproof_1$ by a \mixretore\ reduction
  then $\deseq(\someproof_1)\giso\deseq(\someproof_2)$.
  Moreover, for any proof $\someproof$, $\cdeg(\deseq(\someproof))=1 + \#\mix_2 - \#\mix_0$
  (where $\#\mix_i$ is the number of ($\mix_i$) rules in $\someproof$).
  In particular, proofs without $\mix$ have a connected desequentialization.
\end{lem}

\section{Sequentialization}
\label{sec:seq}

A key result of the theory of proof nets is the fact that desequentialization is surjective,
in the following sense:
\begin{thm}[Sequentialization]\label{thm:seq}
  Given a proof net $\someps$, there exists a derivation $\someproof$ in \MLLhmix\ such that $\someps\giso\deseq(\someproof)$; $\someproof$ is called a \definitive{sequentialization} of $\someps$.
\end{thm}
There are in fact many variants of this result,
depending on whether, \eg,
we consider \(\cut\)-free proof nets (and then obtain \(\cut\)-free proofs),
or we restrict to closed proof nets (and then sequentialize to closed proofs),
or we require connectedness (and then drop the \(\mix\) rules).
The literature on the subject is quite rich, see~\cite{ll,structmult,phddanos,mixpn} for some of the earliest approaches.
This section is dedicated to showing how such sequentialization results
can be deduced from our parametrized and local version of Yeo’s theorem (\cref{th:ParamLocalYeo}),
or even directly from cusp minimization (in the form of \cref{lem:bjumping_Yeo}),
in a uniform and modular way.

\subsection{Splitting vertices}\label{sec:seq:splitting}

Most, if not all, proofs of sequentialization share a common pattern,
where one reasons inductively on the size of proof structures,
and shows that, given a proof net \(\rho\):
either \(\rho\) is obviously the translation of a proof
(\eg, it is empty, or it is reduced to a single edge, 
or it is reduced to an \(\ax\)-vertex with its two conclusions);
or \(\rho\) can be split into smaller proof structures,
all of them still DR-correct, in such a way that 
a sequentialization of \(\rho\) can be obtained
by glueing together sequentializations of those smaller proof nets.
Essentially, this amounts to show that any proof net falls
into one of the seven cases of the definition of desequentialization in
\cref{sec:desequentialization}, with the following difference:
in the inductive definition of \(\deseq(\someproof)\),
new vertices are always introduced in terminal position,
whereas it can sometimes be useful to split a proof net
at some internal vertex.

This leads us to the following definition:
\begin{defi}\label{def:splitting}
  We say a vertex \(\somevertex\) in a proof structure \(\someps\) is
  \definitive{splitting}\footnote{
    The apparent conflict of terminology with \cref{sec:localcolor},
    where we gave another notion of splitting vertex, is only temporary.
    We will soon resolve it by introducing an appropriate local coloring on
    proof structures.
  } when:
  \begin{itemize}
    \item \(\somevertex\) is an \(\ax\)-vertex, a \(\tensor\)-vertex, or a \(\cut\)-vertex which is not in any cycle of \(\someps\);
    \item or \(\somevertex\) is a \(\parr\)-vertex, whose conclusion edge is not in any cycle of \(\someps\).
  \end{itemize}
\end{defi}

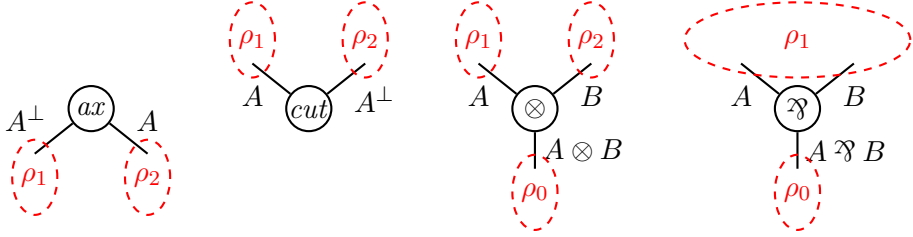
\begin{figure}
  \[
  \begin{tikzpicture}[
    every node/.style={circle,minimum size=6mm,inner sep=0,draw=black,line width=0.8pt},
    every edge/.style={draw=black,-,line width=0.8pt},
    baseline=(v),
  ]
          \node (v) at (3.5,.6) {$\ax$};
          \coordinate (1) at (2.75,0) {};
          \coordinate (2) at (4.25,0) {};
          \path (v) edge node[arete nommee,swap] {$A\orth$} (1);
          \path (v) edge node[arete nommee] {$A$} (2);
          \node[below=-2mm,subps] at (1) {$\someps_1$};
          \node[below=-2mm,subps] at (2) {$\someps_2$};
  \end{tikzpicture}
  \qquad
  \begin{tikzpicture}[
    every node/.style={circle,minimum size=6mm,inner sep=0,draw=black,line width=0.8pt},
    every edge/.style={draw=black,-,line width=0.8pt},
    baseline=(v),
  ]
          \coordinate (a) at (2.65,0.1) {};
          \coordinate (b) at (4.15,0.1) {};
          \node (v) at (3.4,-.5) {$\cut$};
          \path (a) edge node[arete nommee,swap] {$A$} (v);
          \path (b) edge node[arete nommee] {$A\orth$} (v);
          \node[above=-2mm,subps] at (a) {$\someps_1$};
          \node[above=-2mm,subps] at (b) {$\someps_2$};
  \end{tikzpicture}
  \qquad
  \begin{tikzpicture}[
    every node/.style={circle,minimum size=6mm,inner sep=0,draw=black,line width=0.8pt},
    every edge/.style={draw=black,-,line width=0.8pt},
    baseline=(v),
  ]
          \coordinate (a) at (2.65,0.1) {};
          \coordinate (b) at (4.15,0.1) {};
          \node (v) at (3.4,-.5) {$\tensor$};
          \coordinate (ab) at (3.4,-1.3) {};
          \path (a) edge node[arete nommee,swap] {$A$} (v);
          \path (b) edge node[arete nommee] {$B$} (v);
          \path (v) edge node[arete nommee] {$A\tensor B$} (ab);
          \node[above=-2mm,subps] at (a) {$\someps_1$};
          \node[above=-2mm,subps] at (b) {$\someps_2$};
          \node[below=-2mm,subps] at (ab) {$\someps_0$};
  \end{tikzpicture}
  \qquad
  \begin{tikzpicture}[
    every node/.style={circle,minimum size=6mm,inner sep=0,draw=black,line width=0.8pt},
    every edge/.style={draw=black,-,line width=0.8pt},
    baseline=(v),
  ]
          \coordinate (a) at (2.65,0.1) {};
          \coordinate (b) at (4.15,0.1) {};
          \node (v) at (3.4,-.5) {$\parr$};
          \coordinate (ab) at (3.4,-1.3) {};
          \path (a) edge node[arete nommee,swap] {$A$} (v);
          \path (b) edge node[arete nommee] {$B$} (v);
          \path (v) edge node[arete nommee] {$A\parr B$} (ab);
          \node[above=-2mm,subps,minimum width=3cm] at (v|-a) {$\someps_1$};
          \node[below=-2mm,subps] at (ab) {$\someps_0$};
  \end{tikzpicture}
  \]
  \caption{Shape of the connected component of each kind of splitting vertex}
  \label{fig:splitting}
\end{figure}

In other words, (the proof structure induced by) the connected component of
\(\somevertex\) in \(\someps\)
can be decomposed uniquely following one of the four shapes
represented in \cref{fig:splitting}.
Sequentialization then boils down to the following key result:

\begin{lem}[Existence of a splitting vertex]\label{lem:splittingnode}
  In any proof net containing at least one vertex,
  there exists a splitting vertex.
\end{lem}
Indeed, the proof of the sequentialization theorem readily follows:
\begin{proof}[Proof of \cref{thm:seq}]
We use an induction on the number of vertices and edges of $\someps$.
And we actually impose that \(\someproof\) is in \mixretore\ normal form.

If \(\someps\) is empty, then \(\someproof\) is reduced to a (\(\mix_0\)) rule.
If \(\someps\) has more than one connected component,
then it is the disjoint union of non-empty (and necessarily smaller)
proof structures \(\someps_1\) and \(\someps_2\).
Both \(\someps_1\) and \(\someps_2\) are still DR-correct,
and the induction hypothesis yields proofs
\(\someproof_1\) and \(\someproof_2\)
such that $\deseq(\someproof_1)\giso\someps_1$
and $\deseq(\someproof_2)\giso\someps_2$.
Then we can write the conclusion sequent of 
\(\someproof\) as \(\Gamma_1,\Gamma_2\)
so that the conclusion of \(\someproof_i\) is \(\Gamma_i\)
for \(1\le i\le 2\).
We define
\[
  \someproof\quad\eqdef\quad
  \begin{prooftree}[center]
    \hypo{\someproof_1}
    \infer[no rule]1{\vdash \Gamma_1}
    \hypo{\someproof_2}
    \infer[no rule]1{\vdash \Gamma_2}
    \infer2[\mix_2]{\vdash \Gamma_1,\Gamma_2}
  \end{prooftree}
\]
and obtain $\deseq(\someproof)\giso\someps$ directly.

If \(\someps\) is reduced to an edge \(\someedge\) of type \(A\),
then set \(\someproof\) as an \((\hyp)\) rule on \(A\).
Otherwise, \(\someps\) is reduced to a connected component with at least one
vertex, hence a splitting one $\somevertex$ by \cref{lem:splittingnode}.

Assume $\somevertex$ is a $\parr$-vertex.
By removing $\somevertex$
(with premises labeled $A$ and $B$ and conclusion labeled $A\parr B$)
from $\someps$ (so that the target of the premises of \(\somevertex\),
and the source of its conclusion, are no longer defined),
we obtain proof structures $\someps_1$, with $A$ and $B$ as labels of some of its conclusions,
and \(\someps_0\), with \(A\parr B\) as label of one of its premises,
as in \cref{fig:splitting}.
By induction hypothesis,
one gets proofs $\someproof_1$ such that $\deseq(\someproof_1)\giso\someps_1$,
and $\someproof_0$ such that $\deseq(\someproof_0)\giso\someps_0$.
Then we can write the conclusion (resp.\ premise) sequent of 
\(\someproof\) as \(\Gamma_1,\Gamma_0\) (resp.\ \(\Sigma_1,\Sigma_0\))
so that \(\someproof_1\) is a derivation of \(\Sigma_1\vdash\Gamma_1,A,B\)
and \(\someproof_0\) is a derivation of \(A\parr B,\Sigma_0\vdash\Gamma_0\).
We define \(\someproof\) as the substitution of
\[\begin{prooftree}
  \hypo{\someproof_1}
  \infer[no rule]1{\vdash \Gamma_1,A,B}
  \infer1[\parr]{\vdash \Gamma_1,A\parr B}
\end{prooftree}\]
for the \((\hyp)\) rule on \(A\parr B\) in \(\someproof_0\),
and obtain $\deseq(\someproof)\giso\someps$ by \cref{lem:deseqsub}.

Assume $\somevertex$ is a $\tensor$-vertex.
By removing $\somevertex$
(with premises labeled $A$ and $B$, and conclusion labeled $A\tensor B$)
from $\someps$,
we obtain proof structures $\someps_1$ (resp.\ \(\someps_2\)),
with $A$ (resp.\ \(B\)) as label of one of its conclusions,
and \(\someps_0\), with \(A\tensor B\) as label of one of its premises,
as in \cref{fig:splitting}.
By induction hypothesis,
one gets proofs
$\someproof_1$ such that $\deseq(\someproof_1)\giso\someps_1$,
$\someproof_2$ such that $\deseq(\someproof_2)\giso\someps_2$,
and $\someproof_0$ such that $\deseq(\someproof_0)\giso\someps_0$.
Then we can write the conclusion (resp.\ premise) sequent of 
\(\someproof\) as \(\Gamma_1,\Gamma_2,\Gamma_0\) (resp.\ \(\Sigma_1,\Sigma_2,\Sigma_0\))
so that 
\(\someproof_1\) is a derivation of \(\Sigma_1\vdash\Gamma_1,A\),
\(\someproof_2\) is a derivation of \(\Sigma_2\vdash\Gamma_2,B\),
and \(\someproof_0\) is a derivation of \(A\tensor B,\Sigma_0\vdash\Gamma_0\).
We define \(\someproof\) as the substitution of
\[\begin{prooftree}
  \hypo{\someproof_1}
  \infer[no rule]1{\vdash \Gamma_1,A}
  \hypo{\someproof_2}
  \infer[no rule]1{\vdash \Gamma_2,B}
  \infer2[\tensor]{\vdash \Gamma_1,\Gamma_2,A\tensor B}
\end{prooftree}\]
for the \((\hyp)\) rule on \(A\tensor B\) in \(\someproof_0\),
and obtain $\deseq(\someproof)\giso\someps$ by \cref{lem:deseqsub}.

Assume $\somevertex$ is a $\cut$-vertex.
By removing $\somevertex$
(with premises labeled $A$ and $A\orth$)
from $\someps$,
we obtain proof structures $\someps_1$ (resp.\ \(\someps_2\)),
with $A$ (resp.\ \(A\orth\)) as label of some of its conclusions,
as in \cref{fig:splitting}.
By induction hypothesis,
one gets proofs
$\someproof_1$ such that $\deseq(\someproof_1)\giso\someps_1$ and
$\someproof_2$ such that $\deseq(\someproof_2)\giso\someps_2$.
Then we can write the conclusion (resp.\ premise) sequent of 
\(\someproof\) as \(\Gamma_1,\Gamma_2\) (resp.\ \(\Sigma_1,\Sigma_2\))
so that 
\(\someproof_1\) is a derivation of \(\Sigma_1\vdash\Gamma_1,A\), and
\(\someproof_2\) is a derivation of \(\Sigma_2\vdash\Gamma_2,A\orth\).
We define
\[
  \someproof\quad\eqdef\quad
  \begin{prooftree}[center]
  \hypo{\someproof_1}
  \infer[no rule]1{\vdash \Gamma_1,A}
  \hypo{\someproof_2}
  \infer[no rule]1{\vdash \Gamma_2,A\orth}
  \infer2[\cut]{\vdash \Gamma_1,\Gamma_2}
\end{prooftree}\]
and obtain $\deseq(\someproof)\giso\someps$ directly.

Assume $\somevertex$ is an $\ax$-vertex.
By removing $\somevertex$ (of conclusions labeled $A$ and $A\orth$),
we obtain proof structures $\someps_1$ (resp.\ \(\someps_2\)),
with $A\orth$ (resp.\ \(A\)) as label of some of its premises,
as in \cref{fig:splitting}.
By induction hypothesis,
one gets proofs
$\someproof_1$ such that $\deseq(\someproof_1)\giso\someps_1$ and
$\someproof_2$ such that $\deseq(\someproof_2)\giso\someps_2$.
Then we can write the conclusion (resp.\ premise) sequent of 
\(\someproof\) as \(\Gamma_1,\Gamma_2\) (resp.\ \(\Sigma_1,\Sigma_2\))
so that 
\(\someproof_1\) is a derivation of \(\Sigma_1,A\orth\vdash\Gamma_1\), and
\(\someproof_2\) is a derivation of \(\Sigma_2,A\vdash\Gamma_2\).
The substitution of the \((\ax)\) rule 
\begin{prooftree}
  \infer0[\ax]{\vdash A\orth, A}
\end{prooftree}
for the (\(\hyp\)) rule on \(A\orth\) in \(\someproof_1\)
yields a derivation of \(\Sigma_1\vdash\Gamma_1, A\), and
\(\someproof\) is obtained by substituting the latter 
for the  (\(\hyp\)) rule on \(A\) in \(\someproof_2\):
we obtain $\deseq(\someproof)\giso\someps$ by applying
\cref{lem:deseqsub} twice.
\end{proof}

The various approaches to sequentialization then essentially differ only by the way
\cref{lem:splittingnode} is proved:
the sequentialization process induced by \cref{thm:seq} is dictated by
the strategy one follows to find a splitting vertex.
We dedicate the remainder of \cref{sec:seq} to a series of proofs of
\cref{lem:splittingnode}.

The first one follows a variant of our method to establish Yeo’s theorem in
\cref{sec:genyeo}: in \cref{sec:seq:dummies},
we define an order on \(\parr\)-vertices based on cusp-free paths,
and show that a maximal \(\parr\)-vertex is splitting;
and then we apply easy, well-known results on proof nets
to treat the case of \(\parr\)-free proof nets.
This proof can be read without referring to 
\cref{sec:genyeo,sec:graph_comparison},
and it gives us the occasion to provide a fully developed,
yet hopefully accessible proof of sequentialization for 
the debuting linear logician.

The second one derives \cref{lem:splittingnode} directly from
\cref{th:ParamLocalYeo}: in \cref{sec:seq:yeo},
we show in particular how one can recover many classical
strategies to find a splitting vertex just by tuning
the parameter \(\somesetedge\).

In both cases,
the first step is to define a local coloring
of (the partial graph induced by) any proof structure.
\begin{defi}\label{def:wellcolored}
  We say a proof structure is \definitive{well-colored}
  when it is equipped with a local coloring 
  (of its partial graph) such that:
\begin{itemize}
\item for an $\ax$-vertex $\somevertex$ with conclusions $\someedge_1$ and $\someedge_2$,
  $\coloring(\someedge_1, \somevertex) \not= \coloring(\someedge_2, \somevertex)$;
\item for a $\cut$-vertex $\somevertex$ with premises $\someedge_1$ and $\someedge_2$,
  $\coloring(\someedge_1, \somevertex) \not= \coloring(\someedge_2, \somevertex)$;
\item for a $\tensor$-vertex $\somevertex$ with  premises $\someedge_1$ and $\someedge_2$
  and conclusion $\otheredge$,
  $\coloring(\someedge_1, \somevertex)$,
  $\coloring(\someedge_2, \somevertex)$
  and $\coloring(\otheredge, \somevertex)$
  are pairwise distinct;
\item for a $\parr$-vertex $\somevertex$ with premises $\someedge_1$ and $\someedge_2$
  and conclusion $\otheredge$,
  $\coloring(\someedge_1, \somevertex)
  =\coloring(\someedge_2, \somevertex)
  \not=\coloring(\otheredge, \somevertex)$.
\end{itemize}
\end{defi}

It is always possible to turn a proof structure into a well-colored proof
structure, with only three colors, say 
\textcolor{blue}{dashed}, \textcolor{red}{solid} and \textcolor{violet}{dotted}:
\begin{itemize}
  \item use \textcolor{red}{solid} and \textcolor{violet}{dotted}
    for the conclusions of each \(\ax\)-vertex, and
    for the premises of each \(\cut\)-vertex;
  \item use \textcolor{red}{solid} and \textcolor{violet}{dotted}
    for the premises of each \(\tensor\)-vertex,
    and \textcolor{blue}{dashed} for its conclusion;
  \item use \textcolor{red}{solid}
    for the premises of each \(\parr\)-vertex,
    and \textcolor{blue}{dashed} for its conclusion.
\end{itemize}
An example of well-colored proof structure following this convention
is given in \cref{fig:wellcolored}.

Note that, the cusp-points of a well-colored proof structure are exactly the pairs
$(\somevertex,\somecolor)$ where $\somevertex$ is a
$\parr$-vertex and \(\somecolor\) is the color associated with its premises.
This requires at least three colors, and can only be achieved with
a local coloring (see the $\ax$-vertex \(\otherothervertex_1\) in \cref{fig:wellcolored}).
Then, by \cref{lem:localglobal}, a cusp-free path in a proof structure is
nothing but a switching path;
and then a vertex is splitting in the sense of \cref{def:splitting}
if and only if it is splitting in the sense of \cref{sec:localcolor}.
This coincidence is visible in \cref{fig:wellcolored},
where \(\somevertex\), \(\othervertex\) and \(\otherothervertex_2\)
are splitting, and \(\otherothervertex_1\) is not.

\begin{figure}
\centering
\begin{tikzpicture}
\begin{scope}[every node/.style={circle,minimum size=6mm,inner sep=0,draw=black,line width=0.8pt}, every edge/.style={draw=black,-,line width=0.8pt}]
	\node [label=left:\(\otherothervertex_1\)] (ax1) at (-1,2.5) {$\ax$};
	\node [label=below left:\(\somevertex\)] (p) at (-1,1) {$\parr$};
	\node [label=above:\(\othervertex\)] (t) at (1,0.5) {$\tensor$};
	\node [label=above right:\(\otherothervertex_2\)] (ax2) at (2.5,1.5) {$\ax$};
	\coordinate (c1) at (1,-.5);
	\coordinate (c2) at (3.5,.75);
	\coordinate (ax1pl) at (-1.5,1.75);
	\coordinate (ax1pr) at (-.5,1.75);
	\coordinate (pt) at (0,0.75);
	\coordinate (ax2t) at (1.75,1);
\end{scope}
\begin{genscope}[red][-]
	\path[out=-135,in=90] (ax1) edge (ax1pl);
	\path (ax2) edge (c2);
	\path[out=-90,in=135] (ax1pl) edge (p);
	\path[out=-90,in=45] (ax1pr) edge (p);
	\path (pt) edge (t);
\end{genscope}
\begin{genscope}[blue][-,densely dashed]
	\path (p) edge (pt);
	\path (t) edge (c1);
\end{genscope}
\begin{genscope}[violet][-,densely dotted]
	\path[out=-45,in=90] (ax1) edge (ax1pr);
	\path (ax2) edge (ax2t);
	\path (ax2t) edge (t);
\end{genscope}
\end{tikzpicture}
\caption{An example of well-colored proof structure}
\label{fig:wellcolored}
\end{figure}
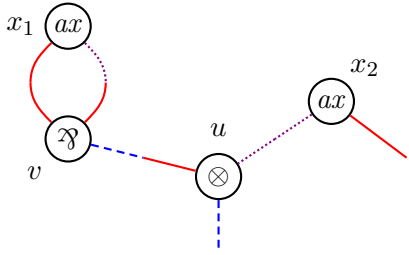

The following easy observations will also be useful:
\begin{lem}\label{lem:terminal:ax:par}
  Every terminal \(\ax\)- or \(\parr\)-vertex is splitting.
\end{lem}
\begin{proof}
  A conclusion of a terminal vertex cannot be part of a cycle.
\end{proof}
\begin{lem}\label{lem:par:free}
  In a \(\parr\)-free (connected component of a) proof net,
  every vertex is splitting.
\end{lem}
\begin{proof}
  A cycle without any \(\parr\)-vertex is always a switching cycle.
\end{proof}

\subsection{Sequentialization from cusp cycling}
\label{sec:seq:dummies}

A \definitive{\(\parr\)-path} in a proof structure is a path whose source and target are \(\parr\)-vertices,
and which starts with the conclusion of its source, and ends with a premise of its target.
The following result is then immediate:

\begin{lem}[Concatenation of \(\parr\)-paths]\label{lem:concatparpaths}
The concatenation of two \(\parr\)-paths (\resp\ cusp-free \(\parr\)-paths) is a \(\parr\)-path (\resp\ cusp-free \(\parr\)-path).
\end{lem}

Recall that $\mincycles{\somevertex}$ denotes the set of cycles of source \(\somevertex\),
without cusp at \(\somevertex\), and with a minimum number of cusps among 
such cycles.
The following is a particular case of \cref{lem:bjumping_Yeo}:

\begin{lem}[Cusp cycling in proof structures]\label{lem:bjumping}
  Let $\somevertex$ be a vertex and $\somecycle$ be a cycle in $\mincycles{\somevertex}$,
  with a cusp at $\somepier$.
  If there exists a simple cusp-free \(\parr\)-path $\otherpath$
  from $\somepier$ to a vertex of $\somecycle$,
  then there exists a switching cycle.
\end{lem}

\begin{defi}\label{def:order}
  We write $\somevertex\orderparr[\somepath]\othervertex$
  when \(\somepath\) is a simple open cusp-free \(\parr\)-path
  from $\somevertex$ to $\othervertex$
  such that there is no simple open cusp-free \(\parr\)-path starting from $\othervertex$,
  and ending on a vertex of $\somepath$.
  We write $\somevertex\orderparr\othervertex$ when there exists $\somepath$
  such that $\somevertex\orderparr[\somepath]\othervertex$.
\end{defi}

\begin{lem}\label{lem:order_is_order}
  In a proof net, $\orderparr$ is a strict partial order relation on $\parr$-vertices.
\end{lem}
\begin{proof}
If $\somevertex\orderparr[\somepath]\othervertex$, then $\somevertex\neq\othervertex$ as $\somepath$ is open, thus $\orderparr$ is irreflexive.
If $\somevertex\orderparr[\somepath]\othervertex\orderparr[\otherpath]\otherothervertex$,
then we prove $\somevertex\orderparr[\concatpath{\somepath}{\otherpath}]\otherothervertex$.

We first show that, except for its source, $\otherpath$ contains no vertex in $\somepath$.
Indeed, otherwise, consider the shortest open prefix $\otherpath'$ of \(\otherpath\)
whose target \(\otherothervertex'\) is in \(\somepath\).
Then
\(\somecycle\eqdef\concatpath{\subpath{\somepath}{\otherothervertex'}{\othervertex}}{\otherpath'}\)
is a cycle:
either \(\otherothervertex'=\othervertex\) and \(\somecycle=\otherpath'\);
or we can apply \cref{lem:concatsimplepaths}.
By DR-correctness, \(\somecycle\) must have a cusp:
by assumption, \(\somepath\) and \(\otherpath\) are cusp-free;
and there is no cusp at \(\otherothervertex\) since \(\otherpath\) is a \(\parr\)-path;
hence \(\somecycle\) must have a cusp at \(\otherothervertex'\).
Thus \(\otherpath'\) is a simple open cusp-free \(\parr\)-path,
which contradicts $\somevertex\orderparr[\somepath]\othervertex$.

Therefore $\concatpath{\somepath}{\otherpath}$ is a simple open cusp-free \(\parr\)-path
(\cref{lem:concatsimplepathsprefix,lem:concatparpaths}).
Consider $\otherotherpath$ a simple open cusp-free \(\parr\)-path with source $\otherothervertex$.
Since $\othervertex\orderparr[\otherpath]\otherothervertex$,
$\otherotherpath$ does not end on a vertex of $\otherpath$.
Thus, $\concatpath{\otherpath}{\otherotherpath}$ is a simple open cusp-free \(\parr\)-path
(\cref{lem:concatsimplepathsprefix,lem:concatparpaths}), with source $\othervertex$.
Since $\somevertex\orderparr[\somepath]\othervertex$, 
$\otherotherpath$ does not end on a vertex of $\somepath$ either.
\end{proof}

We then show how to recover one classical result on splitting vertices:
the existence of \emph{sections},
\ie~splitting $\parr$-vertices~\cite{phddanos} (\cref{lem:splittingpar}).

\begin{lem}\label{lem:ForOrderMax+Connect}
Let $\somevertex$ be a \(\parr\)-vertex in a proof net,
and consider a cycle $\somecycle\in\mincycles{\somevertex}$ that is moreover a \(\parr\)-path.
Then $\somevertex\orderparr\somepier$,
where \(\somepier\) is the vertex of the first cusp of \(\somecycle\).
\end{lem}

\begin{proof}
The prefix $\subpath{\somecycle}{\somevertex}{\somepier}$ is a simple open cusp-free \(\parr\)-path.
Moreover, by \cref{lem:bjumping} and DR-correctness, there is no simple cusp-free \(\parr\)-path from $\somepier$ to $\somecycle$.
\end{proof}

\begin{lem}[Splitting $\parr$]\label{lem:splittingpar}
  A proof net is $\parr$-free or contains a splitting $\parr$-vertex.
\end{lem}
\begin{proof}
  If a proof net contains a $\parr$-vertex,
  then its set of $\parr$-vertices is finite and non-empty,
  thus it contains a maximal element $\somevertex$ with respect to $\orderparr$ (\cref{lem:order_is_order}).
  If $\somevertex$ is not splitting,
  its conclusion belongs to a cycle and $\mincycles{\somevertex}\neq\emptyset$.
  Take some $\somecycle\in\mincycles{\somevertex}$,
  considered as starting by $\somevertex$ with its conclusion:
  \(\somecycle\) is a \(\parr\)-path.
  By DR-correctness, \(\somecycle\) contains at least one cusp:
  then \cref{lem:ForOrderMax+Connect} contradicts the maximality of \(\somevertex\).
\end{proof}
For instance, the proof structure of \cref{fig:wellcolored}
is easily checked to be a proof net,
and then its only \(\parr\)-vertex \(\somevertex\) must be splitting,
which is graphically obvious.

Added to \cref{lem:par:free}, \cref{lem:splittingpar} completes the proof of \cref{lem:splittingnode}.

\subsection{Sequentialization from Parametrized Local Yeo}
\label{sec:seq:yeo}

\Cref{th:ParamLocalYeo} gives a splitting vertex for any set $\somesetedge$ of vertex-color pairs
dominating the cusp-points of a well-colored proof net
-- \ie\ pairs $(\somevertex, \somecolor)$ with $\somevertex$ a $\parr$-vertex
and \(\somecolor\) the color associated with its two premises.
We review how natural choices for the parameter $\somesetedge$ in \cref{th:ParamLocalYeo}
yield various proofs of \cref{lem:splittingnode},
hence various strategies to select splitting vertices along the sequentialization procedure.
Each of these choices satisfies the hypothesis of \cref{th:ParamLocalYeo} trivially:
\(\somesetedge\) contains all cusp-points.

The most direct route is just to consider all vertex-color pairs:
\begin{cor}[Maximal pairs]
  \label{cor:seq:yeo:all}
  Take $\somesetedge$ the set of all vertex-color pairs of the proof net:
  for each $\orderyeo$-maximal element $(\somevertex, \somecolor)\in\somesetedge$,
  the vertex $\somevertex$ is splitting.
\end{cor}
This immediately yields a proof of \cref{lem:splittingnode}, because
$\somesetedge$ is empty if and only if the proof net has no vertex.

\begin{exa}
  \label{ex:seq:yeo:all}
  Consider the proof net of \cref{fig:wellcolored}:
  using \cref{lem:terminal_not_maximal_order},
  it is easy to check that
  \((\otherothervertex_1,\somecolor)
  \orderyeo(\somevertex,\textcolor{red}{solid})
  \orderyeo(\othervertex,\textcolor{red}{solid})\),
  for \(\somecolor\in\{\textcolor{red}{solid},\textcolor{violet}{dotted}\}\);
  moreover, 
  \((\otherothervertex_2,\textcolor{red}{solid})
  \orderyeo(\othervertex,\textcolor{violet}{dotted})
  \orderyeo(\somevertex,\textcolor{blue}{dashed})\),
  and
  \((\othervertex,\othercolor)
  \orderyeo(\otherothervertex_2,\textcolor{violet}{dotted})\)
  for \(\othercolor\in\{\textcolor{red}{solid},\textcolor{blue}{dashed}\}\).
  Note that we do not have \((\somevertex,\textcolor{blue}{dashed})
  \orderyeo(\otherothervertex_1,\somecolor)\) with 
  \(\somecolor\in\{\textcolor{red}{solid},\textcolor{violet}{dotted}\}\):
  each path \(\somepath\) such that 
  \((\somevertex,\textcolor{blue}{dashed})
  \sobfnbpath[\somepath](\otherothervertex_1,\somecolor)\)
  is reduced to a premise of \(\somevertex\),
  and the other premise yields a path
  \(\otherpath\) such that 
  \((\otherothervertex_1,\somecolor)
  \sobfnbpath[\otherpath](\somevertex,\textcolor{red}{solid})\).
  Hence \((\somevertex,\textcolor{blue}{dashed})\) is maximal.
  The only maximal pairs are thus
  \((\somevertex,\textcolor{blue}{dashed})\)
  and
  \((\otherothervertex_2,\textcolor{violet}{dotted})\)
  (whose maximality is obvious)
  and \(\somevertex\) and \(\otherothervertex_2\) are indeed splitting.
  Note that the remaining splitting vertex
  \(\othervertex\) is \emph{not} a component of a maximal pair:
  not all splitting vertices are obtained by \cref{cor:seq:yeo:all}.
\end{exa}

We can moreover recover two of the classical existence results 
for splitting vertices:
the existence of \emph{sections},
\ie~splitting $\parr$-vertices~\cite{phddanos};
and the existence of splitting terminal vertices~\cite{ll}.

\begin{cor}[Splitting \(\parr\)]
  \label{cor:seq:yeo:sections}
  Let $\somesetedge$ be the set of all cusp-points:
  a $\orderyeo$-maximal element of $\somesetedge$ is $(\somevertex, \somecolor)$
  with $\somevertex$ a splitting $\parr$-vertex.
\end{cor}
In the proof net of \cref{fig:wellcolored},
the only \(\parr\)-vertex \(\somevertex\)
is indeed splitting.
To derive a proof of \cref{lem:splittingnode},
it remains only to treat the case $\somesetedge=\emptyset$:
the proof net is \(\parr\)-free, and we apply \cref{lem:par:free}.

\begin{cor}[Splitting terminal vertices]
  \label{cor:seq:yeo:terminal}
  Considering
  \[
    \somesetedge\eqdef
    \{(\somevertex, \somecolor) \in\vertexset\times\colors \mid
    \text{$\somecolor\neq\coloring(\someedge,\somevertex)$ for each conclusion edge \(\someedge\) of \(\somevertex\)}\}
    \,,
  \]
  each $\orderyeo$-maximal element of \(\somesetedge\) is a pair \((\somevertex, \somecolor)\)
  where \(\somevertex\) is a splitting terminal vertex.
\end{cor}
\begin{proof}
  That $\somevertex$ is splitting is again a direct application of \cref{th:ParamLocalYeo}.
  We show that it is also terminal.
  Indeed, otherwise, it would have a conclusion $\someedge$ with defined target $\othervertex$.
  Using \cref{lem:terminal_not_maximal_order},
  one gets $(\somevertex, \somecolor) \orderyeo (\othervertex, \coloring(\someedge, \othervertex))$
  since \(\somecolor \not= \coloring(\someedge, \somevertex)\) (by \cref{def:wellcolored}):
  this contradicts the maximality of $(\somevertex, \somecolor)$.
\end{proof}
In the proof net of \cref{fig:wellcolored},
the only terminal vertex \(\othervertex\)
is indeed splitting.
Assuming (\wolog) that the set \(\colors\) of colors has at least three
elements, the set \(\somesetedge\) is empty iff the proof net has no vertex:
again, we derive a proof of \cref{lem:splittingnode}.

Note that \cref{th:ParamLocalYeo} is flexible enough to 
allow for other, original (although maybe not so interesting)
strategies.
For instance, one may strive to obtain a splitting
non-\(\ax\)-vertex:
\begin{cor}[Splitting non-\(\ax\)-vertices]
  \label{cor:seq:yeo:nonax}
  Take $\somesetedge$ the set of all vertex-color pairs
  \((\somevertex, \somecolor)\),
  with \(\somevertex\) not an $\ax$-vertex.
  By \cref{th:ParamLocalYeo}, each $\orderyeo$-maximal element
  $(\somevertex, \somecolor)\in\somesetedge$ yields a splitting vertex $\somevertex$.
\end{cor}
In the proof net of \cref{fig:wellcolored},
the maximal pairs in \(\somesetedge\) are then
\((\somevertex,\textcolor{blue}{dashed})\) and
\((\othervertex,\textcolor{red}{solid})\),
and indeed, both non-\(\ax\)-vertices are splitting.
Again, to derive a proof of \cref{lem:splittingnode},
this leaves only the case $\somesetedge=\emptyset$,
and we reason as before.

As for \cref{cor:seq:yeo:all},
\cref{cor:seq:yeo:sections,cor:seq:yeo:terminal,cor:seq:yeo:nonax}
need not produce \emph{all} the splitting vertices of the class of interest.
Counter-examples are provided in 
\cref{fig:splitting_parr_no_max}
(for splitting \(\parr\)-vertices)
and \cref{fig:splitting_terminal_no_max}
(both for splitting terminal vertices, and for splitting non-\(\ax\)-vertices).

\begin{figure}
  \centering
  \begin{tikzpicture}[colored ps, node distance=.25cm and .5cm]
    \coordinate (c);
    \node [above=of c, label=right:\(\somevertex\)] (parr1) {$\parr$};
    \coordinate [above left=of parr1] (parr1l);
    \node [above left=of parr1l] (tens) {$\tensor$};
    \coordinate [above left=of tens] (tensl);
    \coordinate [above right=of tens] (tensr);
    \node [above right=of tensr] (ax1) {$\ax$};
    \node [above left=of tensl, label=below left:\(\othervertex\)] (parr2) {$\parr$};
    \coordinate [above left=of parr2] (parr2l);
    \coordinate [above right=of parr2] (parr2r);
    \node [above=.5 of parr2] (ax2) {$\ax$};
    \path
      (parr2)
      to[draw,out=135,in=-135] 
      coordinate[pos=.5] (m)
      (ax2)
      ;
    \draw[red]
      (parr1l) edge (parr1)
      (tensl) edge (tens)
      (parr1) edge[out=45,in=-45] (ax1)
      (parr2) edge[out=45,in=-45] (ax2)
      (parr2) edge[out=135,in=-90] (m)
      ;
    \draw[dotted,violet]
      (tens) edge (ax1)
      (ax2) edge[out=-135,in=90] (m)
      ;
    \draw[dashed,blue]
      (c) edge (parr1)
      (parr1l) edge (tens)
      (tensl) edge (parr2)
      ;
  \end{tikzpicture}
  \caption{
    Proof net with two splitting $\parr$-vertices
    \(\somevertex\) and \(\othervertex\) such that 
    \((\othervertex,\textcolor{red}{solid})\orderyeo (\somevertex,\textcolor{red}{solid})\)
  }
  \label{fig:splitting_parr_no_max}
\end{figure}

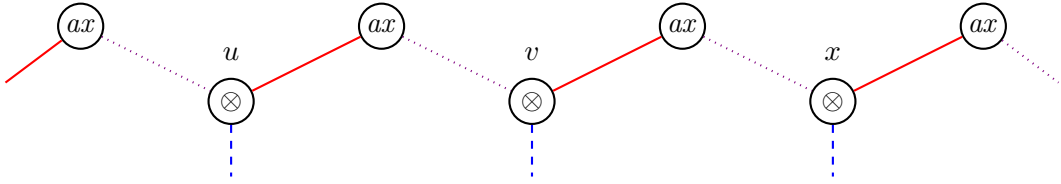
\begin{figure}
  \centering
  \begin{tikzpicture}[colored ps,xscale=2]
    \node (axll) at (-3,1) {$\ax$};
    \node (axl) at (-1,1) {$\ax$};
    \node (axr) at (1,1) {$\ax$};
    \node (axrr) at (3,1) {$\ax$};
    \node [label=above:\(\othervertex\)] (tensl) at (-2,0) {$\tensor$};
    \node [label=above:\(\somevertex\)] (tensc) at (0,0) {$\tensor$};
    \node [label=above:\(\otherothervertex\)] (tensr) at (2,0) {$\tensor$};
    \coordinate (cll) at (-3.5,0.25);
    \coordinate (cl) at (-2,-1);
    \coordinate (cc) at (0,-1);
    \coordinate (cr) at (2,-1);
    \coordinate (crr) at (3.5,0.25);
    \draw[red]
      (cll) edge (axll)
      (tensl) edge (axl)
      (tensc) edge (axr)
      (tensr) edge (axrr)
      ;
    \draw[dotted,violet]
      (axll) edge (tensl)
      (axl) edge (tensc)
      (axr) edge (tensr)
      (axrr) edge (crr)
      ;
    \draw[dashed,blue]
      (tensl) edge (cl)
      (tensc) edge (cc)
      (tensr) edge (cr)
      ;
  \end{tikzpicture}
  \caption{
    Proof net with three terminal splitting \(\otimes\)-vertices \(\somevertex\), \(\othervertex\) and \(\otherothervertex\),
    such that \((\somevertex,\somecolor)\orderyeo(\othervertex,\textcolor{red}{solid})\)
    for \(\somecolor\in\{\textcolor{red}{solid},\textcolor{blue}{dashed}\}\),
    and \((\somevertex,\textcolor{violet}{dotted})\orderyeo(\otherothervertex,\textcolor{violet}{dotted})\)
  }
  \label{fig:splitting_terminal_no_max}
\end{figure}

\subsection{Restrictions}\label{sec:seq:restrictions}

Now that we have sequentialization and desequentialization for full \MLLhmix,
we can consider some restrictions and characterize sub-systems of the sequent calculus,
by means of properties of their image in proof structures.

First, observe that a derivation $\someproof$ contains no ($\hyp$) rule if and only if $\deseq(\someproof)$ is closed.
As an immediate consequence of \cref{thm:seq}, we thus obtain:
\begin{thm}[Closed sequentialization]\label{th:seq:closed}
  Given a closed proof net $\someps$,
  there exists an hypothesis free derivation $\someproof$ such that $\someps\giso\deseq(\someproof)$.
\end{thm}
Note that, by the results of \cref{sec:seq:yeo},
one can impose the splitting vertex obtained by \cref{lem:splittingnode} in the proof of \cref{thm:seq}
to be terminal.
Thus, we can choose to consider ($\hyp$)-free derivations and closed proof structures only,
not only in the statement, but all along the process of sequentialization.
Indeed, if $\someps$ is closed and $\somevertex$ is a splitting terminal vertex:
the components associated with the premises of $\somevertex$ are also closed;
and those associated with its conclusions (if any) are reduced to a single edge, so there is no need to perform any substitution.
Following this approach, one can adapt the proof of \cref{thm:seq} to obtain a proof of 
\cref{th:seq:closed} without ever considering hypotheses in derivations
nor premises in proof structures
-- this is possibly one reason why some authors (\cite{ll}, among others)
favor sequentialization along terminal splitting vertices.

\medskip

Another important sub-system is obtained by removing the $\mix$ rules.
By \cref{lem:deseq:mix}, for any \(\mix\)-free proof \(\someproof\),
$\deseq(\someproof)$ is a connected proof net.
Conversely, given some connected proof net $\someps$,
our proof of \cref{thm:seq} (following any strategy for finding splitting vertices)
yields a \mixretore-normal proof \(\someproof\) with $\someps\giso\deseq(\someproof)$.
Applying \cref{lem:deseq:mix} again, we obtain $\#\mix_2 = \#\mix_0$,
where $\#\mix_i$ is the number of ($\mix_i$) rules in $\someproof$;
and since \(\someproof\) is in \mixretore-normal form,
it must be \(\mix\)-free.
We obtain:
\begin{thm}[Connected sequentialization]
  Given a connected proof net $\someps$,
  there exists a $\mix$-free derivation $\someproof$ such that $\someps\giso\deseq(\someproof)$.
\end{thm}

Of course, one can combine both constraints and recover the original 
result of sequentialization of closed, connected proof nets into proofs of \MLL\
(plain multiplicative linear logic, without hypotheses nor \(\mix\)-rules).

Beyond the fact that it was a requirement of the original Danos-Regnier criterion,
DR-connectedness induces a somehow refined theory:
for instance, it plays a crucial rôle in the study of various notions of 
sub-proof structures, such as kingdoms and empires~\cite{kingemp};
and some classical proofs of sequentialization rely crucially on DR-connectedness
(\eg, the proofs based on empires~\cite{quantif2}).
We thus find interesting to study more in detail what kind of additional
results on DR-connectedness can be derived from our approach:
this will be done in \cref{sec:more_connectedness}.

\subsection{Generalizations}\label{sec:seq:generalizations}

The proof of \cref{thm:seq} in \cref{sec:seq:splitting} is presented in direct relation with the logical system \MLLhmix\
and types of edges allow to build formulas in proofs.
However, the core of this sequentialization process is the existence of splitting vertices which does not rely
on typing at all:
\cref{def:splitting} does not mention typing,
and only the shape of components in \cref{fig:splitting} is relevant;
DR-correctness is defined without reference to typing;
and the proofs of \cref{lem:splittingnode} that we developed in
\cref{sec:seq:dummies,sec:seq:yeo} never mention typing constraints.
So, one could define \definitive{untyped proof structures} exactly as in
\cref{sec:mllpn:ps}, just forgetting about typing labels and constraints,
and still obtain \cref{lem:splittingnode}.

\begin{figure}
  \[
  \begin{tikzpicture}[
    every node/.style={circle,minimum size=6mm,inner sep=0,draw=black,line width=0.8pt},
    every edge/.style={draw=black,-,line width=0.8pt},
    baseline=(v),
  ]
          \node (v) at (0,1) {$\ax$};
          \coordinate (1) at (-1.5,0) {};
          \coordinate (2) at (-.5,0) {};
          \coordinate (n-3) at (.3,0) {};
          \coordinate (n-2) at (.5,0) {};
          \coordinate (n-1) at (.7,0) {};
          \coordinate (n) at (1.5,0) {};
          \path (v) edge (1);
          \path (v) edge (2);
          \path (v) edge[dotted] (n-1);
          \path (v) edge[dotted] (n-2);
          \path (v) edge[dotted] (n-3);
          \path (v) edge (n);
          \node[below=-2mm,subps] at (1) {$\someps_1$};
          \node[below=-2mm,subps] at (2) {$\someps_2$};
          \node[below=-2mm,subps,draw=none] at (n-2) {$\ldots$};
          \node[below=-2mm,subps] at (n) {$\someps_n$};
  \end{tikzpicture}
  \quad
  \begin{tikzpicture}[
    every node/.style={circle,minimum size=6mm,inner sep=0,draw=black,line width=0.8pt},
    every edge/.style={draw=black,-,line width=0.8pt},
    baseline=(v),
  ]
          \coordinate (a) at (-1,1) {};
          \coordinate (b) at (1,1) {};
          \node (v) at (0,0) {$\cut$};
          \path (a) edge (v);
          \path (b) edge (v);
          \node[above=-2mm,subps] at (a) {$\someps_1$};
          \node[above=-2mm,subps] at (b) {$\someps_2$};
  \end{tikzpicture}
  \qquad
  \begin{tikzpicture}[
    every node/.style={circle,minimum size=6mm,inner sep=0,draw=black,line width=0.8pt},
    every edge/.style={draw=black,-,line width=0.8pt},
    baseline=(v),
  ]
          \coordinate (1) at (-1.5,1) {};
          \coordinate (2) at (-.5,1) {};
          \coordinate (n-3) at (.3,1) {};
          \coordinate (n-2) at (.5,1) {};
          \coordinate (n-1) at (.7,1) {};
          \coordinate (n) at (1.5,1) {};
          \node (v) at (0,0) {$\tensor$};
          \coordinate (0) at (0,-1) {};
          \path (v) edge (0);
          \path (v) edge (1);
          \path (v) edge (2);
          \path (v) edge[dotted] (n-1);
          \path (v) edge[dotted] (n-2);
          \path (v) edge[dotted] (n-3);
          \path (v) edge (n);
          \node[above=-2mm,subps] at (1) {$\someps_1$};
          \node[above=-2mm,subps] at (2) {$\someps_2$};
          \node[above=-2mm,subps,draw=none] at (n-2) {$\ldots$};
          \node[above=-2mm,subps] at (n) {$\someps_n$};
          \node[below=-2mm,subps] at (0) {$\someps_0$};
  \end{tikzpicture}
  \qquad
  \begin{tikzpicture}[
    every node/.style={circle,minimum size=6mm,inner sep=0,draw=black,line width=0.8pt},
    every edge/.style={draw=black,-,line width=0.8pt},
    baseline=(v),
  ]
          \coordinate (1) at (-1.5,1) {};
          \coordinate (2) at (-.5,1) {};
          \coordinate (h) at (0,1) {};
          \coordinate (n-3) at (.3,1) {};
          \coordinate (n-2) at (.5,1) {};
          \coordinate (n-1) at (.7,1) {};
          \coordinate (n) at (1.5,1) {};
          \coordinate (0) at (0,-1) {};
          \node (v) at (0,0) {$\parr$};
          \path (v) edge (0);
          \path (v) edge (1);
          \path (v) edge (2);
          \path (v) edge[dotted] (n-1);
          \path (v) edge[dotted] (n-2);
          \path (v) edge[dotted] (n-3);
          \path (v) edge (n);
          \node[above=-3mm,subps,minimum width=4cm,minimum height=1.2cm] at (h) {$\someps_1$};
          \node[below=-2mm,subps] at (0) {$\someps_0$};
  \end{tikzpicture}
  \]
  \caption{Shape of the connected component of each kind of splitting vertex, with arbitrary arities}
  \label{fig:splitting:nary}
\end{figure}
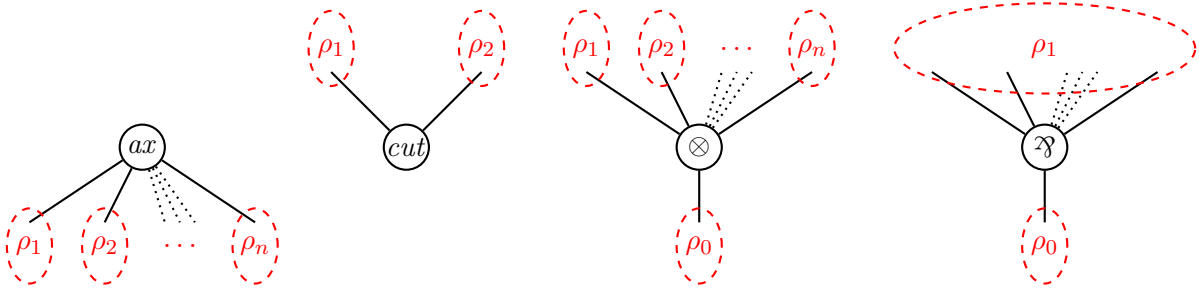

The arities of vertices can moreover be relaxed to allow for 
\(\ax\)-vertices with any number of conclusions,
and \(\parr\)- and \(\tensor\)-vertices with any number of premises,
to obtain \definitive{generalized (untyped) proof structures}.
The notion of splitting vertex is unchanged:
a splitting \(\ax\)-, \(\cut\), or \(\tensor\)-vertex is one that is not in any cycle;
and a splitting \(\parr\)-vertex is one whose conclusion edge is not in any cycle
(see \cref{fig:splitting:nary}).
The notion of well-colored proof structure is then naturally adapted:
a generalized proof structure is well-colored if no two adjacent edges of 
a non-\(\parr\)-vertex are assigned the same color,
and all the premises of each \(\parr\)-vertex are assigned the same color,
different from the one assigned to its conclusion
-- this might require an arbitrarily large number of colors,
depending on the arity of vertices.
Again, \cref{lem:splittingnode} stands for generalized proof structures:
any of the proofs we have provided in \cref{sec:seq:dummies,sec:seq:yeo}
applies \emph{verbatim}.

It follows that our approach to sequentialization can be modularly extended to
richer systems:
one constructs an open derivation from a proof structure,
by reasoning inductively on the number of edges and vertices,
and each splitting vertex provides a decomposition
of the proof structure, compatible with the application of a rule
-- possibly followed by substitutions.
Dealing with multiplicative units is straightforward,
as long as we allow for $\mix$-rules:\footnote{
  The theory of proof nets for multiplicative linear logic with units and without \(\mix\)
  is notably difficult~\cite{mllpnpspace}, and beyond the scope of our discussion.
  Nonetheless, our approach also applies in a framework with a jump edge for each $\bot$-vertex~\cite{pn,hughes*cat,hughesunit} -- with cusps made exactly by the pairs of non-jump premises of $\parr$-vertices.
}
geometrically, units can be treated as unary axioms
(or, equivalently, nullary \(\parr\)- or \(\tensor\)-vertices), which are always splitting vertices.
Similarly, we obtain the sequentialization of multiplicative exponential proof nets,
in presence of structural rules (weakening, contraction, dereliction for the $\wn$-modality)
and promotion (introducing the \(\oc\)-modality):
contraction (resp.\ dereliction; weakening) is treated as a binary (resp.\ unary; nullary) $\parr$-vertex;
and promotion boxes behave like generalized axioms for the purpose of sequentialization
at top-level (the content of each box being sequentialized inductively).

On the other hand, dealing with additive connectives in the spirit of~\cite{mallpnlong}
requires a generalization of our approach relying on \cref{th:ParamLocalYeo}.
This will be treated in \cref{sec:yeo_mall,section:additives,sec:seqmall}.

\section{More on Connectedness for Proof Nets}\label{sec:more_connectedness}

In the present section, we discuss further the notion of connected proof net
along two independent directions.

In \cref{sec:more_connectedness:almost},
we observe that, in a connected proof net, each \(\parr\)-vertex admits a \emph{proper cycle} --
that is (almost) a cusp-free cycle starting and ending with its premise edges --
and we show that a kind of converse property holds:
if a proof net is such that every \(\parr\)-vertex has a proper cycle,
then it is a disjoint union of connected proof nets.
This naturally introduces a variant of connectedness which fits well our approach for sequentialization: \emph{almost connectedness}.

And in \cref{sec:more_connectedness:kingdom}, we compare the order relation \(\orderyeo\) with the order induced by the notion of kingdoms in connected proof nets, which is the usual order in the literature when considering sequentialization.

\subsection{Almost Connected Proof Nets}
\label{sec:more_connectedness:almost}


In the spirit of \cref{lem:localglobal} and \cref{rem:correcgraph},
let us first give an alternative definition of connected DR-correct proof structures.

\begin{defi}
  A DR-correct proof structure is \definitive{cf-connected}
  if it is non-empty and for every two edges \(\someedge\not=\otheredge\),
  there exists a cusp-free simple path from an adjacent vertex of \(\someedge\)
  to an adjacent vertex of \(\otheredge\).
\end{defi}

\begin{lem}\label{lem:cfconnectedness}
  A proof net is connected if and only if it is cf-connected.
\end{lem}

\begin{proof}
  Fix a proof net \(\someps\).
  The underlying partial graph is non-empty in both cases.

  If \(\someps\) is connected,
  it is sufficient to consider one of its correctness graphs:
  every pair of edges is connected in this graph, and a simple path in a correctness graph is cusp-free.

  Conversely, assuming \(\someps\) is cf-connected, we show it is connected by induction on its number of vertices.
  If \(\someps\) is \(\parr\)-free, then its unique correctness graph is just the underlying graph of \(\someps\):
  by assumption, every pair of edges is connected by a (cusp-free) path;
  and since every vertex of a proof structure is adjacent to at least one edge,
  it follows that the correctness graph is connected, hence \(\someps\) is connected.
  If \(\someps\) has a \(\parr\)-vertex then,
  by \cref{lem:splittingpar} (or \cref{cor:seq:yeo:sections}),
  it has a splitting \(\parr\)-vertex \(\somevertex\):
  by induction hypothesis, the connected components \(\someps_0\) and \(\someps_1\)
  of the proof net obtained by removing \(\somevertex\) (see \cref{fig:splitting})
  are connected proof nets;
  each correctness graph of \(\someps\) is then connected, because it is obtained
  by joining a correctness graph of \(\someps_0\) with a correctness graph of \(\someps_1\)
  via \(\somevertex\)
  (which is adjacent to its conclusion, and to the premise selected in the correctness graph of \(\someps\)).
\end{proof}

A \definitive{proper cycle} of a $\parr$-vertex $\somevertex$ of a proof structure is 
a cycle of source \(\somevertex\), with a cusp at \(\somevertex\), but no internal cusp:
it is in particular of the form
$\concatpath{(\somevertex,\someedge_1,\othervertex_1)}{\concatpath{\somepath}{(\othervertex_2,\someedge_2,\somevertex)}}$
where $\someedge_1$ and $\someedge_2$ are the two premises of $\somevertex$ and $\somepath$ is a cusp-free path.

Recall that we defined a strict partial order $\orderparr$ on the $\parr$-vertices of a proof net
(\cref{def:order,lem:order_is_order}).

\begin{lem}\label{lem:clasping_order}
In a well-colored proof net, of local coloring $\coloring$, let $\othervertex$ be a vertex in a proper cycle of a $\parr$-vertex $\somevertex\neq\othervertex$.
Then, for any color $\somecolor$, $(\othervertex,\somecolor)\orderyeo(\somevertex,\coloring(\someedge,\somevertex))$ with $\someedge$ any premise of $\somevertex$.
In particular, if $\othervertex$ is a $\parr$-vertex then $\othervertex\orderparr\somevertex$.
\end{lem}

\begin{proof}
Let $\somecycle$ be a proper cycle of $\somevertex$ containing $\othervertex$.
Since $\somecycle$ has no cusp at $\othervertex$, it contains an edge $\otheredge$ incident to $\othervertex$ and such that $\coloring(\otheredge,\othervertex)\neq\somecolor$: let us call $\somecycle'$ the cycle using the same edges as $\somecycle$, with source $\othervertex$ and with starting edge $\otheredge$.
Notice that $\somecycle'$ has for unique cusp the one at $\somevertex$, and since we are in a proof net there is no cusp-free cycle (\cref{lem:localglobal}): thus $\somecycle'\in\mincycles{\othervertex}$.
Setting $\somepath$ the sub-path of $\somecycle'$ from $\othervertex$ to $\somevertex$, by \cref{lem:bjumping_Yeo} we conclude that $(\othervertex,\somecolor)\orderyeo[\somepath](\somevertex,\coloring(\someedge,\somevertex))$.
\end{proof}

\begin{rem}\label{rem:orders}
  The converse of \cref{lem:clasping_order} does not hold: in the proof net of \cref{fig:splitting_parr_no_max},
  $\othervertex\orderparr\somevertex$ but $\othervertex$ is not in a proper cycle of $\somevertex$.
\end{rem}

We say a DR-correct proof structure is \definitive{almost connected} if it is
non-empty and every $\parr$-vertex has a proper cycle.

\begin{rem}\label{rem:switching-nopiers}
  A path without internal cusp needs not be switching:
  indeed, a proper cycle of a $\parr$-vertex is not switching.
  The proof net of \cref{fig:ps_ex,fig:wellcolored} is connected
  (and almost connected) in particular because the unique proper cycle of the
  $\parr$-vertex has no internal cusp.
\end{rem}

\begin{lem}\label{lem:connectedIsAlmostConnected}
Every connected proof net is almost connected.
\end{lem}

\begin{proof}
  Assume \(\someps\) is a connected proof net.
  We show that every $\parr$-vertex of \(\someps\) has a proper cycle.
  Consider the two premises $\someedge_1$ and $\someedge_2$ of a $\parr$-vertex $\somevertex$.
  Write \(\othervertex_1\eqdef\source(\someedge_1)\) and \(\othervertex_2\eqdef\source(\someedge_2)\).
  By DR-connectedness, there exists a simple path \(\somepath\) from
  \(\othervertex_1\) to \(\othervertex_2\) in a correctness graph of \(\someps\),
  which is a cusp-free path in \(\someps\).

  We first show that \(\somepath\) contains no premise of \(\somevertex\).
  Indeed if it contains one, being a switching path,
  it contains exactly one, say \(\someedge_1\).
  Since \(\somepath\) is simple, \(\someedge_1\) must be the first or last edge of \(\somepath\).
  So, \wolog, we can write \(\somepath=\concatpath{(\othervertex_1,\someedge_1,\somevertex)}{\otherpath}\),
  with \(\otherpath\) a simple cusp-free path from \(\somevertex\) to \(\othervertex_2\),
  not containing \(\someedge_2\);
  but then \(\concatpath{(\othervertex_2,\someedge_2,\somevertex)}{\otherpath}\)
  is a cusp-free cycle, contradicting DR-correctness.

  Hence, \(\somepath\) is a simple cusp-free path from \(\othervertex_1\) to \(\othervertex_2\),
  not containing  \(\someedge_1\) nor \(\someedge_2\),
  and we obtain a proper cycle of \(\somevertex\):
  \(\concatpath{(\somevertex,\someedge_1,\othervertex_1)}{\concatpath{\somepath}{(\othervertex_2,\someedge_2,\somevertex)}}\).
\end{proof}

We are now ready to show that, in an almost connected proof net $\someps$,
being connected by a path amounts to being connected by a cusp-free simple path.
That is, each connected component of the partial graph of $\someps$ defines a connected proof net.

\begin{lem}[Paths in an almost connected proof net]\label{proofnets:lemma:PathsInAlmost}
  Let $\someps$ be an almost connected proof net.
  For every path $\somepath$ in $\someps$,
  there exists a cusp-free simple path $\somepath'$ of $\someps$ having the same endpoints as $\somepath$.
\end{lem}

\begin{proof}
  By \cref{lem:order_is_order},
  the length of sequences $(\somevertex_i)_{1\le i\le n}$
  such that each $\somevertex_i$ is a $\parr$-vertex
  and $\somevertex_{i+1}\orderparr\somevertex_{i}$ for $1\le i<n$ 
  is bounded by the number of $\parr$-vertices of $\someps$.
  Let us call rank of $\somevertex$
  -- and denote by $cr(\somevertex)$ --
  the greatest number $n$ of vertices of such a sequence
  $\somevertex_n\orderparr \cdots \orderparr \somevertex_1=\somevertex$ with last element $\somevertex$.
  Note that $cr(\somevertex)\ge 1$ by definition.
  If $\somepath$ is a path in $\someps$,
  the rank $cr(\somepath)$ of $\somepath$ is the maximum of the ranks of the vertices of the internal cusps of $\somepath$
  -- we set $cr(\somepath)=0$ if $\somepath$ has no internal cusp.

Now, we move to the proof of the statement.
We can assume \wolog\ that $\somepath$ is an open simple path, the result being obvious if $\somepath$ is closed.
We reason by induction on $cr(\somepath)$.
Let $l$ be the number of cusps of $\somepath$, at vertices $\somepier_{1},\ldots,\somepier_{l}$, and call $\someedge_i,\otheredge_{i}$ the two premises of $\somepier_{i}$, of respective sources $\otherothervertex_i,\otherotherothervertex_i$, such that $(\otherothervertex_i,\someedge_i,\somepier_i,\otheredge_i,\otherotherothervertex_i)$ is a sub-path of $\somepath$.

We are first going to build, for every $1\leq i\leq l$, a cusp-free simple path $\gamma_i$ with source $\otherothervertex_i$ and target $\otherotherothervertex_i$.
Since $\someps$ is almost connected, we can select, for every $1\leq i\leq l$, a proper cycle $\pi_{i}$ of $\somepier_{i}$: \wolog\ (otherwise consider the reverse of $\pi_i$), we can write $\pi_i=\concatpath{\concatpath{(\somepier_i,\someedge_i,\otherothervertex_i)}{\gamma_i}}{(\otherotherothervertex_i,\otheredge_i,\somepier_i)}$.
As $\pi_{i}$ is a proper cycle, the path $\gamma_{i}$ is cusp-free and we have $\source(\gamma_i)=\otherothervertex_i$ and $\target(\gamma_i)=\otherotherothervertex_i$.

  We have $\somepath=\concatpath{\concatpath{\somepath_{1}}{(\otherothervertex_1,\someedge_1,\somepier_1,\otheredge_1,\otherotherothervertex_1)}}{\somepath_{2}}\cdots\concatpath{\concatpath{\somepath_{i}}{(\otherothervertex_i,\someedge_i,\somepier_i,\otheredge_i,\otherotherothervertex_i)}}{\somepath_{i+1}}\cdots\concatpath{\concatpath{\somepath_{l}}{(\otherothervertex_l,\someedge_l,\somepier_l,\otheredge_l,\otherotherothervertex_l)}}{\somepath_{l+1}}$
  for some cusp-free simple paths $\somepath_{1},\cdots,\somepath_{l+1}$, and we also know that the simple paths $\gamma_{1},\ldots,\gamma_{l}$ are all cusp-free.
Then the path $\otherpath\eqdef\concatpath{\concatpath{\somepath_{1}}{\gamma_1}}{\somepath_{2}}\ldots\concatpath{\concatpath{\somepath_{i}}{\gamma_i}}{\somepath_{i+1}}\ldots\concatpath{\concatpath{\somepath_{l}}{\gamma_l}}{\somepath_{l+1}}$ has obviously the same endpoints as $\somepath$, and the only possible cusps of $\otherpath$ are at the vertices $\otherothervertex_i=\target(\somepath_{i})=\source(\gamma_{i})$ and at the vertices $\otherotherothervertex_i=\target(\gamma_i)=\source(\somepath_{i+1})$ for $1\leq i\leq l$.
Nonetheless, $\otherpath$ may not be simple.

Define $\otherpath'$ as the simple path obtained from $\otherpath$ by recursively removing one of its non-empty closed sub-paths, until there are none (hence until obtaining a simple path).
Note that $\otherpath'$ has the same endpoints as $\otherpath$, hence as $\somepath$, for removing a closed sub-path preserves the endpoints.
We claim that $cr(\otherpath')<cr(\somepath)$.
This is because a vertex $\otherotherotherothervertex$ of a cusp of $\otherpath'$ must belong to some $\gamma_{i}$ for $1\leq i\leq l$, since either $\otherotherotherothervertex$ is a cusp of $\otherpath$ (\ie\ some $\otherothervertex_i$ or $\otherotherothervertex_i$), or it is a cusp obtained by removing a closed sub-path during the construction of $\otherpath'$, and the only vertices appearing several times in $\otherpath$ are in some $\gamma_{i}$ (all vertices in the $\somepath_j$ are distinct as $\somepath$ is simple and open).
Hence, $\otherotherotherothervertex$ belongs to a proper cycle of $\somepier_i$, and by \cref{lem:clasping_order} we have $\otherotherotherothervertex\orderparr\somepier_{i}$, then $cr(\otherotherotherothervertex)<cr(\somepier_{i})$ and thus $cr(\otherpath')<cr(\somepath)$.
We then apply the induction hypothesis to the path $\otherpath'$ and conclude the existence of a cusp-free simple path $\somepath'$ of $\someps$ having the same endpoints as $\otherpath'$, and thus the same endpoints as $\somepath$.
\end{proof}

\begin{cor}[Decomposition of almost connected proof nets]\label{rem:AlmConnectedIsSumConnected}
An almost connected proof net is a non-empty disjoint union of connected proof nets.
\end{cor}

\begin{proof}
  Let $\someps$ be an almost connected proof net.
  By \cref{proofnets:lemma:PathsInAlmost}, two vertices of $\someps$ are connected
  by a simple cusp-free path iff they are in the same connected component of
  $\someps$. As a consequence every connected component of $\someps$ is a
  connected proof net (\cref{lem:cfconnectedness}).
\end{proof}

\subsection{Kingdoms and Order}
\label{sec:more_connectedness:kingdom}

Our order $\orderyeo$ on vertex-color pairs does not correspond to a known order in the literature of proof nets.
The usual order associated to the sequentialization of connected proof nets derives from the notions of kingdoms and empires~\cite{kingemp}.
In particular, kingdoms define a strict partial order~\cite[Lemma~3]{kingemp} from which one can deduce the sequentialization theorem by means of splitting terminal vertices -- as opposed to empires, to which no strict order is associated.

We prove here that, while the kingdom order and our order $\orderyeo$ are not the same, there is a relation between them in connected closed proof nets.

\begin{defi}[Kingdom~{{\cite[Proposition~3.(II)]{kingemp}}}]
The \definitive{kingdom} $\king{\somevertex}$ of a vertex $\somevertex$ in a connected closed proof net is the smallest connected DR-correct sub-proof structure having $\somevertex$ as a terminal vertex.
Equivalently, $\king{\somevertex}$ is the sub-proof structure with the following vertices and edges:
\begin{itemize}
\item
if $\somevertex$ is an $\ax$-vertex, $\king{\somevertex}$ is $\somevertex$ together with its two conclusions;
\item
if $\somevertex$ is a $\tensor$-vertex with premises $\someedge_1$ and $\someedge_2$, then $\king{\somevertex}$ is obtained from $\king{\source(\someedge_1)}$ and $\king{\source(\someedge_2)}$ by adding $\somevertex$ and its conclusion;
\item
if $\somevertex$ is a $\parr$-vertex, given any proper cycle $\somecycle$ of $\somevertex$, $\king{\somevertex}$ is the union of the kingdoms of the vertices of $\somecycle$ other than $\somevertex$, together with $\somevertex$ and its conclusion.%
\footnote{The resulting proof structure does not depend on the choice of the proper cycle \(\somecycle\)~\cite{kingemp}.}
\end{itemize}
\end{defi}

The kingdom ordering between vertices is simply that $\othervertex$ is smaller than $\somevertex$ when $\othervertex\in\king{\somevertex}$.

\begin{lem}
Consider two distinct vertices $\somevertex$ and $\othervertex$ in a well-colored connected closed proof net, with local coloring $\coloring$.
If $\othervertex\in\king{\somevertex}$, then there exists a premise $\someedge$ of $\somevertex$ such that, for any color $\othercolor\notin\{\coloring(\otheredge,\othervertex) \mid \text{$\otheredge$ is a conclusion of $\othervertex$}\}$, $(\othervertex,\othercolor)\orderyeo(\somevertex,\coloring(\someedge,\somevertex))$.
In particular, $\othervertex\in\king{\somevertex} \implies \exists\somecolor, \othercolor, (\othervertex,\othercolor)\orderyeo(\somevertex,\somecolor)$.
\end{lem}
\begin{proof}
We have three cases to consider, reasoning by induction on the above characterization of kingdoms.
\begin{description}
\item[If $\somevertex$ is an $\ax$-vertex]
Then $\somevertex$ is the sole vertex of $\king{\somevertex}$.
\item[If $\somevertex$ is a $\tensor$-vertex with premises $\someedge_1$ and $\someedge_2$]
Then, $(\source(\someedge_1),\coloring(\otheredge_1,\source(\someedge_1)))\orderyeo(\somevertex,\coloring(\someedge_1,\somevertex))$ and $(\source(\someedge_2),\coloring(\otheredge_2,\source(\someedge_2)))\orderyeo(\somevertex,\coloring(\someedge_2,\somevertex))$ for any premise $\otheredge_1$ of $\source(\someedge_1)$ and any premise $\otheredge_2$ of $\source(\someedge_2)$ (using \cref{lem:terminal_not_maximal_order}).
We conclude by induction hypothesis and the transitivity of \(\orderyeo\).
\item[If $\somevertex$ is a $\parr$-vertex]
Let $\somecycle$ be a proper cycle of $\somevertex$.
For any $\otherothervertex\in\somecycle$ and any color $\otherothercolor$, $(\otherothervertex,\otherothercolor)\orderyeo(\somevertex,\coloring(\someedge,\somevertex))$ for some premise $\someedge$ of $\somevertex$ (using \cref{lem:clasping_order}).
We conclude by induction hypothesis and the transitivity of \(\orderyeo\).
\qedhere
\end{description}
\end{proof}

The converse does not hold as shown on \cref{fig:splitting_terminal_no_max} where the vertices of the kingdom of each $\tensor$-vertex are the $\tensor$-vertex itself and the two $\ax$-vertices above it.

\section{Another Generalization of Yeo's Theorem for MALL}
\label{sec:yeo_mall}

Similarly to how \cref{th:ParamLocalYeo} provided the basis of our sequentialization procedure for multiplicative proof nets, a further generalization of Yeo's theorem can be used to prove sequentialization for proof nets of multiplicative-additive linear logic \emph{à la} Hughes and van Glabbeek~\cite{mallpnlong}.
Our \cref{th:ParamLocalYeo} is not sufficient here, for there exist some cusp-free cycles in these proof nets.

From the union of sub-graphs (defined in \cref{sec:graph_def}), one derives the notion of union of cusp-free cycles.
Such a union is connected when it is connected as a sub-graph.
A connected union of cusp-free cycles $\somesetcycle$ is maximal (for the inclusion) when for every connected union of cusp-free cycles $\somesetcycle'$, $\somesetcycle\subseteq\somesetcycle' \implies \somesetcycle=\somesetcycle'$ (using the ordering of sub-graphs for the inclusion from \cref{sec:graph_def}).

\begin{defi}\label{def:necubfcycles}
Given a partial graph $\somegraph$ with a local coloring $\coloring$, we note $\necubfcycles$ the set of all maximal connected unions of cusp-free cycles of $\somegraph$.
\end{defi}

As an example, see \cref{fig:example_necubfcycles} with a locally colored graph and its set $\necubfcycles$.

\begin{figure}
\centering
\begin{tikzpicture}
\begin{scope}[every node/.style={circle,thick,draw}]
	\node (c1) at (-2,1) {};
	\node (c2) at (0,1) {};
	\node (c3) at (-1,0) {};
	\node (c4) at (-2,-1) {};
	\node (c5) at (0,-1) {};
	\node (c6) at (-1,2) {};
	\node (en) at (1.5,1) {};
	\node (ec) at (.5,0) {};
	\node (es) at (1.5,-1) {};
	\node (p1) at (-2.5,0) {};
	\node (p2) at (2.5,1.5) {};
\end{scope}
	\node at (-1,1.4) {$\somecycle_1$};
	\node at (-1,.6) {$\somecycle_2$};
	\node at (-1,-.6) {$\somecycle_3$};
	\node at (1.1,0) {$\somecycle_4$};
\begin{genscope}[red][-,thick]
	\path (c1) edge (c6);
	\path (c1) edge (c3);
	\path (c3) edge (c4);
	\path (c3) edge (ec);
	\path (en) edge (ec);
	\path (p1) edge (c3);
	\path (en) edge (p2);
	\path (es) edge (p2);
\end{genscope}
\begin{genscope}[blue][-,thick,densely dashed]
	\path (c2) edge (c6);
	\path (c2) edge (c3);
	\path (c3) edge (c5);
	\path (es) edge (en);
\end{genscope}
\begin{genscope}[olive][-,thick,densely dash dot dot]
	\path (c2) edge (c1);
	\path (c5) edge (c4);
	\path (ec) edge (es);
\end{genscope}
\end{tikzpicture}
\caption{Locally colored graph with $\necubfcycles = \{\somecycle_1 \cup \somecycle_2 \cup \somecycle_3,\somecycle_4\}$}
\label{fig:example_necubfcycles}
\end{figure}

We recall the strict partial order $\orderyeo$ on edges can be defined in any partial graph $\somegraph$ with a local coloring (see \cref{def:orderyeo} on \cpageref{def:orderyeo}).
The goal of this section is proving the following theorem.

\begin{thm}\label{th:MALLParamYeo}
Consider $\somegraph$ a partial graph with a local coloring $\coloring$.
For each $\somesetcycle \in \necubfcycles$, select an edge $\edgeofsetcycles(\somesetcycle)$ with endpoints $\sourceofsetcycles(\somesetcycle) \in \somesetcycle$ and $\targetofsetcycles(\somesetcycle) \notin \somesetcycle$ such that:
\begin{enumeratecref}[start=1,label=$(H_\somegraph^\arabic*)$,ref=$(H_\somegraph^\arabic*)$]
\item\label{HG1}
$(\sourceofsetcycles(\somesetcycle), \coloring(\edgeofsetcycles(\somesetcycle), \sourceofsetcycles(\somesetcycle))$ is not a cusp-point;
\item\label{HG2}
$(\targetofsetcycles(\somesetcycle), \coloring(\edgeofsetcycles(\somesetcycle), \targetofsetcycles(\somesetcycle))$ is a cusp-point.
\end{enumeratecref}

Pose $\somesetedge$ a set of vertex-color pairs which dominates cusp-points and is \emph{disjoint} from\\ $\somesetedge_{out} \eqdef
\{(\sourceofsetcycles(\somesetcycle), \coloring(\edgeofsetcycles(\somesetcycle), \sourceofsetcycles(\somesetcycle)) \suchthat \somesetcycle \in \necubfcycles\}$.
Then the vertex of any $\orderyeo$-maximal element of $\somesetedge$ (\ie\ for $\orderyeo$ restricted to $\somesetedge$) is splitting.
\end{thm}

By \cref{HG1}, $\somesetedge_{out}$ does not contain any cusp-point, thus one can always find a set $\somesetedge$ satisfying the requirements: it suffices to take the set of all cusp-points.
We call partial graph \definitive{with an exit function} a locally colored partial graph equipped with functions $\edgeofsetcycles$, $\sourceofsetcycles$ and $\targetofsetcycles$ associating to each of its maximal connected unions of cusp-free cycles $\somesetcycle \in \necubfcycles$ an edge $\edgeofsetcycles(\somesetcycle)$ of endpoints $\sourceofsetcycles(\somesetcycle) \in \somesetcycle$ and $\targetofsetcycles(\somesetcycle) \notin \somesetcycle$.

The proof of this theorem is not immediate, and will be the object of \cref{subsubsec:alt_connex_graph,sec:mall_graphs_splitting}.
\Cref{th:MALLParamYeo} implies a generalization of \cref{th:LocalYeo}.

\begin{thm}\label{th:MALLYeo}
Consider $\somegraph$ a locally colored partial graph with an exit function, and containing at least one vertex.
If $\somegraph$ respects \cref{HG1,HG2}, then there exists a splitting vertex in $\somegraph$.
\end{thm}
\begin{proof}
The set $\somesetedge$ of all vertex-color pairs not in $\somesetedge_{out}$ is non-empty: either $\necubfcycles$ is empty and so is $\somesetedge_{out}$; or there is at least one cusp-point which is not in $\somesetedge_{out}$ by \cref{HG2}.
Thus, $\somesetedge$ contains a maximal element $(\somevertex, \somecolor)$ with respect to $\orderyeo$ (\cref{lem:orderyeo_is_order}), and $\somevertex$ is splitting by \cref{th:MALLParamYeo} since $\somesetedge$ contains all cusp-points.
\end{proof}

\begin{rem}
\Cref{th:MALLYeo} is indeed a generalization of \cref{th:LocalYeo}, for a locally colored partial graph $\somegraph$ with no cusp-free cycle trivially respects hypotheses \cref{HG1,HG2}, which are about connected unions of cusp-free cycles.
\end{rem}

\begin{rem}[All hypotheses are needed]\label{rem:hyp_MALLYeo}
We give here examples showing all hypotheses of \cref{th:MALLYeo} (hence of \cref{th:MALLParamYeo}) are needed, even with an edge-coloring.
On all figures mentioned here, the function $\edgeofsetcycles$ is given explicitly, contrary to the functions $\sourceofsetcycles$ and $\targetofsetcycles$ since they can be deduced.

\begin{figure}
\centering
\begin{tikzpicture}
\begin{scope}[every node/.style={circle,thick,draw}]
	\node (nw) at (-1.5,1) {};
	\node (ne) at (1.5,1) {};
	\node (s) at (0,-1) {};
	\node (se) at (2.5,-0.5) {};
\end{scope}
\begin{genscope}[red][-]
	\path[bend left] (nw) edge (ne);
	\path[bend right] (nw) edge node[arete nommee]{$\edgeofsetcycles(\somecycle)$} (s);
	\path (s) edge (se);
\end{genscope}
\begin{genscope}[blue][-,densely dashed]
	\path[bend left] (ne) edge (nw);
	\path[bend left] (ne) edge (s);
	\path (ne) edge (se);
\end{genscope}
	\node at (0,1) {$\somecycle$};
\end{tikzpicture}
\caption{Graph without a splitting vertex, respecting \cref{HG2}, with $\necubfcycles=\{\somecycle\}$, and $\edgeofsetcycles(\somecycle)$ making no cusp at $\sourceofsetcycles(\somecycle)$ with edges outside $\somecycle$}
\label{fig:hyp_MALLYeo_H1_1}
\end{figure}

\begin{figure}
\centering
\begin{tikzpicture}
\begin{scope}[every node/.style={circle,thick,draw}]
	\node (p) at (4,0) {};
	\node (mp) at (2,0) {};
	\node (nw) at (1,-1) {};
	\node (ne) at (1,1) {};
	\node (m) at (0,0) {};
	\node (s) at (-2,0) {};
\end{scope}
\begin{genscope}[red][-]
	\path (nw) edge node[arete nommee]{$\edgeofsetcycles(\somecycle)$} (m);
	\path (ne) edge (m);
	\path (mp) edge (nw);
	\path (ne) edge node[arete nommee]{$\edgeofsetcycles(\othercycle)$} (mp);
\end{genscope}
\begin{genscope}[blue][-,densely dashed]
	\path (ne) edge (nw);
	\path[bend left] (m) edge (s);
	\path[bend left] (p) edge (mp);
\end{genscope}
\begin{genscope}[violet][-,densely dotted]
	\path[bend right] (m) edge (s);
	\path[bend right] (p) edge (mp);
\end{genscope}
	\node at (-1,0) {$\somecycle$};
	\node at (3,0) {$\othercycle$};
\end{tikzpicture}
\caption{Graph without a splitting vertex, respecting \cref{HG2}, with $\necubfcycles=\{\somecycle,\othercycle\}$ and $\edgeofsetcycles(\somecycle)$ and $\edgeofsetcycles(\othercycle)$ making no cusp with edges in $\somecycle$ and $\othercycle$}
\label{fig:hyp_MALLYeo_H1_2}
\end{figure}

Let us first consider hypothesis \cref{HG1}, namely that $(\sourceofsetcycles(\somesetcycle), \coloring(\edgeofsetcycles(\somesetcycle), \sourceofsetcycles(\somesetcycle))$ is not a cusp-point.
It is important to have this hypothesis both for the sub-graph $\somesetcycle$ and for its complementary: \cref{fig:hyp_MALLYeo_H1_1} (\resp\ \cref{fig:hyp_MALLYeo_H1_2}) is a counter-example with an $\edgeofsetcycles(\somesetcycle)$ making a cusp at $\sourceofsetcycles(\somesetcycle)$ with an edge inside (\resp\ outside) $\somesetcycle$.

\begin{figure}
\centering
\begin{tikzpicture}
\begin{scope}[every node/.style={circle,thick,draw}]
	\node (ww) at (-3,0) {};
	\node (w) at (-1,0) {};
	\node (e) at (1,0) {};
	\node (ee) at (3,0) {};
\end{scope}
\begin{genscope}[red][-]
	\path[bend left] (ww) edge (w);
	\path[bend left] (e) edge (ee);
\end{genscope}
\begin{genscope}[blue][-,densely dashed]
	\path[bend right] (ww) edge (w);
	\path[bend right] (e) edge (ee);
\end{genscope}
\begin{genscope}[violet][-,densely dotted]
	\path (w) edge node[arete nommee,above]{$\edgeofsetcycles(\somecycle)$} node[arete nommee,below]{$\edgeofsetcycles(\othercycle)$} (e);
\end{genscope}
	\node at (-2,0) {$\somecycle$};
	\node at (2,0) {$\othercycle$};
\end{tikzpicture}
\caption{Graph without a splitting vertex, respecting \cref{HG1}, with $\necubfcycles=\{\somecycle,\othercycle\}$}\label{fig:hyp_MALLYeo_H2}
\end{figure}

Considering now hypothesis \cref{HG2}, the graph on \cref{fig:hyp_MALLYeo_H2} respects \cref{HG1} but has no splitting vertex.

\begin{figure}
\centering
\begin{tikzpicture}
\begin{scope}[every node/.style={circle,thick,draw}]
	\node[inner sep=0,minimum size=10] (ww) at (-3,0) {$\othervertex$};
	\node[inner sep=0,minimum size=10] (w) at (-1,0) {$\somevertex$};
	\node (e) at (1,0) {};
	\node (e') at (3,0) {};
\end{scope}
\begin{genscope}[red][-]
	\path[bend left] (ww) edge node[arete nommee]{$\otheredge$}  (w);
\end{genscope}
\begin{genscope}[blue][-,densely dashed]
	\path[bend right] (ww) edge (w);
\end{genscope}
\begin{genscope}[violet][-,densely dotted]
	\path (w) edge node[arete nommee]{$\edgeofsetcycles(\somecycle)$} (e);
	\path (e) edge (e');
\end{genscope}
	\node at (-2,0) {$\somecycle$};
\end{tikzpicture}
\caption{Graph respecting \cref{HG1,HG2}, with $\necubfcycles=\{\somecycle\}$ and $(\somevertex, \textcolor{violet}{dotted}) \in \somesetedge_{out}$ maximal for $\orderyeo$ while $\somevertex$ is not splitting}\label{fig:hyp_MALLYeo_Eout}
\end{figure}

Lastly, looking at \cref{th:MALLParamYeo}, \cref{fig:hyp_MALLYeo_Eout} presents a counter-example in the case where the chosen set of edges $\somesetedge$ is not disjoint from $\somesetedge_{out}$.
On this graph, $(\somevertex, \textcolor{violet}{dotted}) \orderyeo[(\somevertex,\otheredge,\othervertex)] (\othervertex, \textcolor{red}{solid})$ does not hold because $(\othervertex, \textcolor{red}{solid}) \sobfnbpath (\somevertex, \textcolor{blue}{dashed})$ using as a path the {\color{blue}dashed} edge; similarly, $(\somevertex, \textcolor{violet}{dotted}) \cancel{\orderyeo} (\othervertex, \textcolor{blue}{dashed})$; and obviously $(\somevertex, \textcolor{violet}{dotted})$ is not smaller than any of the two other vertices as there is no path to them from $
\somevertex$ not starting with color \textcolor{violet}{dotted}, so that $(\somevertex, \textcolor{violet}{dotted})$ is maximal.
\end{rem}

\begin{rem}
The following is a parallel with the single switching cycle conjecture in~\cite{mallpnlong}: when replacing ``maximal connected unions of cusp-free cycles'' simply with ``cusp-free cycles'' in the definition of $\necubfcycles$ (\cref{def:necubfcycles}), does \cref{th:MALLYeo} still hold?
The answer is no: the graph depicted on \cref{fig:no_MALLYeo_SCC} is a counter-example, that does not respect the hypotheses of \cref{th:MALLYeo} but would respect them were we to replace ``maximal connected unions of cusp-free cycles'' with ``cusp-free cycles''.
Nonetheless, this graph cannot be adapted as a counter-example of the single switching cycle conjecture in the context of proof nets, because in this framework all pairs $(\edgeofsetcycles(\somesetcycle),\targetofsetcycles(\somesetcycle))$ are of the same color.

\begin{figure}
\centering
\begin{tikzpicture}
\begin{scope}[every node/.style={circle,thick,draw}]
	\node (n) at (0,1.5) {};
	\node (nw) at (-1.75,1.5) {};
	\node (ne) at (1.75,1.5) {};
	\node (cw) at (-1.75,0) {};
	\node (ce) at (1.75,0) {};
	\node (s) at (0,0) {};
\end{scope}
\begin{genscope}[red][-]
	\path (n) edge node[arete nommee,right]{$\edgeofsetcycles(\othercycle)$} (cw);
	\path (cw) edge (s);
\end{genscope}
\begin{genscope}[blue][-,densely dashed]
	\path (n) edge node[arete nommee,left]{$\edgeofsetcycles(\somecycle)$} (ce);
	\path (ce) edge (s);
\end{genscope}
\begin{genscope}[olive][-,densely dash dot dot]
	\path (nw) edge (cw);
	\path (ne) edge (ce);
\end{genscope}
\begin{genscope}[violet][-,dotted]
	\path (n) edge (nw);
	\path (n) edge (ne);
\end{genscope}
	\node at (-1.25,1) {$\somecycle$};
	\node at (1.25,1) {$\othercycle$};
\end{tikzpicture}
\caption{Graph without a splitting vertex but with both cusp-free cycles having an edge out of it making a cusp at its endpoint out of the cycle but not at its endpoint in the cycle}\label{fig:no_MALLYeo_SCC}
\end{figure}
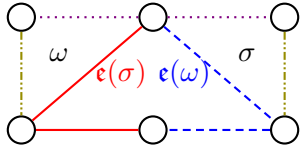
\end{rem}

Although the technical details are a bit more involved, the proof of \cref{th:MALLParamYeo} follows the same idea as the one of \cref{th:ParamLocalYeo} in \cref{sec:paramyeo}: a non-splitting vertex cannot be part of a vertex-color pair that is maximal for $\orderyeo$ (\cref{prop:no_splitting_no_max_graph}).
In particular, we need to study some properties of connected unions of cusp-free cycles, that were absent in our first generalization of Yeo but are central here: this is the object of \cref{subsubsec:alt_connex_graph}.
We then use these properties and cusp minimization in \cref{sec:mall_graphs_splitting} to prove a maximal vertex-color pair contains a splitting vertex.

\paragraph*{Notation}
In the rest of this section, we fix an arbitrary partial graph $\somegraph$ equipped with a local coloring $\coloring$ and an exit function $(\edgeofsetcycles,\sourceofsetcycles,\targetofsetcycles)$.

\subsection{\texorpdfstring{\alternating}{Arrow}-Connectedness}
\label{subsubsec:alt_connex_graph}

\begin{defi}[\alternating-connectedness]\label{def:s_connected_graph}
A sub-graph $\somesubgraph$ of $\somegraph$ is said to be \definitive{\alternating-connected} if for all $\somevertex$ and $\othervertex$ distinct vertices in $\somesubgraph$, for any color $\somecolor$ of $\somegraph$, there exists a path $\somepath$ inside $\somesubgraph$ and a color $\othercolor$ such that $(\somevertex, \somecolor) \sobfnbpath[\somepath] (\othervertex, \othercolor)$.
\end{defi}

\begin{rem}
\label{rem:s_connected_graph}
This notion of connectedness is stronger than the relation ``being linked by a cusp-free simple path''.
\begin{itemize}
\item
It is \emph{transitive}: if $(\somevertex, \somecolor) \sobfnbpath[\somepath] (\othervertex, \othercolor)$ and $(\othervertex, \othercolor) \sobfnbpath[\otherpath] (\otherothervertex, \otherothercolor)$ with $\somepath$ and $\otherpath$ having for sole common vertex $\othervertex$, then $(\somevertex, \somecolor) \sobfnbpath[\concatpath{\somepath}{\otherpath}] (\otherothervertex, \otherothercolor)$.
This is to be opposed with the relation ``being linked by a cusp-free simple path''.
\item
For $\somevertex$ and $\othervertex$ in a \alternating-connected $\somesubgraph$, there are at least two cusp-free paths from $\somevertex$ to $\othervertex$: one starting with some color $\somecolor$ by applying the definition on any color of $\somegraph$, and another starting with color $\othercolor\neq\somecolor$ by applying the definition on color $\somecolor$.
\end{itemize}
\end{rem}

The goal of this section is proving \cref{cor:usc-c=sc_graph}: connected unions of cusp-free cycles are \alternating-connected.

\begin{lem}\label{lem:sc_is_sp_graph}
A cusp-free cycle is \alternating-connected.
\end{lem}
\begin{proof}
Consider a cusp-free cycle $\somecycle$, $\somevertex \neq \othervertex$ vertices in $\somecycle$ and $\somecolor$ a color.
As $\somecycle$ has no cusp at $\somevertex$, seeing it as a cycle of source $\somevertex$ its starting or ending color is not $\somecolor$ (\cref{lem:reversing_trick}): up to reversing $\somecycle$, assume its starting color is not $\somecolor$.
Call $\somepath$ the prefix of $\somecycle$ from $\somevertex$ to $\othervertex$:
then $(\somevertex, \somecolor) \sobfnbpath[\somepath] (\othervertex, \othercolor)$ with $\othercolor$ the ending color of $\somepath$.
\end{proof}

\begin{lem}\label{lem:union_sp_graph}
Let $\somesubgraph$ and $\othersubgraph$ be two \alternating-connected sub-graphs of $\somegraph$ which have at least one vertex in common.
Then $\somesubgraph\cup\othersubgraph$ is \alternating-connected.
\end{lem}
\begin{proof}
Consider vertices $\somevertex \neq \othervertex \in \somesubgraph\cup\othersubgraph$ and $\somecolor$ a color.
If $\somevertex, \othervertex \in \somesubgraph$ or $\somevertex, \othervertex \in \othersubgraph$, then we conclude using \alternating-connectedness of $\somesubgraph$ or of $\othersubgraph$.
Without any loss of generality, assume $\somevertex \in \othersubgraph\backslash\somesubgraph$ and $\othervertex \in \somesubgraph\backslash\othersubgraph$.

By hypothesis, there exists a vertex $\otherothervertex \in \somesubgraph\cap\othersubgraph$; necessarily $\somevertex \neq \otherothervertex$.
By \alternating-connectedness of $\othersubgraph$, there exists a path $\somepath$ in $\othersubgraph$ such that $(\somevertex, \somecolor) \sobfnbpath[\somepath] (\otherothervertex, \othercolor)$.
Consider $\somepath'$ a minimal prefix of $\somepath$ ending in $\somesubgraph\cap\othersubgraph$: $\somepath'$ has no vertex in $\somesubgraph$ except its target $\otherothervertex' \in \somesubgraph\cap\othersubgraph$.
Call $\othercolor'$ the ending color of $\somepath'$: we have $(\somevertex, \somecolor) \sobfnbpath[\somepath'] (\otherothervertex', \othercolor')$ (\cref{rem:sobfnbpath_prefix}).
By \alternating-connectedness of $\somesubgraph$ with $\othersubgraph \ni \otherothervertex' \neq \othervertex \notin \othersubgraph$, there is a path $\otherpath$ in $\somesubgraph$ such that $(\otherothervertex', \othercolor') \sobfnbpath[\otherpath] (\othervertex, \otherothercolor)$.
Thus, $\concatpath{\somepath'}{\otherpath}$ is a simple open cusp-free path from $\somevertex$ to $\othervertex$ (\cref{lem:concatsimplepathsprefix}, since the only vertex of $\somepath'$ in $\somesubgraph$ is its target $\otherothervertex'$).
Its starting color is the starting color of $\somepath'$, so of $\somepath$, thence $(\somevertex, \somecolor) \sobfnbpath[\concatpath{\somepath'}{\otherpath}] (\othervertex, \otherothercolor)$.
\end{proof}

Said in another manner, the previous lemma tells us that a connected union of \alternating-connected sub-graphs is \alternating-connected.
An immediate corollary is that for unions of cusp-free cycles, being \alternating-connected is the same as being connected.

\begin{cor}\label{cor:usc-c=sc_graph}
A finite union of cusp-free cycles of $\somegraph$ is \alternating-connected if and only if it is connected.
\end{cor}
\begin{proof}
The direct implication is trivial.
The converse one follows from \cref{lem:sc_is_sp_graph,lem:union_sp_graph}, with an induction on the number of cusp-free cycles.
Let $\somesetcycle = \bigcup_{i=1}^n\somecycle_i$ be a connected union of $n$ cusp-free cycles.
The empty case $n = 0$ is trivial,
and if $n=1$ then $\somesetcycle = \somecycle_1$ which is \alternating-connected by \cref{lem:sc_is_sp_graph}.
Otherwise, $n > 1$ and $\somesetcycle = \somesubgraph\cup\somecycle_n$ with $\somesubgraph = \bigcup_{i=1}^{n-1}\somecycle_i$ a union of $n-1$ cusp-free cycles.
By connectedness of $\somesetcycle$, each connected component $\somesubgraph_j$ of $\somesubgraph$ respects $\somesubgraph_j\cap\somecycle_n \neq \emptyset$, with $\somesubgraph_j$ a connected union of cusp-free cycles.
By induction hypothesis, each $\somesubgraph_j$ is \alternating-connected, and $\somecycle_n$ also is by \cref{lem:sc_is_sp_graph}.
Then, $\somesetcycle$ is \alternating-connected by repeated applications of \cref{lem:union_sp_graph} on $\somecycle_n\cup\somesubgraph_1$, $\somecycle_n\cup\somesubgraph_1\cup\somesubgraph_2$, \dots, $\somecycle_n\cup\somesubgraph=\somesetcycle$.
\end{proof}

\subsection{Finding a Splitting Vertex}
\label{sec:mall_graphs_splitting}

We now use \alternating-connectedness to prove our theorem, using that each element of $\necubfcycles$ (\cref{def:necubfcycles}) is \alternating-connected thanks to \cref{cor:usc-c=sc_graph}.

\begin{lem}\label{lem:cycles_no_ret_graph}
Consider $\somesetcycle \in \necubfcycles$ and $\someedge$ an edge of endpoints $\somevertex$ and $\othervertex$ such that $\somevertex \in \somesetcycle \not\ni \othervertex$ and $(\somevertex, \coloring(\someedge, \somevertex))$ is not a cusp-point of $\somegraph$.
Then, for all vertex $\otherothervertex\in\somesetcycle$ and color $\somecolor$, $(\othervertex, \coloring(\someedge, \othervertex)) \sobfnbpath (\otherothervertex, \somecolor)$ cannot hold.
\end{lem}

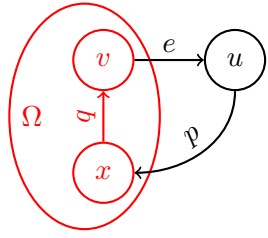
\begin{figure}
\centering
\begin{tikzpicture}
\begin{genscope}[red]
	\draw (.5,0) ellipse (10mm and 15mm);
	\node[draw=none] at (-.2,0) {$\somesetcycle$};
	\node (l) at (.75,.75) {$\somevertex$};
	\node (v) at (.75,-.75) {$\otherothervertex$};

	\path (v) edge node[chemin nomme,above]{$\otherpath$} (l);
\end{genscope}
\begin{genscope}
	\node (w) at (2.5,.75) {$\othervertex$};
	\coordinate (vf) at (2,-.75);

	\path (l) edge node[arete nommee]{$\someedge$} (w);
	\path[out=-90,in=0] (w) edge node[chemin nomme]{$\somepath$} (v);
\end{genscope}
\end{tikzpicture}
\caption{Illustration of the proof of \cref{lem:cycles_no_ret_graph}}
\label{fig:cycles_no_ret_graph}
\end{figure}

\begin{proof}
An illustration of this proof is given on \cref{fig:cycles_no_ret_graph}.
Towards a contradiction, assume $(\othervertex, \coloring(\someedge, \othervertex)) \sobfnbpath[\somepath] (\otherothervertex, \somecolor)$ for some $\otherothervertex \in \somesetcycle$ and color $\somecolor$.
Up to taking a non-empty prefix, $\somepath$ has for only vertex in $\somesetcycle$ its target $\otherothervertex$, with $\somesetcycle \ni \otherothervertex \neq \othervertex\notin \somesetcycle$.
As $\somesetcycle$ is \alternating-connected by \cref{cor:usc-c=sc_graph}, there exists a path $\otherpath$ in $\somesetcycle$ from $\otherothervertex$ to $\somevertex$ with either $\otherpath$ empty or $(\otherothervertex, \somecolor) \sobfnbpath[\otherpath] (\somevertex, \othercolor)$ for some color $\othercolor$.

By \cref{lem:concatsimplepaths,lem:concatsimplepathsprefix}, $\somecycle = \concatpath{(\somevertex,\someedge,\othervertex)}{\concatpath{\somepath}{\otherpath}}$ is a cycle since $\somepath$ cannot start with $\someedge$.
Moreover, it is cusp-free since $(\somevertex, \coloring(\someedge, \somevertex))$ is not a cusp-point.
Therefore, $\somesetcycle\cup\somecycle$ is a connected union of cusp-free cycles, contradicting the maximality of $\somesetcycle \in \necubfcycles$.
\end{proof}

This implies the vertex of a maximal element for $\orderyeo$ cannot belong to a cusp-free cycle.

\begin{lem}\label{lem:in_sc_no_max_graph}
Suppose $\somegraph$ respects \cref{HG1}.
Consider $\somesetcycle \in \necubfcycles$, $\somevertex \in \somesetcycle$ a vertex and $\somecolor$ a color.
Then either $\somevertex = \sourceofsetcycles(\somesetcycle)$ and $\somecolor = \coloring(\edgeofsetcycles(\somesetcycle), \sourceofsetcycles(\somesetcycle))$, or there exists a path $\somepath$ inside $\somesetcycle$ such that $(\somevertex, \somecolor) \orderyeo[\concatpath{\somepath}{(\sourceofsetcycles(\somesetcycle),\edgeofsetcycles(\somesetcycle),\targetofsetcycles(\somesetcycle))}] (\targetofsetcycles(\somesetcycle), \coloring(\edgeofsetcycles(\somesetcycle), \targetofsetcycles(\somesetcycle)))$.
\end{lem}
\begin{proof}
As $\sourceofsetcycles(\somesetcycle) \in \somesetcycle$, using \cref{cor:usc-c=sc_graph} there exists a path $\somepath$ in $\somesetcycle$ from $\somevertex$ to $\sourceofsetcycles(\somesetcycle)$ with either $\somepath$ empty or $(\somevertex, \somecolor) \sobfnbpath[\somepath] (\sourceofsetcycles(\somesetcycle), \othercolor)$ for some color $\othercolor$.
Consider the path $\concatpath{\somepath}{(\sourceofsetcycles(\somesetcycle),\edgeofsetcycles(\somesetcycle),\targetofsetcycles(\somesetcycle))}$.
It is simple and open as all vertices of $\somepath$ belong to $\somesetcycle$ while $\targetofsetcycles(\somesetcycle)$ is outside (\cref{lem:concatsimplepathsprefix}).
Moreover, this path is cusp-free as $\somepath$ is cusp-free and using \cref{HG1}.

Hence, if $\somepath$ is non-empty then $(\somevertex, \somecolor) \sobfnbpath[\concatpath{\somepath}{(\sourceofsetcycles(\somesetcycle),\edgeofsetcycles(\somesetcycle),\targetofsetcycles(\somesetcycle))}] (\targetofsetcycles(\somesetcycle), \coloring(\edgeofsetcycles(\somesetcycle), \targetofsetcycles(\somesetcycle)))$ follows as the starting color of $\concatpath{\somepath}{(\sourceofsetcycles(\somesetcycle),\edgeofsetcycles(\somesetcycle),\targetofsetcycles(\somesetcycle))}$ is the starting color of $\somepath$.
If $\somepath$ is empty, then $\somevertex = \sourceofsetcycles(\somesetcycle)$, and $(\somevertex, \somecolor) \sobfnbpath[(\sourceofsetcycles(\somesetcycle),\edgeofsetcycles(\somesetcycle),\targetofsetcycles(\somesetcycle))] (\targetofsetcycles(\somesetcycle), \coloring(\edgeofsetcycles(\somesetcycle), \targetofsetcycles(\somesetcycle)))$ follows unless $\somecolor = \coloring(\edgeofsetcycles(\somesetcycle), \sourceofsetcycles(\somesetcycle))$, in which case we are done.
In both cases, we get $(\somevertex, \somecolor) \sobfnbpath[\concatpath{\somepath}{(\sourceofsetcycles(\somesetcycle),\edgeofsetcycles(\somesetcycle),\targetofsetcycles(\somesetcycle))}] (\targetofsetcycles(\somesetcycle), \coloring(\edgeofsetcycles(\somesetcycle), \targetofsetcycles(\somesetcycle)))$.

Furthermore, consider any vertex-color pair $(\othervertex, \othercolor)$ and assume $(\targetofsetcycles(\somesetcycle), \coloring(\edgeofsetcycles(\somesetcycle), \targetofsetcycles(\somesetcycle))) \sobfnbpath[\otherpath] (\othervertex, \othercolor)$.
Then $\othervertex \notin \somesetcycle$ by \cref{lem:cycles_no_ret_graph}, so in particular $\othervertex$ is not a vertex of $\somepath$, and $\othervertex \neq \targetofsetcycles(\somesetcycle)$ since $\otherpath$ is open.
Hence, $(\somevertex, \somecolor) \orderyeo[\concatpath{\somepath}{(\sourceofsetcycles(\somesetcycle),\edgeofsetcycles(\somesetcycle),\targetofsetcycles(\somesetcycle))}] (\targetofsetcycles(\somesetcycle), \coloring(\edgeofsetcycles(\somesetcycle), \targetofsetcycles(\somesetcycle)))$.
\end{proof}

This handles vertex-color pairs whose vertices are in cusp-free cycles.
For the others, we can apply cusp cycling (\cref{lem:bjumping_Yeo}) if the cusp we find is at a vertex not in a cusp-free cycle, and a cusp minimization (\cref{lem:bjumping_gen}) if it is in such a cycle.
In this last case, we need some study of paths, making the proof of the next result a bit long.

\begin{prop}\label{prop:no_splitting_no_max_graph}
Suppose $\somegraph$ respects \cref{HG1,HG2}.
Let $(\somevertex, \somecolor)$ be a vertex-color pair \emph{not} in $\{(\sourceofsetcycles(\somesetcycle), \coloring(\edgeofsetcycles(\somesetcycle), \sourceofsetcycles(\somesetcycle)) \suchthat \somesetcycle \in \necubfcycles\}$.
If $\somevertex$ is not splitting, then there exists a cusp-point $(\othervertex, \othercolor)$ such that $(\somevertex, \somecolor) \orderyeo (\othervertex, \othercolor)$.
\end{prop}
\begin{figure}
\centering
\begin{tikzpicture}
\begin{genscope}
	\node (v) at (-2,.3) {$\somevertex$};
	\coordinate (vc) at (-2,-.35);
	\node (fst) at (1,-1) {$\otherothervertex$};
	\node (pier) at (3,-1) {$\somepier$};
	\coordinate (wend) at (4,.75);
	\node (lst) at (1,1.4) {$\otherotherothervertex$};
	\coordinate (tup) at (-.5,1.4);
	\coordinate (tdown) at (-.5,-1);

	\path (v) edge (vc);
	\path[out=-90,in=180] (vc) edge[-] (tdown);
	\path[out=0,in=180] (tdown) edge[-] (fst);
	\path[out=0,in=180] (fst) edge[-] (pier);
	\path[out=0,in=-45] (pier) edge[-] (wend);
	\path[out=135,in=0] (wend) edge[-] (lst);
	\path[out=180,in=0] (lst) edge[-] (tup);
	\path[out=180,in=90] (tup) edge[-] node[chemin nomme]{$\somecycle$} (v);
\end{genscope}
\begin{genscope}[red]
	\node[draw=none] () at (3.5,1.75) {$\somesetcycle$};
	\coordinate (o1) at (1,2);
	\coordinate (o2) at (3.75,2);
	\coordinate (o3) at (3.75,-1.5);
	\coordinate (o4) at (1,-1.5);
	\path (o1) edge[-] (o2);
	\path (o2) edge[-] (o3);
	\path (o3) edge[-] (o4);
	\path (o4) edge[-] (fst);
	\path (fst) edge[-] (lst);
	\path (lst) edge[-] (o1);

	\node (l) at (2.5,.5) {$\sourceofsetcycles$};
	\path (fst) edge node[chemin nomme,below]{$\otherotherpath$} (l);
	\path (lst) edge node[chemin nomme]{$\otherotherotherpath$} (l);
\end{genscope}
\begin{genscope}[teal]
	\node (w) at (.4,.2) {$\targetofsetcycles$};
	\path (l) edge node[arete nommee,above]{$\edgeofsetcycles$} (w);
\end{genscope}
\begin{myscopehigh}{blue}
	\coordinate (vcbis) at (-1.9,-.35);
	\coordinate (tdownbis) at (-.5,-.9);
	\coordinate (tupbis) at (-.5,1.3);
	\path (v.-75) edge (vcbis);
	\path[out=-90,in=180] (vcbis) edge[-] node[chemin nomme]{$\somepath$} (tdownbis);
	\path[out=0,in=180] (tdownbis) edge[-] (fst.165);
	\path (fst.60) edge (l.-150);
	\path (l.202) edge (w.-8);
\end{myscopehigh}
\end{tikzpicture}
\caption{Illustration of named elements in the proof of \cref{prop:no_splitting_no_max_graph}: case where $\somevertex$ is not in a cusp-free cycle but $\somepier$ is}
\label{fig:proof_no_splitting_no_max_graph_p}
\end{figure}
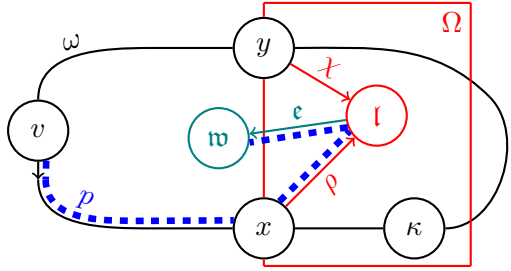
\begin{proof}
If $\somevertex$ belongs to a cusp-free cycle, then the result follows by \cref{lem:in_sc_no_max_graph}, which gives a cusp-point using \cref{HG2}.
Therefore, we assume it is not the case.

As $\somevertex$ is not splitting, $\mincycles{\somevertex} \neq \emptyset$: take some $\somecycle \in \mincycles{\somevertex}$.
Up to reversing $\somecycle$, assume its starting color is not $\somecolor$ (\cref{lem:reversing_trick}).
This cycle contains at least one cusp: denote by $\somepier$ the vertex of the first cusp of $\somecycle$, and by $\othercolor$ its color.
We have two cases, according to whether $\somepier$ belongs to a cusp-free cycle or not.

If $\somepier$ does not belong to a cusp-free cycle, then by \cref{lem:bjumping_Yeo} we have $(\somevertex, \somecolor)\orderyeo[\subpath{\somecycle}{\somevertex}{\somepier}] (\somepier, \othercolor)$.

Thus, suppose from now on that $\somepier$ belongs to a cusp-free cycle, and call $\somesetcycle$ the maximal connected union of cusp-free cycles containing $\somepier$; observe $\somesetcycle \in \necubfcycles$.
We will now name some vertices, edges, paths and colors; see \cref{fig:proof_no_splitting_no_max_graph_p} for an illustration.
We use the notation $\edgeofsetcycles$ for $\edgeofsetcycles(\somesetcycle)$, $\sourceofsetcycles$ for $\sourceofsetcycles(\somesetcycle)$ and $\targetofsetcycles$ for $\targetofsetcycles(\somesetcycle)$.
Pose $\otherothervertex$ (\resp\ $\otherotherothervertex$) the first (\resp\ last) vertex of $\somecycle$ belonging to $\somesetcycle$.
Remark $\somevertex\notin\{\otherothervertex,\otherotherothervertex\}$ as $\somevertex$ is not in any cusp-free cycle, thus not in $\somesetcycle$.
Let $\somecolor_\otherothervertex$ be the ending color of $\subpath{\somecycle}{\somevertex}{\otherothervertex}$, and $\somecolor_\otherotherothervertex$ be the starting color of $\subpath{\somecycle}{\otherotherothervertex}{\somevertex}$.
By \cref{lem:in_sc_no_max_graph}, there exists a path $\otherotherpath$ in $\somesetcycle$ between $\otherothervertex$ and $\sourceofsetcycles$ such that either ($\otherothervertex=\sourceofsetcycles$ and thus) $\otherotherpath$ is empty or $(\otherothervertex, \somecolor_\otherothervertex) \orderyeo[\concatpath{\otherotherpath}{(\sourceofsetcycles,\edgeofsetcycles,\targetofsetcycles)}] (\targetofsetcycles, \coloring(\edgeofsetcycles, \targetofsetcycles))$.
Similarly, by \cref{lem:in_sc_no_max_graph}, there exists a path $\otherotherotherpath$ in $\somesetcycle$ between $\otherotherothervertex$ and $\sourceofsetcycles$ such that either ($\otherotherothervertex=\sourceofsetcycles$ and thus) $\otherotherotherpath$ is empty or $(\otherotherothervertex, \somecolor_\otherotherothervertex) \orderyeo[\concatpath{\otherotherotherpath}{(\sourceofsetcycles,\edgeofsetcycles,\targetofsetcycles)}] (\targetofsetcycles, \coloring(\edgeofsetcycles, \targetofsetcycles))$.

Let us prove that $\targetofsetcycles \notin \subpath{\somecycle}{\somevertex}{\otherothervertex}$, and that if $(\targetofsetcycles, \coloring(\edgeofsetcycles,\targetofsetcycles)) \sobfnbpath (\othervertex, \otherothercolor)$ for some vertex-color pair $(\othervertex, \otherothercolor)$, then $\othervertex\notin\subpath{\somecycle}{\somevertex}{\otherothervertex}$.
To this end, it suffices to show there is no simple, open or empty, cusp-free path $\otherpath$ with source $\targetofsetcycles$, target $\othervertex\in\concatpath{\subpath{\somecycle}{\otherotherothervertex}{\somevertex}}{\subpath{\somecycle}{\somevertex}{\otherothervertex}}$ and whose starting color (if any) is not $\coloring(\edgeofsetcycles, \targetofsetcycles)$.
We proceed by contradiction: take such a $\otherpath$.
Note that no vertex of $\otherpath$ belongs to $\somesetcycle$ for otherwise we contradict \cref{lem:cycles_no_ret_graph} by \cref{HG1}.
In particular $\othervertex\notin\{\otherothervertex,\otherotherothervertex,\sourceofsetcycles\}$, and $\otherpath$ share no vertex with  $\otherotherpath$ nor with $\otherotherotherpath$.
We then have a contradiction by \cref{lem:bjumping_gen}: either $\edgeofsetcycles$ belongs to a cusp-free cycle, contradicting the maximality of $\somesetcycle$, or there is a cycle starting with $\somevertex$, with no cusp at $\somevertex$ and with strictly less cusps than $\somecycle$, contradicting $\somecycle\in\mincycles{\somevertex}$.
We thus conclude that there is no such path as $\otherpath$.

Since $\targetofsetcycles\notin\subpath{\somecycle}{\somevertex}{\otherothervertex}$, by \cref{lem:concatsimplepathsprefix}, $\somepath=\concatpath{\subpath{\somecycle}{\somevertex}{\otherothervertex}}{\concatpath{\otherotherpath}{(\sourceofsetcycles,\edgeofsetcycles,\targetofsetcycles)}}$ is a simple open path (see \cref{fig:proof_no_splitting_no_max_graph_p}), that is cusp-free by construction and \cref{HG1}.
Thus, $(\somevertex, \somecolor) \sobfnbpath[\somepath] (\targetofsetcycles,\coloring(\edgeofsetcycles,\targetofsetcycles))$.
We cannot have $(\targetofsetcycles,\coloring(\edgeofsetcycles,\targetofsetcycles)) \sobfnbpath (\othervertex, \otherothercolor)$ with $\othervertex \in \somepath$: we already have $\othervertex \notin \subpath{\somecycle}{\somevertex}{\otherothervertex}$, and $\othervertex \notin \otherotherpath$ follows from $(\otherothervertex, \somecolor_\otherothervertex) \orderyeo[\concatpath{\otherotherpath}{(\sourceofsetcycles,\edgeofsetcycles,\targetofsetcycles)}] (\targetofsetcycles, \coloring(\edgeofsetcycles, \targetofsetcycles))$ or $\otherotherpath$ is empty.
Therefore, $(\somevertex, \somecolor) \orderyeo[\somepath] (\targetofsetcycles,\coloring(\edgeofsetcycles,\targetofsetcycles))$, where $(\targetofsetcycles,\coloring(\edgeofsetcycles,\targetofsetcycles))$ is a cusp-point by \cref{HG2}.
\end{proof}

\begin{proof}[Proof of \cref{th:MALLParamYeo}]
Take $(\somevertex, \somecolor)\in\somesetedge$ maximal for $\orderyeo$ (restricted to $\somesetedge$): $\somevertex$ is splitting.
Indeed, otherwise there would be some cusp-point $(\othervertex, \othercolor)$ such that $(\somevertex, \somecolor) \orderyeo (\othervertex, \othercolor)$ by \cref{prop:no_splitting_no_max_graph}.
As $\somesetedge$ dominates cusp-points, we would get $(\othervertex, \othercolor) \in \somesetedge$ or $(\somevertex, \somecolor) \orderyeo (\othervertex, \othercolor) \orderyeo (\otherothervertex, \otherothercolor)$ for some $(\otherothervertex, \otherothercolor) \in \somesetedge$, contradicting the maximality of $(\somevertex, \somecolor)$.
\end{proof}


\section{Multiplicative-Additive Proof Nets}
\label{section:additives}

We now adapt our proof of sequentialization of \cref{sec:seq} in presence of the additive connectives, using as proof nets the ones defined by Dominic Hughes and Rob van Glabbeek~\cite{mallpnlong}.
The core of the demonstration is the same as in the multiplicative case: a splitting vertex allows us to conclude by induction, and the main difficulty is finding such a splitting vertex.
The two methods used for multiplicative proof nets can be adapted to the multiplicative-additive case, whether finding a splitting $\parr$- or $\with$-vertex as in \cref{sec:seq:dummies}, or finding some kind of splitting vertex thanks to Yeo's theorem as in \cref{sec:seq:yeo}.
We adapt here only the second method, that yields a more general result.
Furthermore, we modify a little the definition of proof nets from~\cite{mallpnlong} to allow open hypotheses; we also add $\ax$-vertices to have a directed partial graph.

\subsection{Unit-Free Multiplicative-Additive Linear Logic with Mix}

The unit-free multiplicative-additive fragment of linear logic~\cite{ll} has formulas given by the following grammar, where $X$ belongs to a given enumerable set of \definitive{atoms}:
\begin{equation*}
  A~::=~X \mid X\orth \mid A\tensor A \mid A\parr A \mid A \with A \mid A \oplus A
\end{equation*}
The \definitive{dual} operator $(\_)\orth$ is extended to an involution on all formulas by De Morgan duality: $(X\orth)\orth=X$, $(A\tensor B)\orth=A\orth\parr B\orth$, $(A\parr B)\orth=A\orth\tensor B\orth$, $(A\oplus B)\orth=A\orth\with B\orth$ and $(A\with B)\orth=A\orth\oplus B\orth$.

Formally, as for multiplicative linear logic, we consider localized formulas so as to get a notion of occurrence of a formula $A$.

We consider the deduction system \MALLhmix\ given by \emph{cut-free\footnote{Our proof technique also applies in presence of the $(\cut)$ rule, but the definition of proof nets with open hypotheses and $(\cut)$ rules is quite technical.} open} derivations in unit-free multiplicative-additive linear logic with \emph{mix} rules and \emph{atomic axioms} (\ie\ introducing an atom and its dual):
\begin{equation*}
\begin{prooftree}
	\infer0[\ax]{\vdash X\orth, X}
\end{prooftree}
\qquad\quad
\begin{prooftree}
	\infer0[\hyp]{\vdash A}
\end{prooftree}
\end{equation*}
\begin{equation*}
\begin{prooftree}
	\hypo{\vdash A, \Gamma}
	\hypo{\vdash B, \Delta}
	\infer2[\tensor]{\vdash A\tensor B, \Gamma, \Delta}
\end{prooftree}
\qquad\quad
\begin{prooftree}
	\hypo{\vdash A, B, \Gamma}
	\infer1[\parr]{\vdash A\parr B, \Gamma}
\end{prooftree}
\qquad\quad
\begin{prooftree}
	\hypo{\vdash \Gamma}
	\hypo{\vdash \Delta}
	\infer2[\mix_2]{\vdash \Gamma, \Delta}
\end{prooftree}
\qquad\quad
\begin{prooftree}
	\infer0[\mix_0]{\vdash}
\end{prooftree}
\end{equation*}
\begin{equation*}
\begin{prooftree}
	\hypo{\vdash A, \Gamma}
	\hypo{\vdash B, \Gamma}
	\infer2[\with]{\vdash A\with B, \Gamma}
\end{prooftree}
\qquad\quad
\begin{prooftree}
	\hypo{\vdash A, \Gamma}
	\infer1[\oplus_1]{\vdash A\oplus B, \Gamma}
\end{prooftree}
\qquad\quad
\begin{prooftree}
	\hypo{\vdash B, \Gamma}
	\infer1[\oplus_2]{\vdash A\oplus B, \Gamma}
\end{prooftree}
\end{equation*}

The axiom expansion procedure of linear logic ensures that any provable sequent has a derivation using atomic axioms only~\cite{ll}.
When $\someproof$ is a derivation of $\vdash \Gamma$ whose $(\hyp)$ rules are on (localized) formulas $\Delta$, we write $\someproof$ is a derivation of $\Delta\vdash\Gamma$.

The main difference with the multiplicative fragment of linear logic is the $(\with)$ rule, which introduces some sharing of the context $\Gamma$.
From this comes the notion of a \emph{slice}~\cite{ll,pn} which is a partial derivation missing some additive components.
Slices are obtained by using the same rules as for derivations except for the $(\with)$ rule which is replaced by its two sliced versions:
\begin{equation*}
\begin{prooftree}
	\hypo{\vdash A, \Gamma}
	\infer1[\with_1]{\vdash A\with B, \Gamma}
\end{prooftree}
\qquad\quad
\begin{prooftree}
	\hypo{\vdash B, \Gamma}
	\infer1[\with_2]{\vdash A\with B, \Gamma}
\end{prooftree}
\end{equation*}

Furthermore, we introduce a restriction on the $(\hyp)$ rules.
\emph{Given a derivation $\someproof$ with an $(\hyp)$ rule on $\vdash A$, in every slice of $\someproof$ there must be an $(\hyp)$ rule on this occurrence $A$.}
For instance, the first following derivation respects this constraint, while the other two do not:
\begin{gather*}
\begin{prooftree}
	\infer0[\hyp]{\vdash A}
	\infer0[\ax]{\vdash X, X\orth}
	\infer2[\tensor]{\vdash A\tensor X, X\orth}
	\infer0[\hyp]{\vdash A}
	\infer0[\ax]{\vdash X, X\orth}
	\infer2[\tensor]{\vdash A\tensor X, X\orth}
	\infer2[\with]{\vdash A\tensor X, X\orth\with X\orth}
\end{prooftree}
\\
\begin{prooftree}
	\infer0[\hyp]{\vdash X}
	\infer0[\ax]{\vdash X,X\orth}
	\infer2[\mix_2]{\vdash X,X,X\orth}
	\infer0[\ax]{\vdash X,X\orth}
	\infer0[\ax]{\vdash Y,Y\orth}
	\infer1[\parr]{Y\parr Y\orth}
	\infer2[\mix_2]{\vdash X,Y\parr Y\orth,X\orth}
	\infer2[\with]{\vdash X,X\with(Y\parr Y\orth),X\orth}
\end{prooftree}
\qquad
\begin{prooftree}
	\infer0[\hyp]{\vdash X}
	\infer0[\hyp]{\vdash X}
	\infer2[\with]{\vdash X\with X}
\end{prooftree}
\end{gather*}

We restrict $(\hyp)$ rules this way since our sequentialization procedure always produces $(\hyp)$ rules of the aforementioned shape, and the notion of a proof net is simpler with this constraint.
Furthermore, we want a notion of subtitution using the $(\hyp)$ rule, as in multiplicative linear logic, which imposes some constraint on the usage of $(\hyp)$ rules.
If $\someproof_1$ is a derivation of $\Sigma\vdash\Gamma,A$ and $\someproof_2$ is a derivation of $A, \Theta\vdash\Delta$, the \definitive{substitution} of $\someproof_1$ for a hypothesis $A$ in $\someproof_2$ is a derivation of $\Sigma,\Theta\vdash\Gamma,\Delta$: it is obtained from $\someproof_2$ by replacing \emph{every} $(\hyp)$ rule on $\vdash A$ with $\someproof_1$, adding this way $\Gamma$ to all sequents of $\someproof_2$ below a $\vdash A$.
That the result is a derivation does not hold when $(\hyp)$ rules are unconstrained: see the derivation above on $X\vdash X,X\with(Y\parr Y\orth),X\orth$ in which we want to substitute a proof made of an $(\ax)$ rule on $X$.

Lastly, the \definitive{\mixretore\ reduction} is also defined here:
\[
  \begin{prooftree}[center=true]
    \hypo{\vdash \Gamma}
    \infer0[\mix_0]{\vdash}
    \infer2[\mix_2]{\vdash \Gamma}
  \end{prooftree}
  \quad\rightsquigarrow\quad
  \begin{prooftree}
    \hypo{\vdash \Gamma}
  \end{prooftree}
  \qquad\qquad
  \begin{prooftree}[center=true]
    \infer0[\mix_0]{\vdash}
    \hypo{\vdash \Gamma}
    \infer2[\mix_2]{\vdash \Gamma}
  \end{prooftree}
  \quad\rightsquigarrow\quad
  \begin{prooftree}
    \hypo{\vdash \Gamma}
  \end{prooftree}
\]
Again, it defines a confluent and strongly normalizing rewriting system on derivations.

\begin{lem}[\Mixretore\ Normal Forms]
\label{lem:mixretorenf_mall}
If $\someproof$ is a derivation from \MALLhmix\ in \mixretore\ normal form, either it is
\begin{prooftree}
	\infer0[\mix_0]{\vdash}
\end{prooftree},
or it does not contain the $(\mix_0)$ rule.
\end{lem}
\begin{proof}
As in \MLLhmix, $(\mix_2)$ is the only rule accepting an empty sequent as a premise.
\end{proof}

\subsection{Proof Nets for Unit-Free Multiplicative-Additive Linear Logic}\label{subsec:mall-pn}

We use the notion of unit-free MALL proof nets from~\cite{mallpnlong}, that we recall below -- we refer to that paper for more details and the intuitions guiding the definitions.
Alternative definitions of MALL proof nets exist: the original one from Girard~\cite{pn}, or others such as~\cite{jumpadd,conflictnet}.
Still, the definition we follow is one of the most satisfactory, from the point of view of canonicity and cut-elimination for instance:
see~\cite{mallpnlong,mallpncom}, or the introduction of~\cite{conflictnet} for a comparison of alternative definitions.
Nevertheless, we do not take exactly this definition: we allow open hypotheses \emph{shared among all slices} as for derivations; we also add $\ax$-vertices to have a directed partial graph -- otherwise, axiom links are undirected -- and put jump edges from these vertices instead of from leaves.

We always identify (an occurrence of) a formula $A$ with its \definitive{syntactic tree}, with as internal vertices its connectives and as leaves its atoms, edges being oriented towards the root and with a partial edge of source the root and undefined target.
This last edge is the edge of the syntactic tree \definitive{associated with} $A$.
A \definitive{(shared open) hypothesis} on a sequent $\vdash \Gamma$ is a vertex in (a syntactic tree of) $\Gamma$.
Given a set of hypotheses $\Delta$, we denote by $\Delta \vdash \Gamma$ the sequent $\vdash \Gamma$ with these hypotheses $\Delta$.
We silently consider $\Delta \vdash \Gamma$ as the following directed partial graph, called its \definitive{syntactic forest}: take the union of the syntactic trees of the formulas of $\Gamma$ and remove all sub-trees rooted on vertices in $\Delta$, \emph{keeping their partial edges}.
Notice the difference between hypotheses and leaves: an hypothesis corresponds to a partial edge with no source.
For instance, consider $Z\with Z\orth \vdash X\oplus(Y\with Y), (Y\orth\with X\orth) \tensor (Z\with Z\orth)$, whose syntactic forest is depicted on \cref{fig:ex_syntactic_forest}.
We will instantiate the concepts defined in this part on this sequent with hypotheses.

\begin{figure}
\centering
\begin{tikzpicture}
\begin{genscope}
	\node (X) at (-1.75,1) {$Y$};
	\node (Y^) at (-.25,1) {$Y$};
	\node (Y) at (.75,1) {$Y\orth$};
	\node (X^) at (2.25,1) {$X\orth$};
	\coordinate (Z) at (2.75,-.25);
	\node (W) at (-1,0) {$\with$};
	\node (P) at (-1.75,-1) {$\oplus$};
	\node (Pa) at (-2.5,0) {$X$};
	\node (O) at (1.5,0) {$\with$};
	\node (T) at (2.25,-1) {$\tensor$};

	\path (X) edge (W);
	\path (Y^) edge (W);
	\path (Y) edge (O);
	\path (X^) edge (O);
	\path (O) edge (T);
	\path (Z) edge (T);
	\path (W) edge (P);
	\path (Pa) edge (P);
	\path (P) edge ++(0,-.75);
	\path (T) edge ++(0,-.75);
\end{genscope}
\end{tikzpicture}
\caption{Syntactic forest of $Z\with Z\orth \vdash X\oplus(Y\with Y), (Y\orth\with X\orth) \tensor (Z\with Z\orth)$}
\label{fig:ex_syntactic_forest}
\end{figure}
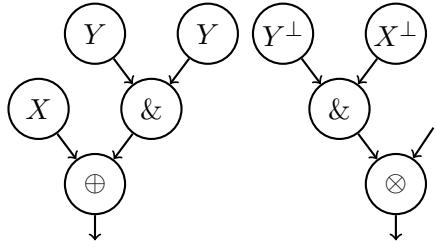

Given a vertex $\somevertex$ in $\Delta \vdash \Gamma$, we call \definitive{premise} of $\somevertex$ an edge of target $\somevertex$ in the syntactic forest.
Similarly, a \definitive{conclusion} of $\somevertex$ is an edge of source $\somevertex$ in the syntactic forest.
An \definitive{argument sub-tree} of (the syntactic tree of) $\somevertex$ is a sub-tree whose partial edge is a premise of $\somevertex$; possibly, an argument sub-tree may consist of a partial edge only.
For instance, the two argument sub-trees of the $\tensor$-vertex in \cref{fig:ex_syntactic_forest} are the tree of $Y\orth\with X\orth$ and the edge with no source.

An \definitive{additive resolution} of a sequent with hypotheses $\Delta \vdash \Gamma$ is any result of deleting one argument sub-tree of each additive ($\with$ or $\oplus$) vertex such that \emph{no partial edge of vertices in $\Delta$ is deleted this way} -- \ie\ all partial edges having undefined source should be kept.
A \definitive{$\with$-resolution} of a sequent with hypotheses $\Delta \vdash \Gamma$ is any result of deleting one argument sub-tree of each $\with$-vertex.
For example, $Z\with Z\orth \vdash X\oplus~, (~\with X\orth) \tensor (Z\with Z\orth)$ is one of the six additive resolutions of $Z\with Z\orth \vdash X\oplus(Y\with Y), (Y\orth\with X\orth) \tensor (Z\with Z\orth)$, while $Z\with Z\orth \vdash X\oplus(~\with Y), (~\with X\orth) \tensor (Z\with Z\orth)$ is one of its four $\with$-resolutions.
Notice the difference on hypotheses between additive and $\with$-resolutions: each partial edge with no source must appear in all additive resolutions, whereas it may not belong to some $\with$-resolutions.

An \definitive{axiom link} on $\Delta \vdash \Gamma$ is an unordered pair of complementary leaves $\{X,X\orth\}$ in $\Delta \vdash \Gamma$ for some atom $X$.
A \definitive{linking} $\somelinking$ on $\Delta\vdash\Gamma$ is a set of axiom links on $\Delta \vdash \Gamma$ that forms a partition of the set of leaves of an additive resolution of $\Delta \vdash \Gamma$; this (unique) additive resolution is denoted $\Delta \vdash \Gamma\upharpoonright\somelinking$.
In this case, we say $\somelinking$ is a linking \definitive{on} the additive resolution $\Delta \vdash \Gamma\upharpoonright\somelinking$.
Similarly, we say $\somelinking$ is \definitive{on} a given $\with$-resolution if each leaf of $\somelinking$ is also a leaf of that $\with$-resolution (note that a given linking may be on several $\with$-resolutions, since a deleted argument sub-tree of a $\oplus$-vertex might contain a $\with$-vertex).
For instance, on the leftmost partial graph of \cref{fig:ex_pn}, the red $\ax$-vertex forms a linking $\somelinking_1 = \{\{X,X\orth\}\}$, whose
additive resolution is $Z\with Z\orth \vdash X\oplus~, (~\with X\orth) \tensor (Z\with Z\orth)$;
it is also a linking on the $\with$-resolutions $Z\with Z\orth \vdash X\oplus(Y\with~), (~\with X\orth) \tensor (Z\with Z\orth)$ and $Z\with Z\orth \vdash X\oplus(Y\with~), (~\with X\orth) \tensor (Z\with Z\orth)$.

A set of linkings $\somesetlinking$ on $\Delta \vdash \Gamma$ \definitive{toggles} a $\with$-vertex $\somewith$ if both premises of $\somewith$ are in ${\Delta \vdash \Gamma\upharpoonright\somesetlinking\coloneqq\bigcup_{\somelinking\in\somesetlinking} \Delta \vdash \Gamma\upharpoonright\somelinking}$.
We say an axiom link $\somelink$ \definitive{depends} on a $\with$-vertex $\somewith$ in $\somesetlinking$ if there exist $\somelinking,\somelinking'\in\somesetlinking$ such that $\somelink\in\somelinking\backslash\somelinking'$ and $\somewith$ is the only $\with$-vertex toggled by $\{\somelinking,\somelinking'\}$.
Looking at our running example on \cref{fig:ex_pn}, with the linkings $\somelinking_1=\{\{X,X\orth\}\}$ and $\somelinking_2=\{\{Y,Y\orth\}\}$, a unique $\with$-vertex is toggled by $\{\somelinking_1,\somelinking_2\}$: the rightmost one.
Furthermore, both axiom links depend on this $\with$-vertex for $\somelinking_1$ and $\somelinking_2$ contain only different pairs.

The directed partial graph $\G_\somesetlinking$ is defined as $\Delta \vdash \Gamma\upharpoonright\somesetlinking$ with for each pair $\{X,X\orth\}$ in $\bigcup\somesetlinking$ an $\ax$-vertex $\somevertex$ with two out-edges to $X$ and $X\orth$, and enriched with jump edges $\somevertex\rightarrow\somewith$ for each $\with$-vertex $\somewith$ such that $\{X,X\orth\}$ depends on $\somewith$ in $\somesetlinking$.
When $\somesetlinking=\{\somelinking\}$ is composed of a single linking, we shall simply denote $\G_\somelinking=\G_{\{\somelinking\}}$ -- which is the partial graph $\Delta \vdash \Gamma\upharpoonright\somelinking$ with the $\ax$-vertices from $\somelinking$ and no jump edge.
Going back to our example, the partial graphs $\G_{\somelinking_1}$, $\G_{\somelinking_2}$, $\G_{\somelinking_3}$ and $\G_{\{\somelinking_1,\somelinking_2,\somelinking_3\}}$ are illustrated on \cref{fig:ex_pn}.

When we write a $\pw$-vertex, we mean a $\parr$- or $\with$-vertex, and similarly for other combinations of labels.
A \definitive{switch edge} of a $\pw$-vertex $\somevertex$ is a premise of $\somevertex$ or a jump edge to $\somevertex$.
A \definitive{switching cycle} is a cycle with at most one switch edge of each $\pw$-vertex.

\begin{defi}[Proof net]\label{def:pn}
A \definitive{proof net} $\somepn$ on a sequent with hypotheses $\Delta \vdash \Gamma$ is a set of linkings satisfying:
\begin{enumeratecref}[start=1,label=\textnormal{(P\arabic*)},ref=\textnormal{(P\arabic*)}]
\item\label{item:P1}
\emph{Resolution:}
There is exactly one linking of $\somepn$ on each $\with$-resolution of $\Delta \vdash \Gamma$.
\item\label{item:P2}
\emph{MLL:}
For every linking $\somelinking\in\somepn$, $\G_\somelinking$ has no switching cycle.
\item\label{item:P3}
\emph{Toggling:}
Every set $\somesetlinking\subseteq\somepn$ of two or more linkings toggles a $\with$-vertex that is in no switching cycle of $\G_\somesetlinking$.
\end{enumeratecref}
\end{defi}

\begin{rem}
A notion of connected proof net for MALL without the mix rules is obtained by replacing the requirement \cref{item:P2} by the following one.
A \definitive{$\parr$-switching} of a linking $\somelinking$ is any sub-graph of $\G_\somelinking$ obtained by disconnecting one of the two premises of each $\parr$-vertex; denoting by $\phi$ this choice of edges, the partial graph it yields is $\G_\phi$.
For example, a cycle in a $\parr$-switching is a switching cycle, as in $\G_\phi$ all $\pw$-vertices have a unique switch edge.
\begin{enumeratecref}[start=2,label=\textnormal{({P\arabic*}\textsuperscript{c})},ref=\textnormal{({P\arabic*}\textsuperscript{c})}]
\item\label{item:P2c}
\emph{MLL:}
For every $\parr$-switching $\phi$ of every linking $\somelinking\in\somepn$, $\G_\phi$ is acyclic and connected.
\end{enumeratecref}
Observe that \cref{item:P2c} implies \cref{item:P2}, so that a connected proof net is a particular kind of proof net.
As for MLL, one can deduce connected sequentialization from sequentialization.
\end{rem}

These conditions are called the \definitive{correctness criterion}.
Condition \cref{item:P1} is a correctness criterion for ALL proof nets~\cite{mallpnlong} and \cref{item:P2} (\resp\ \cref{item:P2c}) is the Danos-Regnier criterion for MLL proof nets (\resp\ connected MLL proof nets)~\cite{structmult}.
However, \cref{item:P1,item:P2} together are insufficient for cut-free MALL proof nets, hence the last condition \cref{item:P3} taking into account interactions between the slices -- see also~\cite{jumpadd} for a similar condition.
Sets composed of a single linking $\somelinking$ are not considered in \cref{item:P3}, for by \cref{item:P2} the partial graph $\G_\somelinking$ has no switching cycle.
We say that $\somepn$ is a \definitive{proof structure} if it satisfies \cref{item:P1}.
One can check that our example $\{\somelinking_1,\somelinking_2,\somelinking_3\}$ on \cref{fig:ex_pn} is a (connected) proof net.

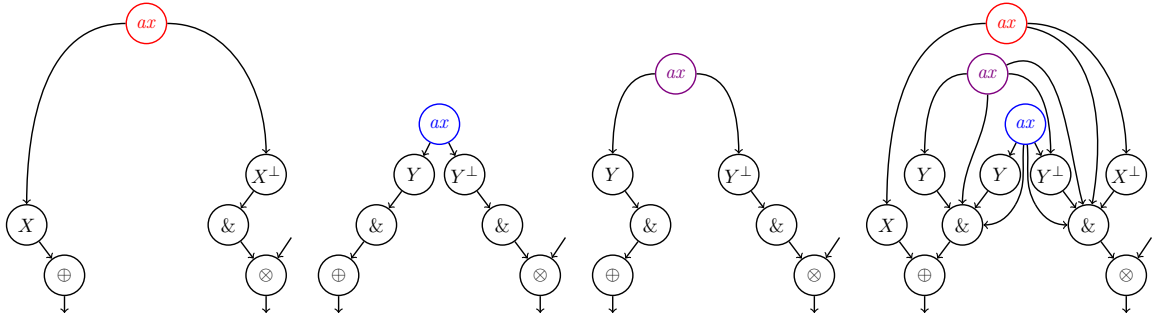
\begin{figure}
\begin{adjustbox}{}
\begin{tikzpicture}
\begin{genscope}
	\node (X^) at (2.25,1) {$X\orth$};
	\coordinate (Z) at (2.75,-.25);
	\node (P) at (-1.75,-1) {$\oplus$};
	\node (Pa) at (-2.5,0) {$X$};
	\node (O) at (1.5,0) {$\with$};
	\node (T) at (2.25,-1) {$\tensor$};

	\path (X^) edge (O);
	\path (O) edge (T);
	\path (Z) edge (T);
	\path (Pa) edge (P);
	\path (P) edge ++(0,-.75);
	\path (T) edge ++(0,-.75);
\end{genscope}
\begin{genscope}[red]
	\node (ax') at (-0.125,4) {$\ax$};
\end{genscope}
\begin{genscope}
	\path[out=180,in=90] (ax') edge (Pa);
	\path[out=0,in=90] (ax') edge (X^);
\end{genscope}
\end{tikzpicture}
\quad
\begin{tikzpicture}
\begin{genscope}
	\node (Y^) at (-.25,1) {$Y$};
	\node (Y) at (.75,1) {$Y\orth$};
	\coordinate (Z) at (2.75,-.25);
	\node (W) at (-1,0) {$\with$};
	\node (P) at (-1.75,-1) {$\oplus$};
	\node (O) at (1.5,0) {$\with$};
	\node (T) at (2.25,-1) {$\tensor$};

	\path (Y^) edge (W);
	\path (Y) edge (O);
	\path (O) edge (T);
	\path (Z) edge (T);
	\path (W) edge (P);
	\path (P) edge ++(0,-.75);
	\path (T) edge ++(0,-.75);
\end{genscope}
\begin{genscope}[blue]
	\node (axY) at (.25,2) {$\ax$};
\end{genscope}
\begin{genscope}
	\path (axY) edge (Y);
	\path (axY) edge (Y^);
\end{genscope}
\end{tikzpicture}
\quad
\begin{tikzpicture}
\begin{genscope}
	\node (X) at (-1.75,1) {$Y$};
	\node (Y) at (.75,1) {$Y\orth$};
	\coordinate (Z) at (2.75,-.25);
	\node (W) at (-1,0) {$\with$};
	\node (P) at (-1.75,-1) {$\oplus$};
	\node (O) at (1.5,0) {$\with$};
	\node (T) at (2.25,-1) {$\tensor$};

	\path (X) edge (W);
	\path (Y) edge (O);
	\path (O) edge (T);
	\path (Z) edge (T);
	\path (W) edge (P);
	\path (P) edge ++(0,-.75);
	\path (T) edge ++(0,-.75);
\end{genscope}
\begin{genscope}[violet]
	\node (axX) at (-0.5,3) {$\ax$};
\end{genscope}
\begin{genscope}
	\path[out=0,in=90] (axX) edge (Y);
	\path[out=180,in=90] (axX) edge (X);
\end{genscope}
\end{tikzpicture}
\quad
\begin{tikzpicture}
\begin{genscope}
	\node (X) at (-1.75,1) {$Y$};
	\node (Y^) at (-.25,1) {$Y$};
	\node (Y) at (.75,1) {$Y\orth$};
	\node (X^) at (2.25,1) {$X\orth$};
	\coordinate (Z) at (2.75,-.25);
	\node (W) at (-1,0) {$\with$};
	\node (P) at (-1.75,-1) {$\oplus$};
	\node (Pa) at (-2.5,0) {$X$};
	\node (O) at (1.5,0) {$\with$};
	\node (T) at (2.25,-1) {$\tensor$};

	\path (X) edge (W);
	\path (Y^) edge (W);
	\path (Y) edge (O);
	\path (X^) edge (O);
	\path (O) edge (T);
	\path (Z) edge (T);
	\path (W) edge (P);
	\path (Pa) edge (P);
	\path (P) edge ++(0,-.75);
	\path (T) edge ++(0,-.75);
\end{genscope}
\begin{genscope}[violet]
	\node (axX) at (-0.5,3) {$\ax$};
\end{genscope}
\begin{genscope}[blue]
	\node (axY) at (.25,2) {$\ax$};
\end{genscope}
\begin{genscope}[red]
	\node (ax') at (-0.125,4) {$\ax$};
\end{genscope}
\begin{genscope}
	\path[out=0,in=90] (axX) edge (Y);
	\path[out=180,in=90] (axX) edge (X);
	\path (axY) edge (Y);
	\path (axY) edge (Y^);
	\path[out=180,in=90] (ax') edge (Pa);
	\path[out=0,in=90] (ax') edge (X^);
	\path[out=-90,in=90] (axX) edge (W);
	\path[out=25,in=105] (axX) edge (O);
	\path[out=-95,in=60] (axY) edge[-] (.05,.4);
	\path[out=-120,in=0] (.05,.4) edge (W);
	\path[out=-85,in=120] (axY) edge[-] (.45,.4);
	\path[out=-60,in=180] (.45,.4) edge (O);
	\path[out=-10,in=80] (ax') edge (O);
\end{genscope}
\end{tikzpicture}
\end{adjustbox}
\caption{Partial graphs from an example of a proof net $\{{\color{red}\somelinking_1},{\color{blue}\somelinking_2},{\color{violet}\somelinking_3}\}$ on $Z\with Z\orth \vdash X\oplus(Y\with Y), (Y\orth\with X\orth) \tensor (Z\with Z\orth)$: from left to right $\G_{{\color{red}\somelinking_1}}$, $\G_{{\color{blue}\somelinking_2}}$, $\G_{{\color{violet}\somelinking_3}}$ and $\G_{\{{\color{red}\somelinking_1},{\color{blue}\somelinking_2},{\color{violet}\somelinking_3}\}}$}
\label{fig:ex_pn}
\end{figure}

\begin{rem}
Condition \cref{item:P1} implies that hypotheses belong to all $\with$-resolutions.
Indeed, there is a linking on each $\with$-resolution, whose additive-resolution (that contains all partial edges with undefined source) is included in this $\with$-resolution.
This means a sequent where an hypothesis belongs to only some $\with$-resolutions, such as $X \vdash X\with (Y\parr Y\orth)$, cannot be made into a proof net.
\end{rem}

We will need a few more definitions.
Fix a set of linkings $\somesetlinking$ on the sequent with hypothesis $\Delta \vdash \Gamma$, and a vertex $\somevertex$ in $\G_\somesetlinking$.
\begin{itemize}
\item
The vertex $\somevertex$ is \definitive{terminal} if it is not a leaf and it is the root of a syntactic tree of $\Delta \vdash \Gamma$, or if it is an $\ax$-vertex with out-edges to leaves $X\orth$ and $X$ that are the roots of two syntactic trees of $\Delta \vdash \Gamma$.
\item
An additive vertex is \definitive{unary} (\resp\ \definitive{binary}) if it has exactly one (\resp\ two) premises in $\G_\somesetlinking$.
\item
A non-leaf vertex $\somevertex$ is \definitive{splitting} when:
\begin{itemize}
\item
$\somevertex$ is an $\ax\backslash\tensor\backslash\oplus$-vertex which is not in any cycle of $\G_\somepn$;
\item
$\somevertex$ is a $\pw$-vertex whose conclusion is not in any cycle of $\G_\somepn$;
\end{itemize}
(see \cref{fig:splitting} for an illustration, with a $\oplus$-vertex similar to a $\tensor$-vertex and a $\with$-vertex similar to a $\parr$-vertex; a leaf is never splitting).
\end{itemize}

\subsection{Desequentialization}

We desequentialize a \MALLhmix\ derivation $\someproof$ of $\vdash \Gamma$ with hypotheses $\Delta$ into a proof net on $\Delta \vdash \Gamma$.
We do so by induction on $\someproof$ using the steps detailed on \cref{fig:translation_inductive}, following~\cite{mallpnlong} with the notation $\somepn\triangleright \Delta\vdash\Gamma$ for ``$\somepn$ is a set of linkings on $\Delta\vdash\Gamma$''.
We note $\deseq(\someproof)$ the set of linkings $\someproof$ desequentializes to.
One can easily prove that $\deseq(\someproof)$ is a proof net by proceeding by induction as in~\cite{mallpnlong}.
For instance, the set $\{\{\{X, X\orth\}\}\}$, composed of a unique linking $\{\{X, X\orth\}\}$, which itself has a unique axiom link $\{X, X\orth\}$, is a proof net on the sequent $\vdash X\orth,X$: this sequent has a unique $\with$-resolution and there is no (switching) cycle in $\G_{\{\{X, X\orth\}\}}$.

\begin{figure}
\centering
\begin{equation*}
\begin{prooftree}
	\infer0[\ax]{\{\{\{X, X\orth\}\}\}~\triangleright~\vdash X\orth,X}
\end{prooftree}
\qquad\quad
\begin{prooftree}
	\infer0[\hyp]{\{\emptyset\}~\triangleright~A \vdash A}
\end{prooftree}
\end{equation*}
\begin{equation*}
\begin{prooftree}
	\hypo{\theta~\triangleright~\Sigma\vdash A,\Gamma}
	\hypo{\vartheta~\triangleright~\Phi\vdash B,\Delta}
	\infer2[\tensor]{\{\lambda\cup\mu \mid \lambda\in\theta,\mu\in\vartheta\}~\triangleright~\Sigma,\Phi\vdash A\tensor B,\Gamma,\Delta}
\end{prooftree}
\qquad\quad
\begin{prooftree}
	\hypo{\theta~\triangleright~\Delta\vdash A,B,\Gamma}
	\infer1[\parr]{\theta~\triangleright~\Delta\vdash A\parr B,\Gamma}
\end{prooftree}
\end{equation*}
\begin{equation*}
\begin{prooftree}
	\hypo{\theta~\triangleright~\Sigma\vdash\Gamma}
	\hypo{\vartheta~\triangleright~\Phi\vdash\Delta}
	\infer2[\mix_2]{\{\lambda\cup\mu \mid \lambda\in\theta,\mu\in\vartheta\}~\triangleright~\Sigma,\Phi\vdash\Gamma,\Delta}
\end{prooftree}
\qquad\quad
\begin{prooftree}
	\infer0[\mix_0]{\{\emptyset\}~\triangleright~\vdash}
\end{prooftree}
\end{equation*}
\begin{equation*}
\begin{prooftree}
	\hypo{\theta~\triangleright~\Delta\vdash A,\Gamma}
	\hypo{\vartheta~\triangleright~\Delta\vdash B,\Gamma}
	\infer2[\with]{\theta\cup\vartheta~\triangleright~\Delta\vdash A\with B,\Gamma}
\end{prooftree}
\end{equation*}
\begin{equation*}
\begin{prooftree}
	\hypo{\theta~\triangleright~\Delta\vdash A,\Gamma}
	\infer1[\oplus_1]{\theta~\triangleright~\Delta\vdash A\oplus B,\Gamma}
\end{prooftree}
\qquad\quad
\begin{prooftree}
	\hypo{\theta~\triangleright~\Delta\vdash B,\Gamma}
	\infer1[\oplus_2]{\theta~\triangleright~\Delta\vdash A\oplus B,\Gamma}
\end{prooftree}
\end{equation*}
\caption{Inductive definition of the translation of \MALLhmix\ derivations to sets of linkings}
\label{fig:translation_inductive}
\end{figure}

\begin{lem}[Desequentialization of a substitution]
\label{lem:deseqsub_mall}
If $\someproof$ is the substitution of a derivation $\someproof_1$ for a hypothesis $A$ in a derivation $\someproof_2$, then $\deseq(\someproof)=\{\somelinking_1\cup\somelinking_2 \mid \somelinking_1\in\deseq(\someproof_1),\somelinking_2\in\deseq(\someproof_2)\}$.
\end{lem}
\begin{proof}
By induction on the derivation $\someproof_1$.
\end{proof}

\section{Sequentialization}\label{sec:seqmall}

Similarly to proof nets of multiplicative linear logic, a key result about proof nets of multiplicative-additive linear logic is that the desequentialization function is surjective:

\begin{thm}[Sequentialization]\label{th:seq_mall}
Given a proof net $\somepn$, there exists a derivation $\someproof$ in \MALLhmix\ such that $\somepn=\deseq(\someproof)$; $\someproof$ is called a \definitive{sequentialization} of $\somepn$.
\end{thm}

Exactly as for MLL proof nets, there are variants of this result: when restricted to proof nets with no hypothesis, one sequentializes to derivations with no $(\hyp)$ rule; when requiring the correctness criterion \cref{item:P2c}, one sequentializes to derivations with neither $(mix_2)$ nor $(\mix_0)$ rule; or one can also add the $(\cut)$ rule.
Such sequentialization results can be deduced from our parametrized and local version of Yeo’s theorem for MALL (\cref{th:MALLParamYeo}) in a uniform and modular way, in the same spirit as what we did for MLL proof nets in \cref{sec:seq}.

We proceed as follows.
\Cref{th:MALLParamYeo} allows us to find a splitting vertex which is the hard intermediate result leading to sequentialization; for different values of the parameter, we find different kinds of splitting vertices (\cref{sec:mall_splitting}).
Then, one can deduce sequentialization from the existence of any splitting vertex, albeit with more difficulties than in the multiplicative case (\cref{sec:mall_seq_from_splitting}).

\subsection{Finding a Splitting Vertex}
\label{sec:mall_splitting}

We begin with the core of the proof: finding a splitting vertex.
As far as we know, all the (few) proofs of sequentialization for the proof nets of~\cite{mallpnlong} follow this approach~\cite{mallpnlong,DiGuardiaLaurent22}.
We could transpose directly the content of \cref{sec:yeo_mall} in the framework of MALL proof nets, simplifying some concepts and proofs in the process, as we did for MLL in \cref{sec:seq:dummies}.
We choose not to do so here, as the proof is quite longer.
Instead, we deduce the existence of various kinds of splitting vertices from a generalization of Yeo's theorem, similarly to what we did in \cref{sec:seq:yeo}.
The main difference between MLL and MALL is the presence of switching cycles, that prevents us from applying \cref{th:ParamLocalYeo}, and which is the reason we need the more general \cref{th:MALLParamYeo}.
For the intuition, a switching cycle can be seen as a zone of contradictory dependencies (\ie\ different choices of premises for a $\oplus$, or an order of sequentialization that depends on slices) that is resolved when projecting on a slice.
We have to find a jump edge from such a cycle to a $\with$-vertex which causes these dependencies, across all slices.
The parallel in \cref{th:MALLParamYeo} is the edge $\edgeofsetcycles(\somesetcycle)$, allowing to go out of the cycles $\somesetcycle$.
In order to apply \cref{th:MALLParamYeo}, we need to give a local coloring and an exit function with the wanted hypotheses.

Concerning the local coloring of (the partial graph induced by) a proof net, we proceed as in MLL.

\begin{defi}
\label{def:wellcolored_mall}
A proof net is \definitive{well-colored} when it is equipped with a local coloring (of its partial graph) with at least three colors and such that:
\begin{itemize}
\item
for a $\tensor$- or a $\oplus$-vertex $\somevertex$ with premises $\someedge_1$ and $\someedge_2$ and conclusion $\otheredge$, $\coloring(\someedge_1, \somevertex)$, $\coloring(\someedge_2, \somevertex)$ and $\coloring(\otheredge, \somevertex)$ are pairwise distinct;
\item
for a $\parr$-vertex $\somevertex$ with premises $\someedge_1$ and $\someedge_2$ and conclusion $\otheredge$, $\coloring(\someedge_1, \somevertex) = \coloring(\someedge_2, \somevertex) \neq \coloring(\otheredge, \somevertex)$;
\item
for a $\with$-vertex with premises $\someedge_1$ and $\someedge_2$, conclusion $\otheredge$ and incoming jump edges $(\otherotheredge_i)_i$, for all $i$ $\coloring(\someedge_1, \somevertex) = \coloring(\someedge_2, \somevertex) = \coloring(\otherotheredge_i, \somevertex) \neq \coloring(\otheredge, \somevertex)$;
\item
for an $\ax$-vertex $\somevertex$ with out-edges $\someedge_1$ and $\someedge_2$ to leaves, and outgoing jump edges $(\otherotheredge_i)_i$, $\coloring(\someedge_1, \somevertex)$, $\coloring(\someedge_2, \somevertex)$ and all $\coloring(\otherotheredge_i, \somevertex)$ are pairwise distinct;
\item
for a leaf $\somevertex$ with conclusion $\otheredge$ and incoming edges $(\otherotheredge_i)_i$, $\coloring(\otheredge, \somevertex)$ and all $\coloring(\otherotheredge_i, \somevertex)$ are pairwise distinct.
\end{itemize}
\end{defi}

It is always possible to turn a proof net into a well-colored proof net:
\begin{itemize}
\item
use different colors for all edges incident to an $\ax$-, $\tensor$- or $\oplus$-vertex;
\item
use \textcolor{red}{solid} for the in-edges of each $\pw$-vertex, and \textcolor{blue}{dashed} for its conclusion;
\end{itemize}
As an example, see \cref{fig:ex_pn_colored} for a local coloring of the proof net $\G_{\{\somelinking_1,\somelinking_2,\somelinking_3\}}$ from \cref{fig:ex_pn}.

\begin{figure}
\begin{adjustbox}{}
\begin{tikzpicture}
\begin{genscope}
	\node (X) at (-1.75,1) {$Y$};
	\node (Y^) at (-.25,1) {$Y$};
	\node (Y) at (.75,1) {$Y\orth$};
	\node (X^) at (2.25,1) {$X\orth$};
	\coordinate (Z) at (2.75,-.25);
	\node (W) at (-1,0) {$\with$};
	\node (P) at (-1.75,-1) {$\oplus$};
	\node (Pa) at (-2.5,0) {$X$};
	\node (O) at (1.5,0) {$\with$};
	\node (T) at (2.25,-1) {$\tensor$};
	\node (axX) at (-0.5,3) {$\ax$};
	\node (axY) at (.25,2) {$\ax$};
	\node (ax') at (-0.125,4) {$\ax$};
\end{genscope}
\begin{scope}[every edge/.style={draw=none},every node/.style={inner sep=0,draw=none,minimum size=0}]
	\path (X) edge node[midway] (t1) {} (W);
	\path (Y^) edge node[midway] (t2) {} (W);
	\path (Y) edge node[midway] (t3) {} (O);
	\path (X^) edge node[midway] (t4) {} (O);
	\path (O) edge node[midway] (t5) {} (T);
	\path (W) edge node[midway] (t6) {} (P);
	\path (Pa) edge node[midway] (t7) {} (P);
	\path[out=-90,in=90] (axX) edge node[midway] (p6) {} (W);
	\coordinate (p8) at (.05,.4);
	\coordinate (p9) at (.45,.4);
\end{scope}
\begin{genscope}[red][-,solid]
	\path (t1) edge (W);
	\path (t2) edge (W);
	\path (t3) edge (O);
	\path (t4) edge (O);
	\path (t5) edge (T);
	\path (t7) edge (P);
	\path[out=-95,in=60] (axY) edge[-] (p8);
	\path[out=-120,in=0] (p8) edge (W);
	\path[out=-60,in=180] (p9) edge (O);
	\path[out=-10,in=80] (ax') edge (O);
	\path[out=25,in=105] (axX) edge (O);
	\path[out=-90,in=70] (p6) edge (W);
\end{genscope}
\begin{genscope}[blue][-,densely dashed]
	\path (X) edge (t1);
	\path (Y^) edge (t2);
	\path (Y) edge (t3);
	\path (X^) edge (t4);
	\path (O) edge (t5);
	\path (W) edge (t6);
	\path (Pa) edge (t7);
	\path (P) edge ++(0,-.75);
	\path (T) edge ++(0,-.75);
	\path[out=-85,in=120] (axY) edge[-] (p9);
	\path[out=-110,in=90] (axX) edge (p6);
\end{genscope}
\begin{genscope}[violet][-,densely dotted]
	\path (t6) edge (P);
	\path (Z) edge (T);
	\path[out=0,in=90] (axX) edge (Y);
	\path (axY) edge (Y^);
	\path[out=180,in=90] (ax') edge (Pa);
\end{genscope}
\begin{genscope}[olive][-,densely dash dot]
	\path[out=180,in=90] (axX) edge (X);
	\path (axY) edge (Y);
	\path[out=0,in=90] (ax') edge (X^);
\end{genscope}
\end{tikzpicture}
\end{adjustbox}
\caption{Well-colored proof net $\G_{\{\somelinking_1,\somelinking_2,\somelinking_3\}}$ from \cref{fig:ex_pn}}
\label{fig:ex_pn_colored}
\end{figure}

Note the cusp-points of a well-colored proof net are exactly the pairs $(\somevertex,\somecolor)$ where $\somevertex$ is a $\pw$-vertex and $\somecolor$ is the color associated with its in-edges.
Then, by \cref{lem:localglobal}, a cusp-free path in a proof net is nothing but a switching path; and then a non-leaf vertex is splitting in the sense of \cref{subsec:mall-pn} if and only if it is splitting in the sense of \cref{sec:localcolor}.

Concerning the exit function, it can be obtained through a basic result from~\cite{mallpnlong}, whose proof is fairly short -- less than one page, with a more detailled proof also available in~\cite{phddiguardia}.
Remarkably, it is the only use of the correctness criterion we need to prove sequentialization, except to get that sub-graphs stay correct.
For completeness' sake, and as we slightly modified the definition of proof nets from~\cite{mallpnlong}, we give a reworked proof of this key result.

\begin{lemC}[{\cite[Lemma~4.32]{mallpnlong}}]\label{lem:4.32}
Given a proof net $\somepn$, every non-empty union $\somesetcycle$ of switching cycles of $\G_\somepn$ has a jump edge out of it: for some $\somevertex\in\somesetcycle$ and some $\with$-vertex $\somewith\notin\somesetcycle$, there is a jump edge $\somevertex\to\somewith$ in $\G_\somepn$.
\end{lemC}
\begin{proof}
For a set of linking $\somesetlinking\subseteq\somepn$, we set $W(\somesetlinking)$ the set of $\with$-vertices whose right premises are in the additive resolution of \emph{no} linking of $\somesetlinking$.
Consider $\somesetlinking\subseteq\somepn$ with $\G_\somesetlinking$ containing $\somesetcycle$, chosen such that $W(\somesetlinking)$ is maximal for the inclusion.
By \cref{item:P2}, if $\somesetlinking$ were a singleton then $\G_{\somesetlinking}$ could not contain the non-empty $\somesetcycle$, a contradiction.
Hence, $\somesetlinking$ is of size at least two and we can apply \cref{item:P3} on it: there is a $\with$-vertex $\somewith$ toggled by $\somesetlinking$ that is not in any switching cycle of $\G_{\somesetlinking}$.
In particular, $\somewith\notin\somesetcycle$ and $\somewith\notin W(\somesetlinking)$.

We will exhibit a jump edge to $\somewith$.
To this end, we define a new set of linkings to use our maximality hypothesis on $W(\somesetlinking)$.
For a linking $\somelinking\in\somesetlinking$, we define $\somelinking^\somewith$ as follows.
Consider $\someres(\somelinking)$ the $\with$-resolution keeping the left premise of $\somewith$, making the same choices of premises than the additive-resolution of $\somelinking$ on the $\with$-vertices the latter contains (except possibly $\somewith$), and keeping the left premise of all other $\with$-vertices.
By \cref{item:P1}, there is a unique linking of $\somepn$ on $\someres(\somelinking)$: this is $\somelinking^\somewith$.
Pose $\somesetlinking^\somewith \eqdef \{\somelinking^\somewith \mid \somelinking \in \somesetlinking\}$.
Observe that $W(\somesetlinking^\somewith)\supseteq W(\somesetlinking) \cup\{\somewith\}$.

By maximality of $\somesetlinking$, $\somesetcycle\not\subseteq\G_{\somesetlinking^\somewith}$: there exists an edge $\somevertex \to \othervertex$ of $\somesetcycle$ not in $\G_{\somesetlinking^\somewith}$.
Without loss of generality, $\somevertex$ is an $\ax$-vertex.
Indeed, if all edges in the cycle $\somesetcycle$ of target some $\ax$-vertices belong to $\G_{\somesetlinking^\somewith}$, then $\somesetcycle$ would be included in $\G_{\somesetlinking^\somewith}$.
It suffices to prove $\somevertex \to \somewith$ is a jump edge in $\G_\somepn$ to conclude.
We proceed by contradiction, which implies the following property: for all $\somelinking\in\somesetlinking$, $\somevertex\in\G_\somelinking \iff \somevertex\in\G_{\somelinking^\somewith}$.
Indeed, if $\somevertex\in\G_\somelinking$ whereas $\somevertex\notin\G_{\somelinking^\somewith}$ (or vice-versa), then $\somevertex \to \somewith$ is a jump edge since the only $\with$-vertex toggled by $\{\somelinking,\somelinking^\somewith\}$ is $\somewith$.

If $\somevertex \to \othervertex$ is not a jump edge, then it belongs to $\G_\somelinking$ for some $\somelinking\in\somesetlinking$ and thus it is also in $\G_{\somelinking^\somewith}$; a contradiction, for $\somevertex \to \othervertex$ is not an edge of $\G_{\somesetlinking^\somewith}$.
Thence, $\somevertex \to \othervertex$ is a jump edge: there exist $\somelinking, \otherlinking\in\somesetlinking$ such that the axiom link $\somelink$ associated with $\somevertex$ belongs to $\somelinking$ but not to $\otherlinking$, and $\somewith$ is the only $\with$-vertex toggled by $\{\somelinking,\otherlinking\}$.
But then, $\somelink\in\somelinking^\somewith\backslash\otherlinking^\somewith$.
Also, $\somewith$ is the only $\with$-vertex toggled by $\{\somelinking^\somewith,\otherlinking^\somewith\}$, since these linkings are on the $\with$-resolutions $\someres(\somelinking)$ and $\someres(\otherlinking)$ that differ only on $\with$-vertices toggled by $\{\somelinking,\otherlinking\}$, namely only on $\othervertex$.
Therefore, $\somevertex \to \othervertex$ belongs to $\G_{\somesetlinking^\somewith}$, a contradiction.
We finally conclude that $\somevertex \to \somewith$ is a jump edge in $\G_\somepn$.
\end{proof}

Thus, given $\somesetcycle$ a maximal connected union of cusp-free (\ie\ switching) cycles, \cref{lem:4.32} gives a jump edge $\edgeofsetcycles(\somesetcycle)$ between an $\ax$-vertex $\sourceofsetcycles(\somesetcycle)$ in $\somesetcycle$ and a $\with$-vertex $\targetofsetcycles(\somesetcycle)$ outside of it.
The needed hypotheses \cref{HG1,HG2} follow directly, using that cusp-points arise exactly for $\pw$-vertices and the colors associated with their in-edges.

\Cref{th:MALLParamYeo} gives a splitting vertex for any set $\somesetedge$ of vertex-color pairs dominating the cusp-points of a well-colored proof net and disjoint from pairs $(\somevertex,\somecolor)$ with $\somevertex$ an $\ax$-vertex and $\somecolor$ the color of a jump edge incident to it.
As for MLL (\cref{sec:seq:yeo}), natural choices for the parameter $\somesetedge$ yield various proofs of \cref{prop:splitting_mall}, hence various strategies to select splitting vertices along the sequentialization procedure.
In particular, we recover statements analogous to those of MLL, with among them the main intermediate results in the proofs of sequentialization from~\cite{mallpnlong,DiGuardiaLaurent22}: the existence of splitting $\pw$-vertices (\aka\ separating $\pw$) from~\cite[Lemma 4.19]{mallpnlong}, and (almost\footnote{In~\cite{DiGuardiaLaurent22}, a splitting $\oplus$-vertex is asked to be unary, which we prove is always the case in the next section.}) the existence of splitting terminal vertices (\aka\ sequentializing vertices) from~\cite{DiGuardiaLaurent22}.

\begin{cor}[Existence of a splitting vertex]
\label{prop:splitting_mall}
If a proof net contains a vertex, then it contains a splitting vertex.
\end{cor}
\begin{proof}
Take $\somesetedge$ the set of all (non-leaf vertex)-color pairs of the proof net, except for pairs $(\somevertex,\somecolor)$ with $\somevertex$ an $\ax$-vertex and $\somecolor$ the color of a jump edge incident to it.
Note that $\somesetedge$ is not empty because the presence of a leaf implies the presence of an $\ax$-vertex, and this vertex has two of its incident edges which are colored differently than all its jump edges.
By \cref{th:MALLParamYeo}, for each $\orderyeo$-maximal element $(\somevertex, \somecolor)\in\somesetedge$, the vertex $\somevertex$ is splitting.
\end{proof}

\begin{cor}[Splitting $\pw$-vertices]
\label{lem:4.19}
If a proof net contains a $\pw$-vertex, then it contains a splitting one.
\end{cor}
\begin{proof}
Let $\somesetedge$ be the set of all cusp-points: a $\orderyeo$-maximal element of $\somesetedge$ is $(\somevertex, \somecolor)$ with $\somevertex$ a splitting $\pw$-vertex by \cref{th:MALLParamYeo}.
\end{proof}

\begin{cor}[Splitting terminal vertices]
\label{prop:terminal_splitting_mall}
If a proof net contains a vertex, then it contains a splitting terminal one.
\end{cor}
\begin{proof}
Set $\somesetedge$ the set of all pairs $(\somevertex, \somecolor)$ where $\somevertex$ is a $\tensor\backslash\parr\backslash\with\backslash\oplus$-vertex and $\somecolor\neq\coloring(\someedge,\somevertex)$ for all outgoing edges $\someedge$ of $\somevertex$.
Considering $(\somevertex, \somecolor)\in\somesetedge$ a maximal element for $\orderyeo$, not only $\somevertex$ is splitting by \cref{th:MALLParamYeo}, but it is also terminal.
Indeed, otherwise and if $\somevertex$ is a $\tensor\backslash\parr\backslash\with\backslash\oplus$-vertex, it would have a conclusion $\someedge$ with defined target $\othervertex$, and \cref{lem:terminal_not_maximal_order} yields $(\somevertex, \somecolor) \orderyeo (\othervertex, \coloring(\someedge, \othervertex))$; and if $\somevertex$ is an $\ax$-vertex then it has an out-edge $\someedge_1$ to a leaf $\someleaf$ which itself has an out-edge $\someedge_2$ with defined target $\othervertex$, and applying twice \cref{lem:terminal_not_maximal_order} yields $(\somevertex, \somecolor) \orderyeo (\someleaf, \coloring(\someedge_1, \someleaf)) \orderyeo (\othervertex, \coloring(\someedge_2, \othervertex))$.
\end{proof}

\begin{cor}[Splitting non $\ax$-vertices.]
If a proof net contains a $\tensor\backslash\parr\backslash\with\backslash\oplus$-vertex, then it contains a splitting one.
\end{cor}
\begin{proof}
Take $\somesetedge$ the set of all vertex-color pairs $(\somevertex, \somecolor)$ with $\somevertex$ neither a leaf nor an $\ax$-vertex.
By \cref{th:MALLParamYeo}, each $\orderyeo$-maximal element $(\somevertex, \somecolor)\in\somesetedge$ yields a splitting vertex.
\end{proof}

\subsection{Sequentialization using Splitting Vertices}
\label{sec:mall_seq_from_splitting}

We detail here how to deduce the sequentialization theorem (\cref{th:seq_mall}) from the existence of a splitting vertex (\cref{prop:splitting_mall}).
Contrary to MLL, in MALL sequentialization is not an immediate consequence of \cref{prop:terminal_splitting_mall}.
We need the following results, whose proofs involve finding some jump edges:
\begin{itemize}
\item
a splitting vertex belongs to the additive resolutions of all linkings -- so that we respect our constraints on hypotheses;
\item
a splitting $\oplus$-vertex is unary -- so that it corresponds to a $\oplus_1$- or $\oplus_2$-rule.
\end{itemize}

This section is divided as follows: first we show the two wanted results on splitting vertices (\cref{sec:results_splittings}), and then we prove the sequentialization theorem (\cref{sec:mall_seq_proof}).

\subsubsection{Results on Splitting vertices}
\label{sec:results_splittings}

The following lemma allows us to easily find jump edges, and will be of use for our wanted results on splitting vertices.

\begin{lem}
\label{lem:use_tp}
Let $\somelinking$ and $\otherlinking$ be two linkings in a proof structure $\somepn$.
Consider a predicate $P$ on linkings of $\somepn$ such that $P(\somelinking)$ is true whereas $P(\otherlinking)$ is not.
Then, there exist linkings $\somelinking',\otherlinking'\in\somepn$ such that $P(\somelinking')$ is true, $P(\otherlinking')$ is not and $\{\somelinking' , \otherlinking'\}$ toggles exactly one $\with$-vertex, which is among those toggled by $\{\somelinking , \otherlinking\}$.
\end{lem}
\begin{proof}
We reason by induction on $n$ the number of $\with$-vertices toggled by $\{\somelinking,\otherlinking\}$.

Set $\someres$ and $\otherres$ two $\with$-resolutions respectively of $\somelinking$ and $\otherlinking$, chosen such that the $n$ $\with$-vertices toggled by $\{\somelinking,\otherlinking\}$ are exactly the $\with$-vertices in $\someres$ and $\otherres$ with different choices of premises.
This can be done as follows, where we explain how to build $\someres$ and a symmetric construction would yield $\otherres$.
For a given $\with$-vertex, if it is in the additive resolution of $\somelinking$ then keep for $\someres$ the same premise as in this additive resolution; otherwise, if it is in the additive resolution of $\otherlinking$, then keep the corresponding premise; otherwise, keep the left premise.

Remark that $n\neq 0$, for otherwise $\someres=\otherres$ and by \cref{item:P1} $\somelinking=\otherlinking$, a contradiction.
So $n\geq 1$: let $\somewith$ be one of those $\with$-vertices.
Consider $\otherotherres$ the $\with$-resolution obtained from $\someres$ by taking the other premise for $\somewith$ and giving to $\with$-vertices introduced this way the same premise they have in $\otherres$.
Call $\otherotherlinking$ the linking of $\somepn$ on $\otherotherres$, using \cref{item:P1}.
If $P(\otherotherlinking)$ is false, then we are done: the wanted linkings are $\somelinking$ and $\otherotherlinking$.
Otherwise, by construction of $\otherotherres$, there are $n-1$ $\with$-vertices with different choices of premises in $\otherotherres$ and $\otherres$: they are those with different choices of premises in $\someres$ and $\otherres$ except for $\somewith$.
As the $\with$-vertices toggled by $\{\otherotherlinking,\otherlinking\}$ are included in those, we can apply the induction hypothesis on $\otherotherlinking$ and $\otherlinking$ to conclude.
\end{proof}

Given vertices $\somevertex$ and $\othervertex$ in the syntactic forest $\Delta\vdash\Gamma$, when $\somevertex$ is in the sub-tree of root $\othervertex$ we say $\somevertex$ is an \definitive{ancestor} of $\othervertex$, and $\othervertex$ is a \definitive{descendant} of $\somevertex$.

\begin{fact}
\label{rem:no_hyp_oplus}
Consider an edge $\someedge$ in (the syntactic forest of) $\Delta \vdash \Gamma$, and a set of linkings $\somepn$ on $\Delta \vdash \Gamma$.
Assume $\someedge$ belongs to the additive resolution of $\otherlinking$ but not to the one of $\somelinking$, for two linkings $\somelinking,\otherlinking\in\somepn$.
Then, $\someedge$ has a defined source $\source(\someedge)$, which has for ancestor a leaf $\someleaf$ that is used in an axiom link of $\otherlinking$.
\end{fact}
\begin{proof}
There cannot be any edge with undefined source in the sub-tree of $\Delta \vdash \Gamma$ with $\someedge$ as a partial edge with no target, for otherwise the additive resolution of $\somelinking'$ does not contain this edge of undefined source, impossible by definition.
Hence, following the ancestor relation from the source of $\someedge$, that is defined, in the additive resolution of $\otherlinking$, one necessarily reaches a leaf $\someleaf$, associated with some axiom link $\somelink$ of $\otherlinking$.
\end{proof}

\begin{lem}\label{lem:splitting_in_all_additive_helper}
Let $\somepn$ be a proof structure on $\Delta \vdash \Gamma$.
Consider $\someedge$ a non-jump edge in $\G_\somepn$ such that $\someedge\notin\G_\somelinking$ for some linking $\somelinking\in\somepn$.
Then, there exists a cycle in $\G_\somepn$ that contains $\someedge$.
\end{lem}
\begin{proof}
Set $P(\otherotherlinking)$ the following predicate on linkings: ``$\someedge$ belongs to $\G_\otherotherlinking$''.
Remark that $P(\otherlinking)$ holds for some linking $\otherlinking\in\somepn$ for $\someedge\in\G_\somepn$ is not a jump edge, whereas $P(\somelinking)$ does not hold.
By \cref{lem:use_tp}, one gets linkings $\somelinking',\otherlinking'\in\somepn$ such that $\someedge$ belongs to $\G_{\otherlinking'}$ but not to $\G_{\somelinking'}$, and $\{\somelinking',\otherlinking'\}$ toggles a unique $\with$-vertex $\somewith$.

One easy case is when $\someedge$ is an outgoing edge from an $\ax$-vertex $\somevertex$ towards a leaf $\someleaf$ that is also in $\G_{\somelinking'}$.
Then, there is another $\ax$-vertex $\othervertex$ with an outgoing edge to $\someleaf$ in $\G_{\somelinking'}$, and we get a cycle containing $\someedge$ in $\G_\somepn$: $\somevertex\to\somewith\gets\othervertex\to\someleaf\gets\somevertex$.

Otherwise, either $\someedge$ is an outgoing edge from an $\ax$-vertex $\othervertex_\someleaf$ to a leaf $\someleaf\in\G_{\otherlinking'}\backslash\G_{\somelinking'}$; or it is an edge in the syntactic forest of $\Delta\vdash\Gamma$ and, by \cref{rem:no_hyp_oplus}, there is a leaf $\someleaf$ that is an ancestor of the source of $\someedge$ and that is used in an axiom link of $\otherlinking'$ corresponding to an $\ax$-vertex $\othervertex_\someleaf$ in $\G_{\otherlinking'}$ -- note that $\someleaf,\othervertex_\someleaf\notin\G_{\somelinking'}$ as $\someedge\notin\G_{\somelinking'}$.
In both cases, call $\otherothervertex$ the descendant of $\someleaf$ where the additive resolutions of $\somelinking'$ and $\otherlinking'$ make different choices of premises.
We will exhibit a cycle containing as a sub-path $\concatpath{\othervertex_\someleaf\to\someleaf}{\somepath_\someleaf}$ with $\somepath_\someleaf$ the path in the syntactic forest of $\Delta\vdash\Gamma$ from $\someleaf$ to $\otherothervertex$, hence a cycle containing $\someedge$ in both cases.
Remark there is a jump edge $\othervertex_\someleaf\to\somewith$ in $\G_\somepn$.

The vertex $\otherothervertex$ is either a $\with$-vertex, in which case it can only be $\somewith$ the unique $\with$-vertex toggled by $\{\somelinking',\otherlinking'\}$, or it is a $\oplus$-vertex.
\begin{itemize}
\item
If $\otherothervertex$ is the $\with$-vertex $\somewith$, then $\concatpath{\concatpath{\othervertex_\someleaf\to\someleaf}{\somepath_\someleaf}}{\otherothervertex\gets\othervertex_\someleaf}$ is a cycle of $\G_\somepn$ containing $\someedge$ (\cref{lem:concatsimplepathsprefix}) -- see \cref{fig:proof:splitting_in_all_additive_with}.
\item
If $\otherothervertex$ is a $\oplus$-vertex, then by applying \cref{rem:no_hyp_oplus} on the premise of $\otherothervertex$ kept in $\G_{\somelinking'}$, there is a leaf $\otherleaf$ and an $\ax$-vertex $\othervertex_\otherleaf$ that corresponds to an axiom link in $\somelinking'$ containing $\otherleaf$, and $\otherleaf$ is an ancestor of $\otherothervertex$.
See \cref{fig:proof:splitting_in_all_additive_oplus} for an illustration of this case.
Thence, there are in $\G_\somepn$ jump edges $\othervertex_\someleaf\to\somewith$ and $\othervertex_\otherleaf\to\somewith$.
Call $\somepath_\otherleaf$ the path in the syntactic forest of $\Delta\vdash\Gamma$ from $\otherleaf$ to $\otherothervertex$.
We can again exhibit a cycle containing $\someedge$: $\concatpath{\concatpath{\somewith\gets\othervertex_\someleaf\to\someleaf}{\somepath_\someleaf}}{\concatpath{\somepath_\otherleaf}{\otherleaf\gets\othervertex_\otherleaf\to\somewith}}$.
This is indeed a cycle thanks to \cref{lem:concatsimplepathsprefix}, using that $\somewith$ cannot be an ancestor of $\otherothervertex$ for it belongs to both the additive resolution of a linking using the left premise of $\otherothervertex$ and one using its right premise.
\end{itemize}
\end{proof}

\begin{figure}
\begin{minipage}[b]{.48\textwidth}
\begin{adjustbox}{}
\begin{tikzpicture}
\begin{genscope}
	\node (l) at (-2,2) {$\someleaf$};
	\node (axl) at (-3,3) {$\othervertex_\someleaf$};
 	\node (x) at (0,0) {$\otherothervertex$};

    \path (axl) edge (l);
    \path (axl) edge ++(-.75,-.75);
    \path (l) edge node[chemin nomme,below]{$\somepath_\someleaf$} (x);
    \path (.6,.6) edge (x);
    \path[out=-90,in=170] (axl) edge (x);
\end{genscope}
\end{tikzpicture}
\end{adjustbox}
\caption{Cycle in the proof of \cref{lem:splitting_in_all_additive_helper} when $\otherothervertex$ is a $\with$-vertex}
\label{fig:proof:splitting_in_all_additive_with}
\end{minipage}
\hfill
\begin{minipage}[b]{.48\textwidth}
\begin{adjustbox}{}
\begin{tikzpicture}
\begin{genscope}
	\node (l) at (-2,2) {$\someleaf$};
	\node (axl) at (-3,3) {$\othervertex_\someleaf$};
 	\node (r) at (2,2) {$\otherleaf$};
	\node (axr) at (3,3) {$\othervertex_\otherleaf$};
 	\node (x) at (0,0) {$\otherothervertex$};
 	\node (w) at (2,.5) {$\somewith$};

    \path (axl) edge (l);
    \path (axl) edge ++(-.75,-.75);
    \path (axr) edge (r);
    \path (axr) edge ++(.75,-.75);
    \path (l) edge node[chemin nomme,below]{$\somepath_\someleaf$} (x);
    \path (r) edge node[chemin nomme,below]{$\somepath_\otherleaf$} (x);
    \path[out=-90,in=135] (axl) edge (w);
    \path[out=-90,in=45] (axr) edge (w);
\end{genscope}
\end{tikzpicture}
\end{adjustbox}
\caption{Cycle in the proof of \cref{lem:splitting_in_all_additive_helper} when $\otherothervertex$ is a $\oplus$-vertex}
\label{fig:proof:splitting_in_all_additive_oplus}
\end{minipage}
\end{figure}

\begin{lem}
\label{lem:splitting_in_all_additive}
Let $\somevertex$ be a splitting vertex in $\G_\somepn$, for $\somepn$ a proof structure.
For every $\somelinking\in\somepn$, $\somevertex$ belongs to $\G_\somelinking$; \ie\ $\somevertex$ belongs to its additive resolution $\Delta \vdash \Gamma\upharpoonright\somelinking$ or is an $\ax$-vertex associated with an axiom link of $\somelinking$.
\end{lem}
\begin{proof}
We proceed by contraposition: given a vertex $\somevertex$ in $\G_\somepn$ and a linking $\somelinking\in\somepn$, if $\somevertex$ does not belong to $\G_\somelinking$ then it is not splitting.
It suffices to this end to exhibit a cycle containg an outgoing edge of $\somevertex$.
We obtain this cycle by applying \cref{lem:splitting_in_all_additive_helper} to any outgoing non-jump edge of $\somevertex$, that belongs to $\G_\somepn$ but not to $\G_\somelinking$ as $\somevertex$ belongs to the first but not to the second.
\end{proof}

\begin{lem}
\label{lem:plus_jump}
A splitting $\oplus$-vertex in $\G_\somepn$, for $\somepn$ a proof structure, is unary.
\end{lem}
\begin{proof}
We prove a binary $\oplus$-vertex $\somevertex$ is not splitting.
As $\somevertex$ is binary, there exist linkings $\somelinking,\otherlinking\in\somepn$ using respectively the left and right premises of $\somevertex$.
Applying \cref{lem:splitting_in_all_additive_helper} on the left premise of $\somevertex$, we get a cycle containing this edge, hence containing $\somevertex$: $\somevertex$ is not splitting.
\end{proof}

\begin{lem}
\label{lem:exists_jump}
Let $\somepn$ be a proof structure, $\somelinking_0,\somelinking_1\in\somepn$ and $\somelink\in\somelinking_0\backslash\somelinking_1$ an axiom link.
There exists a $\with$-vertex $\somewith$ toggled by $\{\somelinking_0,\somelinking_1\}$ such that $\somelink$ depends on $\somewith$ in $\somepn$.
\end{lem}
\begin{proof}
Define a predicate $P$ on linkings by $P(\otherlinking) \eqdef (\somelink\in\otherlinking)$.
Applying \cref{lem:use_tp}, one gets linkings $\otherlinking_0,\otherlinking_1\in\somepn$ such that $\somelink\in\otherlinking_0\backslash\otherlinking_1$ and $\{\otherlinking_0 , \otherlinking_1 \}$ toggles a unique $\with$-vertex $\somewith$, that is also toggled by $\{\somelinking_0,\somelinking_1\}$.
Thus, $\somelink$ depends on $\somewith$ in $\somepn$.
\end{proof}

\subsubsection{Proof of Sequentialization}
\label{sec:mall_seq_proof}

Thanks to the previous section, we now prove that splitting vertices are easy to sequentialize.
After proving so for each kind of vertex, we deduce the sequentialization theorem.
To not repeat the same argumentation many times, we consider first the substitution along an edge, then along a vertex, and then only consider terminal splitting vertices.
We also use the following intermediate result for proving the correctness criterion.

\begin{lem}
\label{lem:P2&3_seq}
Let $\somepn$ be a proof net on $\Delta \vdash \Gamma$ and $\somepn'$ be a set of linkings on $\Delta' \vdash \Gamma'$, where the syntactic forest of $\Delta' \vdash \Gamma'$ is (isomorphic to) a sub-graph of the syntactic forest of $\Delta \vdash \Gamma$.
In particular, leaves of $\Delta' \vdash \Gamma'$ can be seen as leaves of $\Delta \vdash \Gamma$.
Assume that for all $\somesetlinking' \subseteq \somepn'$, there exists $\otherlinking$ such that $\somesetlinking \eqdef \{\somelinking' \cup \otherlinking \mid \somelinking' \in \somesetlinking'\} \subseteq \somepn$ and $\G_{\somesetlinking'}$ is (isomorphic to) a sub-graph of $\G_\somesetlinking$ with, in $\G_\somesetlinking$, no jump edge from the sub-graph $\G_{\somesetlinking'}$ to its complementary sub-graph $\G_\somesetlinking\backslash\G_{\somesetlinking'}$.

Then $\somepn'$ respects the MLL condition \cref{item:P2} and the toggling condition \cref{item:P3}.
\end{lem}
\begin{proof}
The MLL condition \cref{item:P2} for $\somepn'$ holds thanks to \cref{item:P2} for $\somepn$: given $\somelinking'\in\somepn'$, by hypothesis there is some $\somelinking=\somelinking'\cup\otherlinking\in\somepn$ with $\G_{\somelinking'}\subseteq\G_\somelinking$, the latter having no switching cycle.

Considering the toggling condition \cref{item:P3}, let $\somesetlinking'\subseteq\somepn'$ be of size at least two.
Then the hypothesis yields some $\somesetlinking = \{\somelinking' \cup \otherlinking \mid \somelinking' \in \somesetlinking'\} \subseteq \somepn$ of the same size: by \cref{item:P3} for $\somepn$, $\somesetlinking$ toggles a $\with$-vertex $\somewith$ not in any switching cycle of $\G_\somesetlinking$.
To conclude, it suffices to show $\somewith$ belongs to $\G_{\somesetlinking'}$.
Indeed, $\somewith$ is not in any switching cycle of $\G_{\somesetlinking'}$ and would then belong to $\Delta' \vdash \Gamma'$; plus, $\somewith$ is toggled by $\somesetlinking'$ as it is toggled by $\somesetlinking = \{\somelinking' \cup \otherlinking \mid \somelinking' \in \somesetlinking'\}$.

As $\somewith$ is toggled by $\somesetlinking$, there is a jump edge to it from an $\ax$-vertex $\somevertex$ corresponding to an axiom link $\somelink$ in a linking $\somelinking_1'\cup\otherlinking\in\somesetlinking$ and not in another linking $\somelinking_2'\cup\otherlinking\in\somesetlinking$.
Observe that $\somelink$ must belong to $\somelinking_1'$ and not to $\somelinking_2'$.
Hence, the corresponding $\somevertex$ is in $\G_{\somesetlinking'}$.
Thus, $\somewith$ belongs to $\G_{\somesetlinking'}$ too, for there is no jump edge from $\G_{\somesetlinking'}$ to the rest of the graph by assumption.
\end{proof}

\begin{lem}[Splitting Edge]
\label{lem:split_edge}
Consider a proof net $\somepn$ on the sequent with hypotheses $\Delta \vdash \Gamma$.
Let $\someedge$ be an edge in the syntactic forest of $\Delta \vdash \Gamma$, with associated formula $A$.
Assume $\someedge$ belongs to all $\G_\somelinking$ for $\somelinking\in\somepn$, and that $\G_\somepn$ is the disjoint union of two partial graphs $\somegraph_0$ and $\somegraph_1$ where an edge of $\somegraph_0$ of target $\target(\someedge)$ and undefined source, and an edge of $\somegraph_1$ of source $\source(\someedge)$ and undefined target, are identified together as $\someedge$.
Then, there exist proof nets $\somepn_0$ and $\somepn_1$ respectively on $\Delta_0,A\vdash\Gamma_0$ and $\Delta_1\vdash A,\Gamma_1$ such that $\somepn = \{\somelinking_0\cup\somelinking_1 \mid \somelinking_0\in\somepn_0,\somelinking_1\in\somepn_1\}$, $\Delta=\Delta_0,\Delta_1$, $\Gamma=\Gamma_0,\Gamma_1$, $\G_{\somepn_0}\giso\somegraph_0$ and $\G_{\somepn_1}\giso\somegraph_1$.
In particular, $\G_\somepn\giso\G_{\somepn_0}\cup\G_{\somepn_1}$ where the two edges in $\G_{\somepn_0}\cup\G_{\somepn_1}$ corresponding to $A$ are identified together.
\end{lem}
\begin{proof}
For a linking $\somelinking\in\somepn$, define $\somelinking_0$ (\resp\ $\somelinking_1$) as the axiom links of $\somelinking$ whose leaves are in $\somegraph_0$ (\resp\ $\somegraph_1$).
No axiom link uses leaves in both $\somegraph_0$ and $\somegraph_1$ since $\someedge$ is not an out-edge of an $\ax$-vertex: $\somelinking=\somelinking_0\cup\somelinking_1$.
Define $\somepn_0\eqdef\{\somelinking_0\mid\somelinking\in\somepn\}$, a set of linkings on $\Delta_0,A\vdash\Gamma_0$ with $\Gamma_0$ and $\Delta_0$ given by $\somegraph_0$: it contains the (pruned) syntactic trees given by restriction of $\Delta \vdash \Gamma$ to $\somegraph_0$.
That it is indeed a set of linkings, with our contraint on hypotheses, follows from the fact that $\someedge$ belongs to all $\G_\somelinking$ for $\somelinking\in\somepn$.
Moreover, remark that $\G_{\somepn_0}\giso\somegraph_0$.
Similarly, set $\somepn_1\eqdef\{\somelinking_1\mid\somelinking\in\somepn\}$ a set on linkings on the sequent $\Delta_1\vdash A,\Gamma_1$ given by restriction to $\G_1$, with $\G_{\somepn_1}\giso\somegraph_1$.
Thus, $\Delta=\Delta_0,\Delta_1$ and $\Gamma=\Gamma_0,\Gamma_1$ as wanted.
Observe $\somepn \subseteq \{\somelinking_0\cup\somelinking_1 \mid \somelinking_0\in\somepn_0,\somelinking_1\in\somepn_1\}$ by definition of $\somepn_0$ and $\somepn_1$.
Also notice $\somepn_0\neq\emptyset$: otherwise $\somepn = \emptyset$, a contradiction with \cref{item:P1} as $\Delta \vdash \Gamma$ has at least one $\with$-resolution.
Similarly, $\somepn_1\neq\emptyset$.
We prove that $\somepn_0$ and $\somepn_1$ are proof nets, with \cref{item:P1,,item:P2,,item:P3} to show in both cases, as well as the inclusion $\{\somelinking_0\cup\somelinking_1 \mid \somelinking_0\in\somepn_0,\somelinking_1\in\somepn_1\} \subseteq \somepn$ using \cref{item:P1}.

Let us prove the resolution condition \cref{item:P1} for $\somepn_0$.
A similar argument will prove it also for $\somepn_1$.
Let $\someres_0$ be a $\with$-resolution of $\Delta_0,A_0\vdash\Gamma_0$.
We can extend it into an arbitrary $\with$-resolution $\someres$ of $\Delta\vdash \Gamma$.
There exists a unique $\somelinking\in\somepn$ on $\someres$ by \cref{item:P1}, yielding a $\somelinking_0$ on $\someres_0$.
For the uniqueness condition, assume we have $\somelinking_0,\otherlinking_0\in\somepn_0$ both on $\someres_0$, with $\somelinking_0\neq\otherlinking_0$.
Without any loss of generality, there exists $\somelink\in\somelinking_0\backslash\otherlinking_0\subseteq\somelinking\backslash\otherlinking$.
By \cref{lem:exists_jump}, there exists in $\G_\somepn$ a $\with$-vertex $\somewith$ on which $\somelink$ depends and toggled by $\{\somelinking,\otherlinking\}$.
It follows that $\somewith$ is in $\somegraph_0$ as there is no jump edge between $\somegraph_0$ and $\somegraph_1$ in $\G_\somepn$.
But then $\{\somelinking_0,\otherlinking_0\}$ toggles $\somewith$, a contradiction as they are on the same $\with$-resolution $\someres_0$.

Let us now prove $\{\somelinking_0\cup\somelinking_1 \mid \somelinking_0\in\somepn_0,\somelinking_1\in\somepn_1\} \subseteq \somepn$.
Consider $\somelinking_0\in\somepn_0$ and $\otherlinking_1\in\somepn_1$, as well as $\someres_0$ and $\someres_1$ $\with$-resolutions they are respectively on.
Set $\someres$ the $\with$-resolution of $\somepn$ taking the same choices of premises as $\someres_0$ and $\someres_1$.
By \cref{item:P1} of $\somepn$, there is a unique $\otherotherlinking\in\somepn$ on it.
But then for $i\in\{0,1\}$, $\otherotherlinking_i$ is on $\someres_i$.
By \cref{item:P1} of $\somepn_0$, $\otherotherlinking_0=\somelinking_0$, and by \cref{item:P1} of $\somepn_1$, $\otherotherlinking_1=\otherlinking_1$.
Hence $\otherotherlinking = \somelinking_0\cup\otherlinking_1$, and finally $\somepn = \{\somelinking_0\cup\somelinking_1 \mid \somelinking_0\in\somepn_0,\somelinking_1\in\somepn_1\}$.

We will obtain \cref{item:P2,item:P3} by applying \cref{lem:P2&3_seq}.
Remark that for any $\somesetlinking_0\subseteq\somepn_0$, $\somesetlinking_1\subseteq\somepn_1$, $\G_{\{\somelinking_0\cup\somelinking_1 \mid \somelinking_0\in\somesetlinking_0,\somelinking_1\in\somesetlinking_1\}}\giso\G_{\somesetlinking_0}\cup\G_{\somesetlinking_1}$ where the edge of the hypothesis $A$ in $\G_{\somesetlinking_0}$ is identified with the edge of the conclusion $A$ in $\G_{\somesetlinking_1}$.
This uses that $\someedge$ belongs to $\G_{\{\somelinking_0\cup\somelinking_1 \mid \somelinking_0\in\somesetlinking_0,\somelinking_1\in\somesetlinking_1\}}$ by hypothesis.
Let us now apply \cref{lem:P2&3_seq} on each $\somepn_i$. 
Consider $\somesetlinking_i\subseteq\somepn_i$, and take an arbitrary $\somelinking_{1-i}\in\somepn_{1-i}$.
Then, $\somesetlinking\eqdef\{\somelinking_i\cup\somelinking_{1-i}\mid\somelinking_i\in\somesetlinking_i\}$ is a set of linkings in $\somepn$.
Moreover, as $\someedge$ is in no cycle of $\G_\somepn\supseteq\G_\somesetlinking$, there is no jump edge from $\G_{\somesetlinking_i}$ to its complementary in $\G_{\somesetlinking}$.
Hence, we can apply \cref{lem:P2&3_seq} to deduce \cref{item:P2,item:P3}.
\end{proof}

\begin{lem}[Splitting Substitution]
\label{lem:split_seq}
Let $\somevertex$ be a splitting non-$\ax$-vertex in a proof net $\somepn$ on the sequent with hypotheses $\Delta \vdash \Gamma$, with $A$ the formula associated with the conclusion of $\somevertex$.
Then, there exist proof nets $\somepn_0$ and $\somepn_1$ respectively on $\Delta_0,A\vdash\Gamma_0$ and $\Delta_1\vdash A,\Gamma_1$ such that $\somepn = \{\somelinking_0\cup\somelinking_1 \mid \somelinking_0\in\somepn_0,\somelinking_1\in\somepn_1\}$, $\Delta=\Delta_0,\Delta_1$, $\Gamma=\Gamma_0,\Gamma_1$, and $\somevertex$ is splitting and terminal in $\somepn_1$.
\end{lem}
\begin{proof}
As $\somevertex$ is splitting, its conclusion edge $\someedge$ is in no cycle of $\G_\somepn$.
Hence, $\G_\somepn$ is the disjoint union of two partial graphs $\somegraph_0$ and $\somegraph_1$ where we identify an edge of $\somegraph_0$ and one of $\somegraph_1$ as $\someedge$: $\somegraph_0$ is the connected component of $\target(\someedge)$ (if defined) in the removal of $\someedge$ in $\G_\somepn$, plus an edge of target $\target(\someedge)$ and undefined source; and $\somegraph_1$ is the rest of the graph, plus an edge of source $\source(\someedge)$ and undefined target -- in particular, $\somegraph_1$ contains $\somevertex=\source(\someedge)$.
Furthermore, $\someedge$ belongs to all $\G_\somelinking$ for $\somelinking\in\somepn$, because $\somevertex$ belongs to them thanks to \cref{lem:splitting_in_all_additive}.
We get the wanted $\somepn_0$ and $\somepn_1$ by applying \cref{lem:split_edge}.
That $\somevertex$ is splitting and terminal in $\somepn_1$ follows directly from $\G_\somepn\giso\G_{\somepn_0}\cup\G_{\somepn_1}$, and that $\somevertex$ is splitting in $\somepn$.
\end{proof}

\begin{lem}[Splitting Terminal $\tensor$]
\label{lem:tensor_seq}
Let $\somevertex$ be a splitting terminal $\tensor$-vertex in a proof net $\somepn$ on $\Delta\vdash A_0\tensor A_1,\Gamma$, with $A_0\tensor A_1$ the formula associated with the conclusion of $\somevertex$.
There exists $\somepn_0$, $\somepn_1$, $\Gamma_0$, $\Gamma_1$, $\Delta_0$, $\Delta_1$ such that $\somepn = \{\somelinking\cup\otherlinking \mid \somelinking\in\somepn_0,\otherlinking\in\somepn_1\}$, $\Delta = \Delta_0,\Delta_1$, $\Gamma = \Gamma_0,\Gamma_1$, and $\somepn_i$ is a proof net on $\Delta_i\vdash A_i,\Gamma_i$ for $i\in\{0,1\}$.
\end{lem}
\begin{proof}
Call $\someedge_0$ the premise of $\somevertex$ corresponding to $A_0$, and $\someedge_1$ its other premise, corresponding to $A_1$.
Note that $\someedge_0$ and $\someedge_1$ belong to all $\G_\somelinking$ for $\somelinking\in\somepn$, because $\somevertex$ belongs to them since it is terminal.
As $\somevertex$ is splitting, both $\someedge_0$ and $\someedge_1$ are in no cycle of $\G_\somepn$.
Hence, $\G_\somepn$ is the disjoint union of three partial graphs $\somegraph_0$, $\somegraph_1$ and $\somegraph_2$ where we identify an edge of $\somegraph_0$ and one of $\somegraph_2$ as $\someedge_0$, and also identify an edge of $\somegraph_1$ and one of $\somegraph_2$ as $\someedge_1$: $\somegraph_2$ is $\somevertex$ and its three incident edges; $\somegraph_0$ is the connected component of $\source(\someedge_0)$ in the removal of $\someedge_0$ in $\G_\somepn$, plus an edge of source $\source(\someedge_0)$ and undefined target; and $\somegraph_1$ is the rest of the graph, plus an edge of source $\source(\someedge_1)$ and undefined target.
We apply \cref{lem:split_edge} twice, first on $\someedge_0$ and then on $\someedge_1$, obtaining proof nets $\somepn_0$, $\somepn_1$ and $\somepn_2$ with $\G_{\somepn_0}\giso\somegraph_0$, $\G_{\somepn_1}\giso\somegraph_1$ and $\G_{\somepn_2}\giso\somegraph_2$.
Observe $\somepn_2$ is a proof net on $A_0,A_1\vdash A_0\tensor A_1$, which has no leaf: hence, $\somepn_2=\{\emptyset\}$.
Therefore, $\somepn = \{\somelinking\cup(\otherlinking\cup\otherotherlinking) \mid \somelinking\in\somepn_0,\otherlinking\in\somepn_1,\otherotherlinking\in\somepn_2\} = \{\somelinking\cup\otherlinking \mid \somelinking\in\somepn_0,\otherlinking\in\somepn_1\}$.
\end{proof}

\begin{lem}[Splitting Terminal $\parr$]
\label{lem:parr_seq}
Let $\somevertex$ be a terminal (hence splitting) $\parr$-vertex in a proof net $\somepn$ on $\Delta\vdash A_0\parr A_1,\Gamma$, with $A_0\parr A_1$ the formula associated with the conclusion of $\somevertex$.
Then $\somepn$ is a proof net on $\Delta\vdash A_0,A_1,\Gamma$.
\end{lem}
\begin{proof}
The syntactic forests of $\Delta\vdash A_0\parr A_1,\Gamma$ and $\Delta\vdash A_0, A_1,\Gamma$ differ only by the presence/absence of $\somevertex$, so have the same set of leaves, the same $\with$- and additive resolutions.
Hence, $\somepn$ is also a set of linkings on $\Delta\vdash A_0,A_1,\Gamma$ and respects \cref{item:P1} on it, for it respects it on $\Delta\vdash A_0\parr A_1,\Gamma$.
For any $\somesetlinking\subseteq\somepn$, $\G_\somesetlinking$ on $\Delta\vdash A_0,A_1,\Gamma$ is $\G_\somesetlinking$ on $\Delta\vdash A_0\parr A_1,\Gamma$ where the vertex $\somevertex$ and its conclusion are removed.
Therefore, \cref{item:P2,item:P3} for $\somepn$ on $\Delta\vdash A_0,A_1,\Gamma$ follow from \cref{lem:P2&3_seq}.
\end{proof}

\begin{lem}[Splitting Terminal $\with$]
\label{lem:with_seq}
Let $\somevertex$ be a terminal (hence splitting) $\with$-vertex in a proof net $\somepn$ on $\Delta\vdash A_0\with A_1,\Gamma$, with $A_0\with A_1$ the formula associated with the conclusion of $\somevertex$.
There exist proof nets $\somepn_0$ and $\somepn_1$ respectively on $\Delta\vdash A_0,\Gamma$ and $\Delta\vdash A_1,\Gamma$ such that $\somepn = \somepn_0\cup\somepn_1$.
\end{lem}
\begin{proof}
Define $\somepn_0$ (\resp\ $\somepn_1$) as the linkings of $\somepn$ using the left (\resp\ right) premise of $\somevertex$.
We have $\somepn=\somepn_0\cup\somepn_1$, for $\somevertex$ is terminal so any linking uses exactly one premise of $\somevertex$.

That $\somepn_i$ is a set of linkings on $\Delta\vdash A_i,\Gamma$ follows from the fact that it is composed of linkings on $\Delta\vdash A_0\with A_1,\Gamma$ keeping the premise associated with $A_i$ in its additive resolution.

To obtain that $\somepn_i$ is a proof net, we have \cref{item:P1,,item:P2,,item:P3} to prove.
The resolution condition \cref{item:P1} follows directly from the one on $\somepn$, using that $\somepn_i\subseteq\somepn$.
We obtain \cref{item:P2,item:P3} thanks to \cref{lem:P2&3_seq}.
Indeed, for all $\somesetlinking\subseteq\somepn_i$, $\G_\somesetlinking$ in $\somepn_i$ is $\G_\somesetlinking$ in $\somepn$ where the unary vertex $\somevertex$ and its conclusion are removed, with no jump edge to $\somevertex$ as it is not toggled.
\end{proof}

\begin{lem}[Splitting Terminal $\oplus$]
\label{lem:oplus_seq}
Let $\somevertex$ be a splitting terminal $\oplus$-vertex in a proof net $\somepn$ on $\Delta\vdash A_0\oplus A_1,\Gamma$, with $A_0\oplus A_1$ the formula associated with the conclusion of $\somevertex$.
There exists $i\in\{0,1\}$ such that $\somepn$ is a proof net on $\Delta\vdash A_i,\Gamma$.
\end{lem}
\begin{proof}
By \cref{lem:plus_jump}, $\somevertex$ is unary: set $i\in\{0,1\}$ such that all linkings of $\somepn$ keep the premise of $\somevertex$ corresponding to $A_i$.
As $\somevertex$ is unary, a linking $\somelinking\in\somepn$ is also a linking on $\Delta\vdash A_i,\Gamma$.
Let us prove that $\somepn$ is a proof net on $\Delta\vdash A_i,\Gamma$.

Considering \cref{item:P1}, let $\someres'$ be a $\with$-resolution of $\Delta\vdash A_i,\Gamma$.
We extend it into a $\with$-resolution $\someres$ of $\somevertex\vdash A_0\oplus A_1,\Gamma$ by giving arbitrary choice of premises to the $\with$-vertices of $A_{1-i}$.
As all linkings of $\somepn$ use the premise $A_i$ of $\somevertex$, linkings on $\someres'$ of $\somevertex\vdash A_i,\Gamma$ are exactly those on $\someres$ of $\somevertex\vdash A_0\oplus A_1,\Gamma$.
Hence \cref{item:P1} for $\somepn$ on $\somevertex\vdash A_i,\Gamma$ follows from \cref{item:P1} for $\somepn$ on $\somevertex\vdash A_0\oplus A_1,\Gamma$.

Considering \cref{item:P2,item:P3}, observe that for all $\somesetlinking\subseteq\somepn$, $\G_\somesetlinking$ on $\Delta\vdash A_i,\Gamma$ is $\G_\somesetlinking$ on $\Delta\vdash A_0\oplus A_1,\Gamma$ where we remove the unary $\somevertex$ and its conclusion.
We conclude by \cref{lem:P2&3_seq}.
\end{proof}

\begin{lem}[Splitting $\ax$]
\label{lem:ax_seq}
Let $\somevertex$ be a splitting $\ax$-vertex in a proof net $\somepn$ on the sequent with hypotheses $\Delta \vdash \Gamma$, with $X\orth$ and $X$ the two formulas associated with the axiom link of $\somevertex$ in $\somepn$.
Then there exist proof nets $\somepn_0$ and $\somepn_1$ respectively on $\Delta_0,X\orth\vdash\Gamma_0$ and $\Delta_1,X\vdash \Gamma_1$ such that $\somepn = \{\{\{X\orth,X\}\}\cup\somelinking_0\cup\somelinking_1 \mid \somelinking_0\in\somepn_0,\somelinking_1\in\somepn_1\}$, $\Delta=\Delta_0,\Delta_1$ and $\Gamma=\Gamma_0,\Gamma_1$.
\end{lem}
\begin{proof}
As $\somevertex$ is splitting, it belongs to all $\G_\somelinking$ for $\somelinking\in\somepn$ by \cref{lem:splitting_in_all_additive}.
Thence, the two leafs $X\orth$ and $X$ it is on are in all $\G_\somelinking$, with as sole in-edges the out-edges of $\somevertex$.
In particular, their respective conclusion edges $\someedge_0$ and $\someedge_1$ are in no cycle of $\G_\somepn$ and belong to all $\G_\somelinking$ for $\somelinking\in\somepn$.
Thence, $\G_\somepn$ is the disjoint union of three partial graphs $\somegraph_0$, $\somegraph_1$ and $\somegraph_2$ where we identify an edge of $\somegraph_0$ and one of $\somegraph_2$ as $\someedge_0$, and also identify an edge of $\somegraph_1$ and one of $\somegraph_2$ as $\someedge_1$: $\somegraph_2$ is composed of $\somevertex$, $X\orth$ and $X$ as well as their four incident edges; $\somegraph_0$ is the connected component of $\target(\someedge_0)$ in the removal of $\someedge_0$ in $\G_\somepn$, plus an edge of target $\target(\someedge_0)$ and undefined source; and $\somegraph_1$ is the rest of the graph, plus an edge of target $\target(\someedge_1)$ and undefined source.
We apply \cref{lem:split_edge} twice, first on $\someedge_0$ and then on $\someedge_1$, obtaining proof nets $\somepn_0$, $\somepn_1$ and $\somepn_2$ with $\G_{\somepn_0}\giso\somegraph_0$, $\G_{\somepn_1}\giso\somegraph_1$ and $\G_{\somepn_2}\giso\somegraph_2$.
Observe $\somepn_2$ is a proof net on $\vdash X\orth, X$, hence $\somepn_2=\{\{\{X\orth,X\}\}\}$.
Therefore, $\somepn = \{\somelinking\cup(\otherlinking\cup\{\{X\orth,X\}\}) \mid \somelinking\in\somepn_0,\otherlinking\in\somepn_1\}$.
\end{proof}

We now prove the sequentialization theorem thanks to \cref{lem:split_seq,lem:tensor_seq,lem:parr_seq,lem:with_seq,lem:oplus_seq,lem:ax_seq}.
As for MLL in \cref{sec:seq:splitting}, we only use that a proof net with at least a vertex contains a non-leaf splitting vertex (\cref{prop:splitting_mall}).
Different splitting vertices obtained by Yeo's theorem (\cref{sec:mall_splitting}) yield various specializations of this sequentialization procedure, as in MLL (\cref{sec:seq:yeo}).
Moreover, the derivation $\someproof$ we built is in \mixretore\ normal form.

\begin{proof}[Proof of \cref{th:seq_mall}]
We reason by induction on the the number of vertices of $\G_\somepn$.
Assume first that $\G_\somepn$ contains some vertex: by \cref{prop:splitting_mall}, there is a non-leaf splitting vertex $\somevertex$ in $\G_\somepn$.
If $\somevertex$ is not terminal nor an $\ax$-vertex, then applying \cref{lem:split_seq} we obtain proof nets $\somepn_0$ and $\somepn_1$ such that $\somepn = \{\somelinking_0\cup\somelinking_1\mid\somelinking_0\in\somepn_0, \somelinking_1\in\somepn_1\}$ and $\somevertex$ is terminal and splitting in $\somepn_1$.
Observe that $\G_{\somepn_0}$ contains at least one vertex (the one below $\somevertex$ in the syntactic forest $\Delta\vdash\Gamma$) and $\G_{\somepn_1}$ contains $\somevertex$, so that both have less vertices than $\G_\somepn$.
By induction hypothesis, we obtain derivations $\someproof_0$ of $\Delta_0,A\vdash\Gamma_0$ and $\someproof_1$ and $\Delta_1\vdash A,\Gamma_1$ such that $\deseq(\someproof_0)=\somepn_0$ and $\deseq(\someproof_1)=\somepn_1$.
Thanks to \cref{lem:deseqsub_mall}, we have $\somepn=\deseq(\someproof)$ with $\someproof$ the substitution of $\someproof_1$ for hypothesis $A$ in $\someproof_0$.

If $\somevertex$ is a $\ax$-vertex, then using \cref{lem:tensor_seq}, one gets proof nets $\somepn_0$ and $\somepn_1$ such that $\somepn = \{\{\{X\orth,X\}\}\cup\somelinking_0\cup\somelinking_1\mid\somelinking_0\in\somepn_0, \somelinking_1\in\somepn_1\}$.
By induction hypothesis, one gets derivations $\someproof_0$ of $\Delta_0,X\orth \vdash \Gamma_0$ and $\someproof_1$ of $\Delta_1,X \vdash \Gamma_1$ such that $\deseq(\someproof_0)=\somepn_0$ and $\deseq(\someproof_1)=\somepn_1$.
By applying \cref{lem:deseqsub_mall} twice, we have $\somepn=\deseq(\someproof)$ with $\someproof$ the substitution of $\someproof_0$ for hypothesis $X\orth$ in the substitution of $\someproof_1$ for hypothesis $X$ in an $(ax)$ rule on $\vdash X\orth, X$.

We now suppose that $\somevertex$ is terminal and not an $\ax$-vertex.
We distinguish cases according to its kind.

If $\somevertex$ is a $\tensor$-vertex, then $\Gamma = A_0\tensor A_1,\Gamma'$ with $A_0\tensor A_1$ the formula associated with the conclusion of $\somevertex$.
Using \cref{lem:tensor_seq}, one gets $\Delta = \Delta_0,\Delta_1$, $\Gamma' = \Gamma'_0,\Gamma'_1$ and $\somepn = \{\somelinking_0\cup\somelinking_1 \mid \somelinking_0\in\somepn_0, \somelinking_1\in\somepn_1\}$ with $\somepn_i$ a proof net on $\Delta_i\vdash A_i,\Gamma'_i$ for $i\in\{0,1\}$.
By induction hypothesis, one gets derivations $\someproof_i$ of $\Delta_i \vdash A_i,\Gamma'_i$ such that $\somepn_i=\deseq(\someproof_i)$.
We add a $(\tensor)$ rule to $\someproof_0$ and $\someproof_1$, obtaining a derivation $\someproof$ of $\Delta\vdash A_0\tensor A_1,\Gamma'$ satisfying $\somepn=\deseq(\someproof)$.

If $\somevertex$ is a $\parr$-vertex, then $\Gamma = A_0\parr A_1,\Gamma'$ with $A_0\parr A_1$ the formula associated with the conclusion of $\somevertex$.
Using \cref{lem:parr_seq}, $\somepn$ is a proof net on $\Delta\vdash A_0,A_1,\Gamma'$.
By induction hypothesis, one gets a derivation $\someproof_0$ of $\Delta\vdash A_0,A_1,\Gamma'$ such that $\somepn=\deseq(\someproof_0)$.
We add a $(\parr)$ rule to $\someproof_0$, obtaining a derivation $\someproof$ of $\Delta \vdash A_0\parr A_1,\Gamma'$ satisfying $\somepn=\deseq(\someproof)$.

If $\somevertex$ is a $\with$-vertex, then $\Gamma = A_0\with A_1,\Gamma'$ with $A_0\with A_1$ the formula associated with the conclusion of $\somevertex$.
Using \cref{lem:with_seq}, $\somepn = \somepn_0\cup\somepn_1$ with $\somepn_i$ a proof net on $\Delta \vdash A_i,\Gamma'$ for $i\in\{0,1\}$.
By induction hypothesis, one gets derivations $\someproof_i$ of $\Delta \vdash A_i,\Gamma'$ such that $\somepn_i=\deseq(\someproof_i)$.
We add a $(\with)$ rule to $\someproof_0$ and $\someproof_1$, obtaining a derivation $\someproof$ of $\Delta \vdash A_0\with A_1,\Gamma'$ satisfying $\somepn=\deseq(\someproof)$.

If $\somevertex$ is a $\oplus$-vertex, then $\Gamma = A_0\oplus A_1,\Gamma'$ with $A_0\oplus A_1$ the formula associated with the conclusion of $\somevertex$.
Using \cref{lem:oplus_seq}, there exits $i\in\{0,1\}$ such that $\somepn$ is a proof net on $\Delta \vdash A_i,\Gamma'$.
By induction hypothesis, one gets a derivation $\someproof_0$ of $\Delta\vdash A_i,\Gamma'$ such that $\somepn=\deseq(\someproof_0)$.
We add a $(\oplus_i)$ rule to $\someproof_0$, obtaining a derivation $\someproof$ of $\Delta \vdash A_0\oplus A_1,\Gamma'$ satisfying $\somepn=\deseq(\someproof)$.

Now assume that $\G_\somepn$ contains no vertex.
This implies $\Delta=\Gamma$, and $\Delta\vdash\Gamma$ has no $\with$-vertex nor leaf, hence $\somepn = \{\emptyset\}$ using \cref{item:P1}.
We have $\{\emptyset\}=\deseq(\someproof)$ with $\someproof$ made of $(\hyp)$ rules on all formulas of $\Gamma$ joined by $(\mix_2)$ rules if $\Gamma$ is not empty, and otherwise $\someproof$ being a $(\mix_0)$ rule.
\end{proof}

\section{Conclusion}

We gave a new simple proof of sequentialization for multiplicative proof nets, as a corollary of a generalization of Yeo's theorem (\cref{th:ParamLocalYeo}), by defining an appropriate coloring.
This new theorem is very modular: it can give a splitting terminal vertex, a splitting $\parr$-vertex, or a general splitting vertex.
This generalization has a simple proof, that can be reformulated as a proof of sequentialization just by defining what is a cusp in proof structures.
It also allows us to deduce theorems known to be equivalent (Theorems~\ref{th:yeo}, \ref{th:kotzig}, \ref{th:seymourandgiles}, \ref{th:grossmanandhaggkvist}, \ref{th:shoesmithandsmiley} and~\ref{th:HYeo}), again just by defining a coloring.
Thus, our simple proof can also be adapted as one of any of these results by defining what is a cusp.

Focusing on proof nets, this approach can be extended to richer systems than MLL.
Dealing with multiplicative units is straightforward, as long as we allow for $\mix$-rules and forget about DR-connectedness~\cite{mixpn}: those just amount to the introduction of premise-free vertices, without any particular treatment.
Our approach is also successful without the $\mix$-rules, in a framework with a jump edge for each $\bot$-vertex (linking it with another vertex)~\cite{pn,hughes*cat,hughesunit}.
One should only take care that cusps are exactly made by the non-jump premises of $\parr$-vertices; this can be done \eg\ by giving a new color to each jump edge.

Similarly, sequentialization in presence of exponentials -- with structural rules (weakening, contraction, dereliction for the $\wn$ modality) and promotion -- is also easy to deduce from the multiplicative case: contraction is treated as a $\parr$-vertex, and promotion boxes allow to sequentialize inductively.
Again, this works both with the $\mix$-rules or with jump edges~\cite{pn}.

Dealing with additive connectives in the spirit of the unit-free multiplicative-additive proof nets from Hughes and van Glabbeek~\cite{mallpnlong} requires more work.
We adapted our approach, yielding a proof of sequentialization in a much more involved context.
The main price to pay is establishing a further generalization of \cref{th:ParamLocalYeo} allowing some cusp-free cycles (\cref{th:MALLParamYeo}) -- whose proof also reposes on the cusp minimization lemma.
The argumentation for the additives then follows the same idea as the one for MLL: a non-splitting vertex cannot be maximal for the ordering $\orderyeo$.
As in MLL, the approach is robust enough to also enable sequentialization through splitting terminal vertices, splitting $\pw$-vertices, or general splitting vertices.


\bibliographystyle{alphaurl}
\bibliography{main}

\end{document}